\documentclass[11pt,a4paper]{article}
\usepackage{fullpage} 
\usepackage[tbtags]{amsmath}
\usepackage{amssymb}
\usepackage{amsthm}
\usepackage{mathtools}
\usepackage{xspace}  % define commands that don't eat spaces

\usepackage[T1]{fontenc}
\usepackage[tt=false]{libertine}
\usepackage[libertine,vvarbb]{newtxmath}

\allowdisplaybreaks

\usepackage{bm}  % boldmath

\usepackage{hyperref}
\usepackage{footnotebackref}  % makes footnote numbers link both ways instead of just one
\usepackage[svgnames]{xcolor}
\hypersetup{colorlinks={true},urlcolor={DarkBlue},linkcolor={DarkBlue},citecolor=[named]{DarkGreen}}

\usepackage{subfig}
\usepackage{caption}
\usepackage{microtype}

\usepackage[capitalise,nameinlink,noabbrev]{cleveref}

\usepackage[backend=biber, style=alphabetic, maxbibnames=999, maxalphanames=999, backref=true]{biblatex}
\usepackage{paralist}  % custom enumeration

\newtheorem{theorem}{Theorem}[section]
\newtheorem{corollary}[theorem]{Corollary}
\newtheorem{lemma}[theorem]{Lemma}
\theoremstyle{definition}
\newtheorem{definition}[theorem]{Definition}
\newtheorem{claim}[theorem]{Claim}
\newtheorem*{remark}{Remark}
\def \R{\mathbb R}
\def \N{\mathbb N}
\newcommand{\sset}[1]{\left\{ #1\right\}}
\newcommand{\ssets}[1]{\{ #1\}}
\newcommand{\fwh}[1]{\; \left| \; #1 \right.}
\newcommand{\card}[1]{\left| #1 \right|}

\newcommand{\norm}[1]{\left\lVert #1 \right\rVert}

\newcommand{\prob}[1]{\ensuremath{\mathrm{Prob}\left[#1\right]}}

\DeclareMathOperator*{\argmax}{argmax}
\DeclareMathOperator*{\argmin}{argmin}

\newcommand{\union}{\cup}
\newcommand{\bigunion}{\bigcup}     %Set union big
\newcommand{\map}{\longrightarrow}
\newcommand{\inters}{\cap}    %Set intersection
\newcommand{\vecc}[1]{\bm{#1}} % Vector bold-face macro
\newcommand*{\vect}{\vecc}

\DeclareMathOperator*{\expectation}{\mathbb E}
\newcommand{\expect}[2][]{\expectation_{#1}\nolimits\left[#2\right]}

\newcommand*{\costv}{\bullet}
\newcommand*{\noncostv}{\circ}
\newcommand*{\eps}{\varepsilon}
\newcommand*{\intupto}[1]{\ensuremath{[#1]}}
\newcommand*{\intuptozero}[1]{\ensuremath{\lBrack #1\rBrack}}
\newcommand*{\core}[1]{\mathsf{core}(#1)}
\newcommand*{\range}[2]{\mathsf{range}_{#1}(#2)}
\newcommand*{\diversity}[2]{\delta_{#2}(#1)}

\newcommand*{\landau}[1]{\operatorname{\mathcal{O}}\left(#1\right)}

\newcommand*{\lantheta}[1]{\operatorname{\Theta}\left(#1\right)}
\newcommand*{\symd}[2]{\ensuremath{#1 \triangle #2}} % symmetric difference
\newcommand*{\indicator}[1]{\ensuremath{\vmathbb{1}\left[#1\right]}} % indicator function
\DeclareMathOperator*{\esssup}{ess\,sup}

\def\fp/{\textup{\textsf{FP}}}
\def\p/{\textup{\textsf{P}}}
\def\np/{\textup{\textsf{NP}}}
\def\conp/{\textup{\textsf{co-NP}}}
\def\fnp/{\textup{\textsf{FNP}}}
\def\tfnp/{\textup{\textsf{TFNP}}}
\def\pls/{\textup{\textsf{PLS}}}

\newcommand{\separ}[3]{\text{(\ensuremath{#1, #2, #3})-separable}}
\newcommand{\nei}[1]{\ensuremath{N(#1)}\xspace}
\newcommand{\inc}[1]{\ensuremath{E(#1)}\xspace}

\newcommand{\problemfont}[1]{\textup{\textsc{#1}}}
\newcommand{\tsp}{\problemfont{TSP}\xspace}
\newcommand{\tspIIopt}{\problemfont{TSP/$2$-Opt}\xspace}
\newcommand{\tspkopt}{\problemfont{TSP/$k$-Opt}\xspace}
\newcommand{\kopt}{\problemfont{$k$-Opt}\xspace}
\newcommand{\IIopt}{\problemfont{$2$-Opt}\xspace}
\newcommand{\atsp}{\problemfont{ATSP}\xspace}
\newcommand{\atspkopt}{\problemfont{ATSP/$k$-Opt}\xspace}
\newcommand{\maxsat}{\problemfont{MaxSat}\xspace}
\newcommand{\lmaxsatk}{\problemfont{MaxSat/$k$-Flip}\xspace}
\newcommand{\lmaxsatI}{\problemfont{MaxSat/$1$-Flip}\xspace}
\newcommand{\kflip}{\problemfont{$k$-Flip}\xspace}
\newcommand{\Iflip}{\problemfont{$1$-Flip}\xspace}
\newcommand{\lmaxcutd}{\problemfont{LocalMaxCut-$d$}\xspace}\newcommand{\lmaxcutV}{\problemfont{LocalMaxCut-$5$}\xspace}
\newcommand{\lmaxcutIII}{\problemfont{LocalMaxCut-$3$}\xspace}
\newcommand{\lmaxcut}{\problemfont{LocalMaxCut}\xspace}
\newcommand{\lmaxkcut}{\problemfont{LocalMax-$k$-Cut}\xspace}
\newcommand{\flip}{\problemfont{Flip}\xspace}
\newcommand{\mcavar}[3]{\problemfont{MCA-$(#1, #2, #3)$}\xspace}
\newcommand{\mcapqr}{\mcavar{p}{q}{r}}
\newcommand{\wIIIdmpq}{\problemfont{W3DM-$(p, q)$}\xspace}
\newcommand{\xIIIck}{\problemfont{X3C-$k$}\xspace}
\newcommand{\sck}{\problemfont{SC-$k$}\xspace}
\newcommand{\hsk}{\problemfont{HS-$k$}\xspace}
\newcommand{\hsI}{\problemfont{HS-$1$}\xspace}
\newcommand{\netcoord}{\problemfont{NetCoordNash}\xspace}
\newcommand{\congestion}{\problemfont{PNE-Congestion}\xspace}
\newcommand{\circuitflip}{\problemfont{CircuitFlip}\xspace}

\begin{document}
\title{On the Smoothed Complexity of Combinatorial Local Search}
\author{Yiannis Giannakopoulos\thanks{University of Glasgow.
Email: \href{mailto:yiannis.giannakopoulos@glasgow.ac.uk}{\nolinkurl{yiannis.giannakopoulos@glasgow.ac.uk}}}
\and
Alexander Grosz\thanks{Technical University of Munich.
Email: \href{mailto:alexander.grosz@tum.de}{\nolinkurl{alexander.grosz@tum.de}}.
This author was supported by the Alexander von Humboldt Foundation with funds from the German Federal Ministry of Education and Research (BMBF).}
\and
Themistoklis Melissourgos\thanks{University of Essex.
Email: \href{mailto:themistoklis.melissourgos@essex.ac.uk}{\nolinkurl{themistoklis.melissourgos@essex.ac.uk}}}
}

\date{26 July, 2025}

\maketitle
\thispagestyle{empty}
%%%%% END OF TITLEPAGE %%%%%

\begin{abstract}
	We propose a unifying framework for smoothed analysis of combinatorial local
    optimization problems, and show how a diverse selection of problems within
    the complexity class \pls/ can be cast within this model. This abstraction
    allows us to identify key structural properties, and corresponding
    parameters, that determine the smoothed running time of local search
    dynamics. We formalize this via a black-box tool that provides concrete
    bounds on the expected maximum number of steps needed until local search
    reaches an exact local optimum. This bound is particularly strong, in the
    sense that it holds for any starting feasible solution, any choice of
    pivoting rule, and does not rely on the choice of specific noise
    distributions that are applied on the input, but it is parameterized by just
    a global upper bound $\phi$ on the probability density. The power of this
    tool can be demonstrated by instantiating it for various \pls/-hard problems
    of interest to derive efficient smoothed running times (as a function of
    $\phi$ and the input size).
    
    Most notably, we focus on the important local optimization problem of
    finding pure Nash equilibria in Congestion Games, that has not been studied
    before from a smoothed analysis perspective. Specifically, we propose novel
    smoothed analysis models for general and Network Congestion Games, under
    various representations, including explicit, step-function, and polynomial
    resource latencies. We study \pls/-hard instances of these problems and show
    that their standard local search algorithms run in polynomial smoothed time. 
    
    Finally, we present further applications of our framework to a wide range of
    additional combinatorial problems, including local Max-Cut in weighted
    graphs, the Travelling Salesman problem (TSP) under the $k$-opt local
    heuristic, and finding pure equilibria in Network Coordination Games. 
\end{abstract}

\newpage
\thispagestyle{empty}
\tableofcontents
\newpage

\setcounter{page}{1}

\section{Introduction}

Local search heuristics are some of the most prominent, and widely used in
practice, algorithms for solving computationally hard, combinatorial
problems~\parencite{Aarts1997,Michiels2007}. Their appeal stems not only from
their simplicity and theoretical elegance, but also from the fact that, for many
applications, they seem to perform remarkably well both in terms of their
running times and the quality of the solutions they produce.

Theoreticians have long tried to rigorously study the performance of local
search, but also explain its prevalence in practice. \textcite{Johnson:1988aa}
introduced the class \pls/ to capture the complexity of local optimization
problems; since then, many important such problems have been shown to be
\pls/-complete, implying that they most likely cannot be solved (exactly) in
polynomial time. 
This hardness applies not only to local search algorithms, but
to arbitrary local optimization methods. For local search, in particular,
\pls/-hardness (under tight reductions) implies the provable existence of instances leading to
exponentially slow convergence \parencite{Yannakakis1997}.
Examples include the Travelling Salesman problem under the \kopt
heuristic~\parencite{Krentel:structurelocalopt,CKT:kOpt} (\tspkopt in the
following), Local Maximum Cut on weighted graphs~\parencite{Schaffer91}
(\lmaxcut), and the problem of finding pure Nash equilibria in Congestion
Games~\parencite{FPT04} (\congestion).

On the other hand, \textcite{Orlin2004} designed a local-search-based polynomial-time scheme for
efficiently computing \emph{approximately} locally-optimal solutions for general
combinatorial optimization problems (with linear objectives). Although this
result provides concrete justification for the practical tractability of local
optimization, there are still many important aspects that call for further
investigation. First, if one requires exponential accuracy, their FPTAS still cannot
provide polynomial running times. Secondly, approximate solutions do not always
make sense for all local optimization problems; there are problems in \pls/ that
are inherently exact, and they are not derived by simply considering the local
version of some ``master'' global optimization problem. A notable example is
\congestion~\parencite{Roughgarden16}. Finally, we would like to be able to
argue about the more general family of ``vanilla'' local search, and to ideally
get positive results that do not depend on additional details and specific
choices of pivoting rules. 
Addressing these points is a key objective of the present paper.

Smoothed analysis was introduced by~\textcite{Spielman2004} as a more realistic
alternative to traditional worst-case analysis, where now the adversarially
selected input is submitted to small random shocks of its numerical parameters,
\emph{before} being presented to the algorithm; the running time is then measured
\emph{in expectation} with respect to these perturbations. Under this model,
\citeauthor{Spielman2004} were able to show that Simplex, the archetypical
method for solving linear programs, is guaranteed to terminate in polynomial
time (under a shadow pivoting rule) -- as opposed to its exponential complexity under worst-case analysis. This
remarkable result established smoothed analysis as a canonical framework for
studying the performance of algorithms beyond the worst-case
(see~\parencite{Roughgarden2021} for an overview of this field). 

In particular, smoothed analysis has been applied successfully to important
local search algorithms, providing thus a theoretical basis for the
justification of their good performance in practice; these include, e.g., the
\kopt heuristic for the \tsp and the \flip heuristic
for \lmaxcut. (A more detailed exposition of related work on this front is
deferred to the following sections of this paper, where each of our local
optimization problems of interest is explicitly studied; see, namely,
\cref{sec:congestion-games-smoothed-models,sec:sin-step-apps}.)
A common characteristic of this prior work, though, is that running-time
analysis is usually tailored specifically to the local optimization problem at
hand. This naturally creates the need for technically heavy derivations, from
which it is generally not clear how to pin down the core properties of the
underlying local-search structure that allow for the efficient smoothed
complexity. Furthermore, as a result, it is often not easy to immediately
generalize these results to capture interesting extensions, e.g., argue about
the \emph{asymmetric} version of the \tsp, or go from \tspIIopt to \tspkopt.
Finally, this lack of sufficient abstraction is one of the reasons that
smoothed analysis has not been yet considered at all for prominent PLS-hard
problems, including, e.g., \congestion.
Dealing with this set of challenges is another driving force behind our paper.

\subsection{Our Results and Outline}

We start by proposing an abstract model for smoothed analysis of combinatorial
local optimization (CLO) in~\cref{sec:model}. Our family of CLO problems
includes problems in \pls/ that have an arbitrary combinatorial neighbourhood
structure and linear objective functions; essentially, our model generalizes
that of~\parencite{Orlin2004} beyond binary configurations.
In~\cref{sec:smoothPLS}, we add our smoothness layer that introduces
probabilistic noise (independently) to the cost parameters of the CLO problem.
No further assumptions are made on the distributions of the perturbed costs;
their densities are only parameterized by a global upper bound of $\phi$. This is
a standard model (employed, e.g., in \parencite{Beier2006a,Roeglin2007a,Englert_2016,Etscheid17}) that makes positive smoothed-analysis
results even stronger, and extends the seminal model of Gaussian
perturbations from~\parencite{Spielman2004}.  

\Cref{sec:sin-step} contains our key technical result for deriving upper bounds
on the expected (under smoothness) number of local-search steps, until an
\emph{exact} local optimum is reached. Our black-box tool
(\cref{thm:single-step-bound}) can be readily applied to an abstract CLO
problem, once its underlying neighbourhood structure is appropriately captured;
this is formalized through the notion of \emph{separability}
(see~\cref{def:separable-instances}), quantified by a collection of three
parameters that are critical for the upper bound given
by~\cref{thm:single-step-bound}. We note here that like many problem-specific state-of-the-art results, our bounds are robust against
the specific choice of a starting point for the local search dynamics, as well as
the pivoting rule utilized at every step to transition to an improving neighbour. In
other words, our main black-box tool establishes bounds for the entire
\emph{family} of local search heuristics of a local optimization problem. At a
technical level, our proof works by lower-bounding the probability that
\emph{all} steps of the local-search sequence improve the
objective sufficiently, thereby applying a standard approach to a general, application independent abstract problem.

In~\cref{sec:congestion-games-smoothed-models} we demonstrate the applicability
of our general framework by instantiating it for \congestion, a prominent
\pls/-complete problem which has not been studied before from a smoothed analysis
perspective. First, we propose smoothed analysis formulations for various
representations of interest for the problem, namely explicit, step-function, or
polynomial resource latencies.\footnote{Very recently, and after our paper had
appeared online, \textcite{g2024} built upon our model to establish the
existence of a fully polynomial-time approximation scheme (FPTAS) for computing
\emph{approximate} pure Nash equilibria in smoothed congestion games.} Next,
after formally establishing how \congestion is indeed a CLO problem
in~\cref{sec:NE_as_CLO}, we identify a natural parameterization of the problem
that we call \emph{$B$-restrained games} (\cref{sec:congestion-positive}), where
$B$ is an upper bound on the number of resources that can be changed during a
single-player deviation (see~\cref{def:restrained-games}). Interestingly enough,
the case of constant $B$ is still rich enough to encode the full \pls/-hardness
of \congestion (\cref{sec:congestion-hardness}), while at the same time it can
be shown (see the proof
of~\cref{th:smoothed-p-bounded-interaction-congestion-games}) to be
appropriately separable in order to immediately provide polynomial smoothed
running time bounds via our black-box tool developed in~\cref{sec:sin-step}.
Similarly, in~\cref{sec:network-games-compact} we also study a special class of
\emph{network} congestion games, which we term \emph{$(A,B)$-compact}
(\cref{def:compact-network-game}); we establish polynomial smoothed complexity
for various families of instances, including the one where $A$ is polynomial and
$B$ is constant (\cref{th:compact-netc-smoothed-poly}), which we pair with a
complementing \pls/-hardness proof (\cref{sec:pl-hardness-network-compact}).

In the remainder of our paper, we apply our high-level framework to various
other local optimization problems of interest, by first formally establishing
that they can be viewed as CLO problems and then identifying the proper
separability structures that can be plugged into our black-box tool to provide
good smoothed bounds for local search. For some of these problems we rederive
existing bounds from the literature, in a much simpler way; for others we
improve or extend existing results; and for others, we are the first to propose
and perform smoothed analysis on them. An important underlying feature of all
our proofs is the unifying and straightforward way in which our bounds are
derived across the entire spectrum of local search problems studied in this
paper.

In particular, in~\cref{sec:TSP-k-opt} we study the smoothed running time of \tspkopt. We improve on this problem in two directions: we consider the most general version, namely the \emph{asymmetric} \tsp~\parencite{STV:ATSP_constantapprox}, and without any metric constraints, while we study the \kopt heuristic for all $k \geq 2$. In \cref{thm: smoothed-tsp}, we show that the \kopt heuristic has a polynomial smoothed running time for every fixed $k$. To the best of our knowledge, our work is the first to extend the results of  \cite{Englert_2016} to the case beyond $k=2$.

In \cref{sec:max-sat} we study \lmaxsatk \cite{YagiuraIbaraki:kFlip_MaxSat, Szeider:kFlipSAT} in which we search for a truth assignment whose flipping of at most $k$ variables cannot improve its sum of clause-weights. We show that, for any fixed $k$, when each variable appears constantly many times in the clauses, the smoothed running time is polynomial, while the problem remains \pls/-hard \cite{Krentel:structurelocalopt}.
To the best of our knowledge, prior to this work, only the case of $k=1$ had been studied under smoothness in \cite{Brandts-Longtin:smoothed_maxksat}, whose analysis is limited to the case of complete Boolean formulas; that is, given a fixed clause size, the formula must contain \emph{all} possible clauses of that size.

Then, in \cref{sec:localmaxcut} we apply our black-box tool to \lmaxcut, and show that when the input graph has logarithmic degree, the smoothed running time is polynomial. 
Our $\landau{2^\Delta \card{V} \card{E}^2}$ running time, where $\Delta$ is the maximum degree of the input graph $G=(V, E)$, improves the bound given in~\parencite{ElsasserT11}, while at the same time extends it beyond Gaussian perturbations. 
Furthermore, for the generalization of that problem, namely \lmaxkcut, for $k > 2$, we provide a bound that does \emph{not} depend on $k$.
This results in a smoothed polynomial running time for graphs with logarithmic degree. We note that~\textcite{Bibak21} had shown smoothed super-polynomial running time for graphs with arbitrary degree.

Network Coordination Games are another problem on which we apply our unified smoothed analysis framework in \cref{sec:netcoord}.
Our results establish polynomial smoothed complexity for graphs with either (a) constant degree and arbitrary number of strategies per player, or (b) logarithmic degree and a constant number of strategies per player.
Prior work~\cite{Boodaghians20} provides \emph{quasipolynomial} bounds under the assumption of a constant number of strategies.

A selection of weighted set problems has been put into the context of \pls/ by \cite{Dumrauf10} using canonical $k$-change neighbourhoods, where already for small constant values of $k$, \pls/-hardness was established.
In \cref{sec:weightedset}, we define a smoothed problem formulation and prove polynomial smoothed complexity for Weighted 3D-Matching, Exact Cover by 3-Sets and Set-Cover.
For the Hitting Set problem, we identify a property that allows us to parameterize the smoothed running time, and prove accompanying \pls/-hardness for a smoothly efficient subclass of instances.
Finally, we show that the related problem Maximum Constraint Assignment has polynomial smoothed complexity when its parameters are appropriately bounded.

\section{Smoothed Combinatorial Local Optimization}
\label{sec:model}

In this section we formalize our model and fix the necessary notation.

We denote by $\N$ and $\R$, the natural and
real numbers, respectively. We also denote $\N^*\coloneqq \N\setminus\ssets{0}$
and $\intupto{k}\coloneqq\ssets{1,2,\dots,k}$, $\intuptozero{k}\coloneqq \ssets{0}\union[k]$ for
$k\in\N$.
We will use boldface notation for vectors,
$\vecc{s}=(s_1,s_2,\dots,s_n)\in\R^n$. For an index $i\in\intupto{n}$, we use
$\vecc{s}_{-i}$ to denote the $(n-1)$-dimensional vector that results from an
$n$-dimensional vector $\vecc s$ if we remove its $i$-th component; in that way,
we can express $\vecc{s}$ as $(s_i,\vecc{s}_{-i})$ in order to easily denote deviations in the $i$-th component.
More generally, for a set of indices $I\subseteq\intupto{n}$, $\vecc{s}_{I}$
denotes the $\card{I}$-dimensional vector that we get if we keep only the
components of $\vecc{s}$ whose indices are in $I$.
For finite sets $I$, $\mathcal I$, we say that $\mathcal I$ is a
\emph{cover} of $I$ if $I\subseteq\bigunion_{I'\in\mathcal{I}} I'$. 
All logarithms appearing in our paper are of base $2$.

\paragraph{Combinatorial local optimization (CLO)}
An instance of a \emph{combinatorial local optimization (CLO)} problem is composed
of:
\begin{itemize}
	\item A set of feasible \emph{configurations} $\vecc S\subseteq \intuptozero{M}^{\nu} \times
	\ssets{0,1}^{\bar{\nu}}$, where $M,\nu\in\N^*$,
	$\bar\nu\in\N$. 
	A configuration $\vect{s}$ can be expressed as $\vect{s} = (\vecc{s}^\costv,
\vecc{s}^\noncostv)$, where $ \vecc{s}^\costv = (s_1,\dots,s_\nu)\in \intuptozero{M}^\nu $
is called its \emph{cost part} and $
\vecc{s}^\noncostv=(s_{\nu+1},\dots,s_{\nu+\bar\nu})\in \ssets{0,1}^{\bar{\nu}}$ its
\emph{non-cost part}.

	\item A vector of \emph{costs} $\vecc c=(c_1,c_2,\dots,c_\nu)\in [-1,1]^\nu$. We call the $c_i$ \emph{cost coefficients}.%
    \footnote{Note that the restriction of costs to the range of $[-1, 1]$ has been made to facilitate the translation to smoothed CLO problems below (see~\cref{sec:smoothPLS}).
    In fact, due to the linearity of our objective, scaling $\vecc c$ by any positive constant does not change the structure of the problem, and we will implicitly use this fact when formulating other problems as CLO problems.}

	\item A \emph{neighbourhood} function $N:\vecc{S}\map 2^{\vecc{S}}$. For
	$\vecc s\in \vecc{S}$, any configuration $\vecc{s}'\in N(\vecc s)$ will be
	called a \emph{neighbour} of $\vecc s$.
\end{itemize} 
For the special case of $M=1$ we will call our problem \emph{binary}.

The \emph{cost} of a configuration $\vecc s$ (with
respect to a fixed cost vector $\vecc c$) is given by
		$$C(\vecc s) \coloneqq \vecc c\cdot \vecc s^\costv=\sum_{i=1}^\nu
	c_is_i.$$ 
A configuration $\vecc s$ is said to be a \emph{local optimum
	(minimum\footnote{Here we choose to write the local optimality condition
	with respect to minimization problems. All the theory developed in this
	paper applies immediately to local maximization problems as well: just flip
	the inequality in~\eqref{eq:local-opt-min-def} or, equivalently, negate the
	cost vector $\vecc c$.})} if there are no neighbours with better costs;
	formally, 
	\begin{equation} \label{eq:local-opt-min-def}
		C(\vecc s) \leq C(\vecc s')\qquad \text{for all}\;\; \vecc s'\in N(\vecc s).
	\end{equation}

For every CLO problem we can define its \emph{configuration} (or
\emph{neighbourhood}) \emph{graph}, made up by all edges pointing from
configurations to their neighbours. Formally, it is the directed graph $G=(V,E)$
with node set $V=\vecc S$ and edges $E=\sset{(\vecc s,\vecc s')\fwh{\vecc
s\in\vecc S,\; \vecc s'\in N(\vecc s)}}$. Observe that the configuration graph
does not depend on the costs $\vecc c$, but only on the combinatorial structure
of the problem. Taking the cost vector $\vecc c$ into consideration, we can now
restrict the configuration graph edges to those that correspond to locally
improving moves, i.e.\ take $E'=\sset{(\vecc s,\vecc s')\in E\fwh{C(\vecc
s')<C(\vecc s)}}$. The resulting acyclic subgraph $G'=(V,E')$ is called the \emph{transition
graph}.

The transition graph provides an elegant and concise interpretation of local
optimization: finding a local optimum of a combinatorial (local optimization)
problem translates to finding a sink of its transition graph. This task is
exactly the object of interest of our paper: we study (the running time of)
algorithms that find local optima of such problems.

\paragraph{Complexity} 
If the configuration set and the neighbourhoods $N(\cdot )$ were given
explicitly in the input of a local optimization algorithm, then finding a
locally optimal solution would have been a computationally trivial task: we
could afford (in polynomial time) to exhaustively go over all nodes of the
transition graph, until we find a sink. Notice that the existence of a sink
is guaranteed, by the fact that $\vecc S$ is finite.
However, most problems of interest (including all the problems that we study here) have exponential-size configuration
graphs, with respect to the dimension $\nu + \bar \nu$ of the problem.
Therefore, $\vecc S$ and $N$ are usually instead described \emph{implicitly},
via some succinct representation (that is polynomial on $\nu$, $\bar \nu$ and $M$). Then, the
computational complexity of our algorithms is naturally measured as a function
of the critical parameters $\nu$, $\bar \nu$ and $M$ (and the bit representation of the costs
$\vecc c$).

\paragraph{Polynomial local optimization (\texorpdfstring{\pls/}{PLS})}
In this paper we want to study
problems contained in \pls/, the ``canonical'' complexity class for local
optimization problems, introduced by \textcite{Johnson:1988aa}. Therefore,
without further mention, from now on we will assume that our CLO problems
further satisfy the following properties:
\begin{itemize}
	\item An \emph{initial} configuration $\vecc{s}_0\in \vecc{S}$ can be
	computed in polynomial time (on the size of the input).
	\item There exists a polynomial-time algorithm that, given as input any
	configuration $\vecc s$, decides whether $\vecc s$ is a local optimum and,
	if not, returns an improving neighbour $\vecc s'\in N(\vecc s)$ with
	$C(\vecc s')< C(\vecc s)$. Such an algorithm is called a \emph{pivoting rule}.
\end{itemize}

It is important to clarify that the polynomial time algorithms of the bullets above are formally
\emph{not} part of the description of the CLO
problem, neither are they required to be specified in the input of a local
optimization algorithm. The \emph{existence} of these algorithms is merely a
requirement for the \emph{membership} of the problem in the class \pls/
(similarly to the existence of short certificates and an efficient verifier for
membership in \np/).
As a result, \pls/ can be interpreted as the class of problems that correspond
to looking for a sink in an (implicitly given, possibly exponentially large)
directed acyclic graph\footnote{This is the transition graph mentioned earlier.}, where
at least one node can be found efficiently and, for any given node, at least one
neighbour (if such exists) can be found efficiently~\parencite{Daskalakis2011,Fearnley2020}.

This interpretation naturally gives rise to the \emph{standard local search}
heuristic: start from an arbitrary initial configuration, and at every iteration
perform an (arbitrary\footnote{Thus, standard local search is essentially a
\emph{family} of algorithms; different pivoting rules can give rise to different
local search algorithms.}) locally improving move until no such move exists any more.
For a fixed pivoting rule, this corresponds to traversing a single path of
the transition graph. Due to the definition of \pls/ membership above, the
running time of such a process is thus determined by the \emph{number} of these
local search iterations. In the worst case, this amounts to bounding the length
of the \emph{longest path} in the transition graph.

It is important to emphasize that, regardless of being a very natural heuristic,
standard local search is definitely not the only method for finding local
optima: as a matter of fact, we know that there exist local optimization
problems that are efficiently solvable via more involved, ``centralized''
methods (e.g., by using linear programming), but for which standard local search
would provably require exponential time (see, e.g., \parencite{Ackermann2008}): more
precisely, there exist nodes in their transition graphs from which \emph{all}
paths have exponential length~\parencite{Schaffer91}.

To assert the intractability of a problem, the common argumentation is via complexity theory techniques that prove the conditional inexistence of polynomial time algorithms. If a problem is \pls/-hard, then unless $\pls/=\fp/$, there is no algorithm that solves it in polynomial time under traditional worst-case analysis.
However, since we are studying the performance of local search heuristics, we are in fact interested in specifically \emph{their} running times.
\textcite{Schaffer91} introduced the notion of a stronger reduction among problems in \pls/, named \emph{tight \pls/-reduction}, which preserves key structural properties of the initial instance to the target-problem's instance.
An important implication of an initial-problem $P$ having a tight \pls/-reduction to target-problem $Q$ is the following: 
for any instance $I$ of $P$, a path in the transition graph of the corresponding reduced instance $J$ of $Q$ induces a path of $I$ of length no larger than its own.
In particular, if for an instance $I$ there exists a starting configuration from which all paths to solutions are exponentially long, this also applies to the reduced instance $J$.
We show that all the problems we study have this property, either by referring to explicit bad instances for the problem at hand, or through (chains of) tight reductions, which we construct if they haven't been established before.

For a more thorough treatment of the complexity of local optimization and the
class \pls/, the interested reader is referred to the excellent monograph
of~\textcite{Yannakakis1997}.

\subsection{Smoothed Combinatorial Local Search}
\label{sec:smoothPLS}
Under traditional, worst-case algorithmic analysis, the running time of an
algorithm for a CLO problem would be evaluated against an adversarially selected
cost vector $\vecc c=(c_1,\dots,c_\nu)$. Instead, our goal here is to propose a
systematic \emph{smoothed analysis}~\parencite{Spielman:2009aa,Roglin:50046,Roughgarden2021} framework for local optimization. Therefore,
for the remainder of this paper we assume that $\vecc c$ is not fixed, but drawn
randomly from a product distribution. More specifically, each cost coefficient
$c_i$ is drawn independently from a continuous probability distribution with
density $f_i:[-1,1]\map [0,\phi]$. We note that $\phi\geq 1/2$ is necessary for a probability distribution.
These distributions can be adversarially selected, but their realizations $c_i$ are
not; the running time is then computed \emph{in expectation} with respect to the
random cost vector $\vecc c$. An efficient algorithm runs in time polynomial in
the combinatorial-structure parameters $\nu$, $\bar{\nu}$ and $M$ of the problem, and in the
smoothness parameter $\phi$.
We will sometimes refer to this model as \emph{smoothed CLO}, if we
want to give particular emphasis to the fact that we are performing a smoothed
running time analysis (as opposed to worst-case analysis).

In that sense, smoothed analysis can be seen as interpolating between two
extremes: an average-case analysis setting where all $c_i$ are drawn independently from the uniform distribution over $[-1,1]$, derived for $\phi=1/2$; and
traditional worst-case analysis that can be derived in the limit, via
$\phi\to\infty$, as distributions $f_i$ approximate adversarial,
single-point-mass instances.

In this paper we focus on smoothed analysis for standard local search, and so
our quantity of interest will be the expected number of improving
moves (for any set of adversarially given input distributions) until a local optimum is reached; that is, the expected length of the
longest path in the transition graph. This allows us to also avoid some delicate
representation issues that are typical of smoothed analysis, and which have to
do with how the realizations $c_i$ of \emph{continuous} distributions (which therefore produce irrational numbers almost surely) can be
handled as inputs to a Turing-machine-based computational model. For a careful
discussion about this topic we point to the papers of~\textcite{Beier2006a}
or~\textcite{Roeglin2007a}. In a nutshell, for our purposes it is safe to think
of the polynomial-time improving-local-move oracle (from the \pls/ definition)
as having access to real-number arithmetic.

\section{Smoothed Analysis of Local Search}\label{sec:sin-step}

In this section we present our first main result, which is a black-box tool for
upper-bounding the number of improving moves of standard local search, under
smoothed analysis.
To achieve this, we first highlight an appropriate underlying structure
of CLO problems and identify key parameters that
characterize it (see~\cref{def:separable-instances}). Then, our bounds for
standard local search are directly expressed as a function of these parameters
(see~\cref{thm:single-step-bound}). As we will demonstrate in
\cref{sec:sin-step-apps,sec:network-games-compact,sec:congestion-positive}, for various CLO problems of
interest these parameters are well-behaved enough to result in
polynomial smoothed running times.

We now introduce some terminology and notation that will be necessary for
stating our main result in~\cref{thm:single-step-bound}.
Fix an instance of a CLO problem with cost coordinates $\intupto{\nu}$ and
configurations $\vecc S$ (see~\cref{sec:model}), and let $G=(\vecc{S},E)$ be its
neighbourhood graph.
A \emph{covering} $(\mathcal{E},\mathcal{I})$ of this instance consists of a cover $\mathcal{E}$ of the edges of its neighbourhood graph $G$, and a cover $\mathcal{I}$ of its cost coordinates. That is, $\mathcal E\subseteq 2^E$ and $\mathcal{I}\subseteq 2^{\intupto{\nu}}$ such that 
$E= \bigunion_{E'\in\mathcal{E}} E'$ and $\intupto{\nu}= \bigunion_{I\in\mathcal{I}} I$.
We call $\mathcal{E}$ the \emph{transition cover} and $\mathcal{I}$ the \emph{coordinate cover}, and they contain \emph{transition clusters} and \emph{coordinate clusters} respectively.
Recall that, under our previous discussion (see~\cref{sec:model}), $\mathcal E$
can be simply interpreted as covering all potential configuration transitions
that can be made by standard local search, clustered appropriately in different
groups $T\in\mathcal{E}$.
For an arbitrary set of such transitions $T\subseteq{E}$, we also define its \emph{core} to be the set of coordinates affected by any of the transitions:
$$
\core{T}\coloneqq 
\sset{i\in\intupto{\nu}\fwh{\vecc{s}^\costv_i \neq \vecc{s'}^\costv_i \;\;\text{for some}\; (\vecc{s}, \vecc{s'})\in T}},
$$
and its \emph{diversity} with respect to a given set of coordinates $I\subseteq\intupto{\nu}$ to be the number of different configuration changes when projected to $I$:
$$
\diversity{T}{I} \coloneqq \card{\range{I}{T}},
\quad\text{where}\;\;
\range{I}{T}\coloneqq
\sset{\vecc{s}^\costv_I-\vecc{s'}^\costv_I \fwh{(\vecc{s}, \vecc{s'})\in T}}.
$$

\begin{definition}[Separable instances]
	\label{def:separable-instances}
	An instance of a CLO problem will be called \emph{\separ{\lambda}{\beta}{\mu}} if it has a covering $(\mathcal{E},\mathcal{I})$ with $\card{\mathcal{E}}\leq \lambda$, such that every transition cluster $T\in \mathcal{E}$ has:
	\begin{enumerate}[(a)]
		\item\label{item:separable_prop_1} a core that can be covered by using $\beta$ many coordinate clusters from the cover $\mathcal I$; formally, there exists an $\mathcal{I}_T\subseteq \mathcal{I}$ with $\card{\mathcal{I}_T}\leq \beta$ such that $\core{T}\subseteq\bigunion_{I\in \mathcal{I}_T} I$, and 
		\item\label{item:separable_prop_2} diversity at most $\mu$ with respect to all coordinate clusters; formally, $\max_{I\in \mathcal{I}} \diversity{T}{I}\leq \mu$.
	\end{enumerate}
\end{definition}

\begin{theorem}\label{thm:single-step-bound} 
	For every \separ{\lambda}{\beta}{\mu}
	smoothed combinatorial local optimization instance, standard local search
	terminates after at most
        $$ 3\cdot \mu^{\beta} \lambda \cdot \nu^2 M\log(M+1) \cdot \phi $$
    many steps (in expectation).
\end{theorem}

For the proof of~\cref{thm:single-step-bound} we will need the following technical lemmas. Their proofs can be found in~\cref{sec:proof-lemmas-appendix}. We would like to highlight the fact that the structural quantities $(\lambda, \beta, \mu)$ appear as the expression $\mu^\beta \lambda$ in the statement of the theorem and are the deciding quantities for the running time complexity of the problem, as evidenced in the applications further below. As the choice of the covering is not unique (and often even allows for more than just one \emph{natural} choice), this exact expression helps to explain the quantitative interaction between the properties of the covering.
Finally, we believe it is interesting to note here that the \emph{specific} proof of~\cref{thm:single-step-bound} that we provide below, imposes a \emph{tight} dependence on the $\mu^{\beta}$ parameter on the theorem's bound, due to the union-bound technique we deploy. In other words, in order to improve the dependence on $\beta$ in the exponent (even asymptotically), a different proof technique would be required.

\begin{lemma}
\label{lemma:diversity-cover}
Let $I\subseteq\intupto{\nu}$ be a set of cost coordinates and $\mathcal{J}\subseteq 2^{\intupto{\nu}}$ be a cover of $I$, i.e.\ $I\subseteq\bigunion_{J\in\mathcal{J}}J$. Then, for every set of transitions $T\subseteq E$,
$$
\diversity{T}{I} \leq \prod_{J\in\mathcal{J}} \diversity{T}{J}.
$$
\end{lemma}

\begin{lemma} 
	\label{lemma:ER-rank1-maxdensity}
Let $\phi>0$ and let $\vecc{X}=(X_1,X_2,\dots,X_m)$ be a random real
vector, where each component $X_i$ is drawn independently from a continuous
distribution with density $f_i:\R\map[0,\phi]$. Then, for every nonzero
vector $\vecc{\xi}\in\R^m$ and every $\varepsilon\geq 0$,
\begin{equation}
	\label{eq:ER-rank1-maxdensity-strong}
\prob{0\leq \vecc{\xi}\cdot \vecc X \leq \varepsilon} \leq \min\left(\frac{1}{\norm{\vecc{\xi}}_\infty},\frac{\sqrt{2}}{\norm{\vecc{\xi}}_2}\right) \cdot \varepsilon\phi,
\end{equation}
where $\norm{\cdot}_2$ and $\norm{\cdot}_\infty$ denote the Euclidean and maximum norms, respectively. 
For the special\footnote{In this paper we will be actually making use of
\cref{lemma:ER-rank1-maxdensity} only via its weaker
bound~\eqref{eq:ER-rank1-maxdensity-weaker}, rather than the stronger
form~\eqref{eq:ER-rank1-maxdensity-strong}. This is sufficient for our purposes
because, as it turns out, this has an asymptotically negligible effect on our
bounds. However, we still choose to state (and prove)
\cref{lemma:ER-rank1-maxdensity} in its full generality, since we expect it to
be of potential independent interest for future extensions, especially if one
considers more involved structures, or non-integral configurations.} case where $\vecc{\xi}$ is a (nonzero) integer vector we get
\begin{equation}
	\label{eq:ER-rank1-maxdensity-weaker}
\prob{0\leq \vecc{\xi}\cdot \vecc X \leq \varepsilon} \leq \varepsilon\phi.
\end{equation}
\end{lemma}

\begin{proof}[Proof of~\cref{thm:single-step-bound}]
Fix a $\separ{\lambda}{\beta}{\mu}$ smoothed CLO instance, with neighbourhood
graph $G=(\vecc{S},E)$, and let $(\mathcal{E},\mathcal{I})$ be a covering
satisfying the conditions of~\cref{def:separable-instances}.  

We introduce the following notation:
    $$ 
	\Delta(\vecc s, \vecc{s'}) \coloneqq 
	\begin{cases}
        C(\vecc{s})-C(\vecc{s}'), & \text{if this is positive}, \\
        \infty, & \text{otherwise},
    \end{cases}
	$$
for all $(\vecc{s},\vecc{s}')\in E$, and 
$$\Delta \coloneqq \min_{(\vecc{s},\vecc{s}')\in E} \Delta(\vecc s,
\vecc{s'}).$$ Notice that these are random variables, depending on the
realizations of the cost vector $\vecc c$.

Our first goal is to give a bound on the probability that there exists a local move
that improves the cost only by (at most) $\eps>0$, as a function of this
improvement bound $\varepsilon$. If this quantity is sufficiently small, then
with high probability, standard local search will achieve improvements more than
$\eps$ at \emph{every} step, thus resulting in faster convergence. To upper
bound this probability  $\prob{\Delta\leq \eps}$, we first use a union bound
over the cover $\mathcal{E}$ of all transitions $(\vecc{s},\vecc{s}')\in E$ to
get

\begin{equation}
	\label{eq:single-step-bound-helper-1}
     \prob{\Delta \leq \eps}
        = \prob{\bigunion_{(\vecc{s}, \vecc{s'}) \in E} \left[\Delta(\vecc{s}, \vecc{s'})\leq \eps\right]}
        \leq \sum_{T\in \mathcal{E}} \prob{\bigunion_{(\vecc{s}, \vecc{s'})\in T} \left[ \Delta(\vecc{s},\vecc{s'})\leq \eps \right]}.
\end{equation}

Next, for a fixed transition cluster $T\in\mathcal{E}$ we can express the inner union of events in~\eqref{eq:single-step-bound-helper-1} as
\begin{align}
\bigunion_{(\vecc{s}, \vecc{s'})\in T}\left[ \Delta(\vecc{s},\vecc{s'})\leq \eps \right] 
	&= \bigunion_{(\vecc{s}, \vecc{s'})\in T}\left[ 0<\vecc{c}_{\core{T}}\cdot\left(\vecc{s}_{\core{T}}-\vecc{s'}_{\core{T}}\right)\leq \eps \right] \notag \\
	&= \bigunion_{\vecc{x}\in\range{\core{T}}{T}} \left[0<\vecc{c}_{\core{T}}\cdot \vecc{x} \leq \varepsilon \right]. \label{eq:single-step-bound-helper-2}
\end{align}
By the separability assumption of our CLO instance (see \cref{def:separable-instances}), for any cluster $T\in\mathcal{E}$ there exists a subset $\mathcal{I}_T$ of $\mathcal{I}$ with $|\mathcal{I}_T| \leq \beta$ that covers $\core{T}$ such that, additionally, $\max_{J \in \mathcal{I}_T} \diversity{T}{J}\leq \max_{I \in \mathcal{I}} \diversity{T}{I}\leq \mu$. So, from \cref{lemma:diversity-cover} we can deduce that
$$
\card{\range{\core{T}}{T}}=\diversity{T}{\core{T}} \leq \prod_{J \in \mathcal{I}_T} \diversity{T}{J} \leq \mu^\beta.
$$
Furthermore, observe that each $\vecc{x}\in \range{\core{T}}{T}$ is a nonzero integral vector. Thus, applying~\eqref{eq:ER-rank1-maxdensity-weaker} of \cref{lemma:ER-rank1-maxdensity} we can derive that, for a fixed $\vecc{x}\in \range{\core{T}}{T}$, we have
$$
\prob{0<\vecc{c}_{\core{T}}\cdot \vecc{x} \leq \varepsilon} \leq \phi\eps.
$$
Therefore, using again a union bound, this time on event~\eqref{eq:single-step-bound-helper-2}, we can see that
$$
\prob{\bigunion_{(\vecc{s}, \vecc{s'})\in T}\left[ \Delta(\vecc{s},\vecc{s'})\leq \eps \right]}
\leq \sum_{\vecc{x}\in\range{\core{T}}{T}}\prob{0<\vecc{c}_{\core{T}}\cdot \vecc{x} \leq \varepsilon} \leq \mu^{\beta}\phi\eps.
$$
Plugging this into~\eqref{eq:single-step-bound-helper-1}, and using again the separability of our instance, we can finally bound our desired probability
\begin{equation}
	\label{eq:single-step-bound-helper-3}
	\prob{\Delta \leq \eps} \leq \lambda \mu^{\beta} \eps \phi.
\end{equation}

Now we continue with the second part of the proof, in which we utilize the
probability bound in~\eqref{eq:single-step-bound-helper-3} to derive a concrete
bound on the expected number of iterations of standard local search. Let
$\mathcal{T}$ denote the random variable of that maximum length among all paths
in the transition graph of our instance. Then, our goal is to bound
$\expect{\mathcal{T}}$. 

Recall that the number of different possible cost-part configurations
is trivially upper-bounded by $ (M+1)^{\nu}$. Also, at every step of standard
local search, the configuration cost is \emph{strictly} decreasing. Thus,
$\mathcal{T}\leq (M+1)^{\nu}$. Furthermore, since the range of configuration
costs is within $[-M\nu,M\nu]$, and the minimum cost improvement of every iteration
is $\Delta$, we also know that $\mathcal{T}\leq
({2M\nu}/{\Delta})$. Using these, we get
\begin{align*}
    \expect{\mathcal{T}}
        &= \sum_{t=1}^{(M+1)^\nu}\prob{\mathcal{T}\geq t} \\
        &\leq \sum_{t=1}^{(M+1)^\nu}\prob{\Delta\leq \frac{2M\nu}{t}} \\
        &\leq 2\lambda \mu^\beta \phi M\nu\sum_{t=1}^{(M+1)^\nu}\frac{1}{t}, &&\text{due to~\eqref{eq:single-step-bound-helper-3},} \\
		&\leq 2\lambda \mu^\beta \phi M\nu \frac{3}{2} \log (M+1)^\nu, &&\text{since}\;\; 1+\frac{1}{2}+ \dots +\frac{1}{n} \leq \frac{3}{2}\log{n}\;\;\forall n\geq 2,\\
		&= 3\mu^{\beta} \lambda \phi \nu^2 M\log(M+1). 
\end{align*}
\end{proof}

\subsection{Limitations of our main theorem}
\label{sec:limitations-theorem}

The method we used in the proof of our main result, \cref{thm:single-step-bound}, is based on bounding the improvement of a single step in the procedure.
This approach has been successful in the derivation of a variety of prior work in smoothed analysis as we outline in \cref{sec:congestion-games-smoothed-models} and \cref{sec:sin-step-apps}, and it turns out that our generalized proof captures all the aforementioned works' results.

However, our model and proof do not capture some other recent results on smoothed analysis. We identified the following two different approaches in the literature:
\begin{enumerate}
	\item For some problems, the results rely on bounds on the improvement of \emph{sequences} of steps, rather than improvements of a single step. Notable examples are problems on graphs, and include \cites{Angel:2017aa}{Etscheid17}{Bibak21}{Chen20} on \lmaxcut and \cite{Englert_2016} on \tspIIopt. These results provide smoothed (quasi)polynomial bounds for the case of general graphs, whereas our method gives such bounds only for the case of graphs with at most (poly)logarithmic degree.
	\item Other results, not only use the above sequence-improvement analysis, but also heavily rely on the rich random interaction of the model at hand. For example, the smoothed polynomial time results of \cites{Angel:2017aa}{Bibak21}{Chen20} on \lmaxcut and \cite{Boodaghians20} on \netcoord rely on the fact that the graph is complete in the underlying problem.
\end{enumerate}
To tackle these cases with a generalized methodology, it seems that significantly new techniques are required.

Left as open research directions for future work, the following approaches could potentially capture the aforementioned cases.
\begin{enumerate}
	\item \label{enum:open-problem-theorem3-1-1} Instead of working on the base problem $A$, we could try to define a problem $B$ for which a single step essentially realizes a sequence of polynomially many single steps in $A$. Thus, the additional properties could potentially carry over to the single step of problem $B$. A generalized notion of separability could now apply to the single step of $B$ to achieve smoothed analysis results.
	\item \label{enum:open-problem-theorem3-1-2}We conjecture that a modification of our framework, which would take advantage of the high amount of randomness introduced in every single step of the procedure, would be capable to establish improved smoothed complexity results for some problems, such as \lmaxcut or \netcoord on graphs with $n$ nodes and degree at least $n - \lantheta{\log n}$. The modified framework would have to involve a complementary notion of separability that enables it to unify this proof approach.
\end{enumerate}

\section{Smoothed Analysis for Congestion Games}
\label{sec:congestion-games-smoothed-models}

Congestion games are composed of finite nonempty sets of \emph{players}
$\mathcal{N}=\intupto{n}$ and \emph{resources} $\mathcal{R}$. Each
player $i\in\mathcal{N}$ has a \emph{strategy set} $\Sigma_i\subseteq
2^\mathcal{R}$ and each resource $r\in\mathcal{R}$ has a \emph{cost} (or
\emph{latency}) \emph{function} $\kappa_r: \intupto{n} \map
\R_{\geq 0}$.
Each (pure) \emph{strategy profile} (or \emph{outcome})
$\vecc{\sigma}=(\sigma_1,\dots,\sigma_{n})\in
\vecc{\Sigma}\coloneqq \Sigma_1\times\dots\times\Sigma_{n}$
induces a \emph{load} on each resource $r$, equal to the number of players that use
it: 
$$\ell_r(\vecc{\sigma})\coloneqq\card{\ssets{i\in \mathcal{N} \fwh{r\in
\sigma_i}}}.$$
Then, the \emph{cost of player} $i$ is the total cost she experiences
from all resources that she is using: 
$$C_i(\vecc{\sigma})\coloneqq\sum_{r\in
\sigma_i} \kappa_r(\ell_r(\vecc{\sigma})).$$ 
An (exact) pure \emph{Nash equilibrium (PNE)} is an outcome $\vecc{\sigma}^*$ from
which no player can improve her cost by unilaterally deviating. Formally, for
any player $i\in \mathcal{N}$ and any deviation $\sigma_i'\in \Sigma_i$:  
$$C_i(\vecc{\sigma}^*)\leq
C_i(\sigma_i',\vecc{\sigma}^*_{-i}).$$
Thus, if a strategy profile $\vecc{\sigma}$ is \emph{not} a PNE, there has to
exist a player $i$ and a deviation $\sigma_i'\in \Sigma_i$ that reduces her
cost, i.e.
\begin{equation*}
\label{eq:BRD-def}
C_i(\sigma_i',\vecc{\sigma}_{-i}) < C_i(\vecc{\sigma}).
\end{equation*}
Such a strategy $\sigma_i'$ is then called a \emph{better-response} (of player $i$ with respect to the profile $\vecc{\sigma}$).

By the seminal work of~\textcite{Rosenthal1973a} we know that function $\Phi:\vecc{\Sigma}\map\R_{\geq 0}$ defined by
\begin{equation}
	\label{eq:rosenthal_potential}
	\Phi(\vecc{\sigma})=\sum_{r\in \mathcal R}\sum_{\ell=1}^{\ell_r(\vecc{\sigma})} \kappa_r(\ell)
\end{equation}
and commonly referred to as \emph{Rosenthal's potential}, has the property that
\begin{align*}
C_i(\sigma_i',\vecc{\sigma}_{-i}) - C_i(\vecc{\sigma})
=
\Phi(\sigma_i',\vecc{\sigma}_{-i}) - \Phi(\vecc{\sigma}) && \forall\vecc{\sigma}\in\vecc{\Sigma}\;\forall i\in\mathcal{N}\; \forall \sigma_i'\in\Sigma_i.
\end{align*}
In other words, function $\Phi$ is an (exact)
\emph{potential}~\parencite{Monderer1996a} of the corresponding congestion game.
This implies that PNE of a congestion game correspond \emph{exactly} to the set
of \emph{local} minimizers of its Rosenthal's
potential~\eqref{eq:rosenthal_potential}, meaning that $\Phi(\vecc{\sigma}^*)
\leq \Phi(\sigma_i',\vecc{\sigma}^*_{-i})$ for any player $i$ and any deviation
$\sigma_i'\in \Sigma_i$. This also immediately establishes the existence of PNE
in \emph{all} congestion games, since the potential function
$\Phi(\vecc{\sigma})$ can only take finitely many different values.

Depending on the type and representation of the cost functions, different
classes of congestion games can arise. Below, we describe three prominent ones which we will focus on
in this paper:

\begin{itemize}
	\item \emph{General.} The cost functions are given explicitly as a list of
	nonnegative values, one for each possible load on the resource;
	$\left(\kappa_r(1),\kappa_r(2),\dots,\kappa_r(n)\right)$.
	Notice how in this model we do not impose a monotonicity constraint; this is
	deliberate, to maintain full generality. If one wants to
	focus on nondecreasing cost functions (as is often the case in the
	literature), then the step-function representation (see below) can easily be used
	instead.
	\item \emph{Polynomials (of degree $d$).} The cost functions are polynomials
	of maximum degree $d\in\N$ with nonnegative coefficients. More specifically,
	the cost functions are given implicitly by the coefficients
	$\sset{\alpha_{r,j}}_{r\in \mathcal{R},j\in \intuptozero{d_r}}\subseteq \R_{\geq 0}$, where $d_r\leq
	d$, via
	\begin{equation}
	\label{eq:congestion_poly_costs_def}
	\kappa_r(\ell)=\sum_{j=0}^{d_r} \alpha_{r,j} \ell^j
	\qquad\text{for all}\;\; \ell\in\intupto{n}.
	\end{equation} 
	\item \emph{Step functions (with $d$ break-points).} The cost functions are
	nondecreasing, piecewise constant, given by pairs of break-points and
	value-increases. More specifically, for each resource $r \in \mathcal{R}$ there is a list of break-points $\mathcal{B}_r \subseteq \intupto{n}$ and associated jumps $\sset{\alpha_{r,j}}_{j\in\mathcal{B}_r}$. We denote the number of break-points of a resource $r$ by $d_r\coloneqq \card{\mathcal{B}_r}$, and we set $d\coloneqq\max_{r\in \mathcal{R}} d_r$. Then, the cost functions are given via
	\begin{equation}
		\label{eq:congestion_step_costs_def}
	\kappa_r(\ell)=\sum_{j \in \mathcal{B}_r \inters \intupto{\ell}}\alpha_{r,j}\qquad\text{for all}\;\; \ell\in\intupto{n}.
	\end{equation}
\end{itemize}

To the best of our knowledge, no smoothed analysis model has been established so
far for congestion games. In this paper, we propose and study the following
perturbation semantics for the aforementioned classes.
Note that as is commonly the case in smoothed analysis, the impact of randomness in the given input parameters depends crucially on the way they are used to compute the cost in the model.
While the notion of \emph{general congestion games} does typically capture all possible latency functions in the deterministic setting, the results for smoothed analysis do not directly translate between the models.\footnote{See e.g.\ the discussion in \cref{sec:localmaxcut} for \lmaxcut, where the current state-of-the-art results for the complete graph are better than for arbitrary graphs.}
\begin{itemize}
	\item \emph{General congestion games.} The costs $\kappa_{r}(\ell)$ are
	independently distributed according to densities $f_{r,\ell}:[0,1]\map
	[0,\phi]$, for all $r\in \mathcal{R}$ and $\ell\in \intupto{n}$.
	\item \emph{Polynomial games.} The coefficients $\alpha_{r,j}$ are independently distributed according to densities $f_{r,j}:[0,1]\map [0,\phi]$, for all $r\in \mathcal{R}$ and $j\in\intuptozero{d_r}$.
	\item \emph{Step-function games.} The jump increases $\alpha_{r,j}$ are
	independently distributed according to densities $f_{r,j}:[0,1]\map
	[0,\phi]$, for all $r\in \mathcal{R}$ and $j\in\mathcal{B}_r$. Notice,
	however, that the break-points $\mathcal{B}_r$ themselves are not subjected
	to any noise, and they are assumed to be fixed (and adversarially selected).
\end{itemize}

\subsection{Nash Equilibria as Combinatorial Local Optimization Problems}
\label{sec:NE_as_CLO}
We now show how \congestion, the problem of finding a pure Nash equilibrium in
congestion games, can actually be interpreted as a combinatorial local
optimization problem (with respect to our definitions in~\cref{sec:smoothPLS}),
for any of the cost models described above. For all models, their randomness
semantics translate directly to the respective randomness of cost coefficients
in the smoothed CLO problem. 
\begin{itemize}
	\item \emph{General congestion games.} By~\eqref{eq:rosenthal_potential} we
	know that PNE correspond exactly to local minimizers of the potential
	function
	$\Phi(\vecc{\sigma})=\sum_{r\in \mathcal{R}}\sum_{j=1}^{\ell_r(\vecc{\sigma})}\kappa_r(j)$.
	Therefore, finding a PNE of a general congestion game can be viewed as a binary ($M=1$)
	CLO problem, with cost dimension $\nu=\card{\mathcal
	R}n$ (the cost coordinates are given by $\mathcal{R} \times \intupto{n}$), where each strategy profile $\vecc{\sigma}$ is mapped to a
	cost configuration $\vecc s^\costv=\left(s_{r,j}\right)_{r\in \mathcal{R},j\in\intupto{n}}$ given by the indicator functions:
	\begin{equation}
		\label{eq:congestion-games-cost-config-map}
	s_{r,j}= \indicator{j \leq \ell_r(\vecc{\sigma})}.
	\end{equation}
	The CLO cost coefficients are given by $c_{r,j}=\kappa_{r}(j)$. 

	Furthermore, we want to establish a one-to-one correspondence between
	configurations $\vecc s$ of the CLO problem and strategy profiles
	$\vecc{\sigma}$ of our congestion games where, in particular, the neighbours
	of $\vecc s$ are exactly the configurations corresponding to single-player
	deviations $\ssets{(\sigma_i',\vecc{\sigma}_{-i})}_{i\in\mathcal N,
	\sigma_i'\in \Sigma_i}$. In that way, better-responses of the congestion
	game would correspond exactly to local improvements in the CLO formulation.
	However, this cannot be achieved by using just the cost part defined
	above by~\eqref{eq:congestion-games-cost-config-map}, since in that case,
	different strategy profiles may end up being mapped to the same cost
	part configuration (when they induce the same resource loads). To overcome this
	technical pitfall, we also maintain a non-cost part $\vecc s^\noncostv$, which keeps track of the actual strategies of the players: this can
	be easily achieved, with only an additional polynomial size
	burden.\footnote{E.g.\ we can choose $\bar{\nu}=\card{\mathcal{N}}$ and let
	$s^\noncostv_i\in\intupto{\card{\Sigma_i}}$ be the index of strategy
	$\sigma_i$ deployed by player $i$ in profile $\vecc{\sigma}$. To simplify
	our exposition, in the remaining congestion game classes studied below, we
	will avoid explicitly discussing these non-cost parts; they are
	identical to the general congestion game model.}

	Finally, notice that the neighbourhoods of the CLO problem we created can be explicitly listed and
	efficiently searched for a better (smaller cost) value: they have a maximum
	size of $n\cdot \max_{i\in \mathcal N} \card{\Sigma_i}$, which is polynomial in the
	description of the original congestion game.

	\item \emph{Polynomial games.} Using~\eqref{eq:congestion_poly_costs_def}, Rosenthal's potential~\eqref{eq:rosenthal_potential} can now be written as
	$$
	\Phi(\vecc{\sigma})
	=\sum_{r\in \mathcal R} \sum_{\ell=1}^{\ell_r(\vecc{\sigma})}\sum_{j=0}^{d_r} \alpha_{r,j} \ell^j
	= \sum_{r\in \mathcal R}\sum_{j=0}^{d_r} \alpha_{r,j} \sum_{\ell=1}^{\ell_r(\vecc{\sigma})} \ell^j
	=\sum_{r\in \mathcal R}\sum_{j=0}^{d_r}\alpha_{r,j} \mathfrak{S}_j(\ell_r(\vecc{\sigma})),
	$$
	where 
	$$\mathfrak{S}_{j}(\ell)\coloneq\sum_{k=1}^{\ell}
	k^{j}\leq \ell^{j+1}\leq \ell^{d+1} \leq n^{d+1}$$ for any
	$\ell\in\N$.
	This induces a CLO problem with parameters 
	$$
	\nu=\sum_{r\in \mathcal R} (d_r+1)\leq \card{\mathcal R}(d+1)
	\qquad\text{and}\qquad
	M = {n}^{d+1},
	$$ 
	with each strategy profile $\vecc{\sigma}$ corresponding to a cost configuration
	$\vecc{s}^\costv=\left(s_{r,j}\right)_{r\in\mathcal{R},j\in\intuptozero{d_r}}$
	given by:
	$$
	s_{r,j} = \mathfrak{S}_{j}(\ell_r(\vecc{\sigma})).
	$$
	The costs coefficients are given by $c_{r,j}=\alpha_{r,j}$. Notice that again all
	neighbourhoods are efficiently searchable since they have a polynomial
	maximum size of $n\cdot \max_{i\in \mathcal{N}}\card{\Sigma_i}$.
	
	\item \emph{Step-function games.} 
	Using~\eqref{eq:congestion_step_costs_def}, Rosenthal's potential~\eqref{eq:rosenthal_potential} can now be written as
	\begin{align*}
	\Phi(\vecc{\sigma})
		= \sum_{r\in\mathcal{R}} \sum_{\ell=1}^{\ell_r(\vecc{\sigma})}\sum_{j\in\mathcal{B}_r\inters\intupto{\ell}}\alpha_{r,j}
		&= \sum_{r\in \mathcal R} \sum_{j\in\mathcal{B}_r\inters\intupto{\ell_r(\vecc{\sigma})}}(\ell_r(\vecc{\sigma})-j+1)\alpha_{r,j}\\
	    &= \sum_{r\in \mathcal R} \sum_{j\in\mathcal{B}_r}\max(0,\ell_r(\vecc{\sigma})-j+1)\alpha_{r,j}.
	\end{align*}
	
	This induces a CLO problem with dimension 
	$$
	\nu= \sum_{r\in \mathcal{R}} \card{\mathcal{B}_r}=\sum_{r\in \mathcal{R}} d_r\leq \card{\mathcal R} d
	\qquad\text{and}\qquad
	M = n,
	$$ 
	with each strategy profile $\vecc{\sigma}$ corresponding to configuration
	$\vecc{s}^\costv=\left(s_{r,j}\right)_{r\in\mathcal{R},j\in\intupto{d_r}}$
	given by
	$$
	s_{r,j} = \max(0, \ell_r(\vecc{\sigma}) - \mathcal{B}_r(j) + 1),
	$$
	where we use $\mathcal{B}_r(j)$ to denote the $j$-th
	break-point\footnote{That is, if
	$\mathcal{B}_r=\ssets{b_1,b_2,\dots,b_{d_r}}\subseteq\intupto{n}$,
	then $\mathcal{B}_{r}(j)=b_j$, for every $j\in\intupto{d_r}$.} of resource
	$r$. The costs are given by $c_{r,j}=\alpha_{r,\mathcal{B}_r(j)}$. Again, it
	is straightforward to see that all neighbourhoods are efficiently
	searchable.
\end{itemize}

\subsection{Restrained Congestion Games}
\label{sec:congestion-positive}

Although congestion games are guaranteed to have (at least one) PNE, the
computational problem of actually finding one is considered hard; as a matter of
fact, the problem \congestion is one of the most prominent \pls/-complete
problems.
Our goal in this section is to investigate whether this computational
barrier can be bypassed, under the lens of the more optimistic complexity model
of smoothed analysis. 

To achieve this, we establish an upper bound
(\cref{th:smoothed-p-bounded-interaction-congestion-games}) on the expected
number of better-responses that need to be performed until a PNE is found in a
congestion game, as a function of a critical structural parameter of its action
space that we identify (see~\cref{def:restrained-games}). Then, we can deduce
that for congestion games in which this parameter is appropriately bounded, the
smoothed running time becomes tractable (\cref{th:restrained-cg-smoothed-poly}).
At the same time, we show how such a restriction does not make the
problem trivially tractable, by proving that \congestion remains \pls/-complete even for this
subclass of games (\cref{th:pls-hardness-restrained}).

\begin{definition}[Restrained Congestion Games]
	\label{def:restrained-games}
	A congestion game will be called \emph{$B$-restrained}, where $B\in\N$, if the maximum number of resources changed by any single-player deviation is at most $B$. Formally, 
	$$
	\max_{i\in\mathcal{N}}\max_{\sigma,\sigma'\in\Sigma_i} \card{\symd{\sigma}{\sigma'}} \leq B,
	$$
	where $\symd{}{}$ denotes the standard symmetric difference operator (recall that $\sigma,\sigma'$ are subsets of resources).
\end{definition}

\begin{theorem}
	\label{th:smoothed-p-bounded-interaction-congestion-games}
	Consider a $B$-restrained $n$-player congestion game, under any
	of the smoothed-analysis models described
	in~\cref{sec:congestion-games-smoothed-models} (namely general, polynomial,
	or step-function latencies), with maximum density parameter $\phi$. Then,
	performing any better-response dynamics, starting from an arbitrary strategy
	profile, converge to an (exact) PNE of the congestion
	game in an expected number of iterations that is bounded by
	\begin{itemize}
		\item $\landau{n^{B+3}k^2m^2 \phi}$ for general games,
		\item $\landau{n^{B+d+2}\log(n)(d+1)^3k^2 m^2\phi}$ for polynomial games with degree at most $d$, and
		\item $\landau{(d+1)^{B+2} n^2\log(n) k^2 m^2 \phi}$ for step-function games with at most $d$ break-points,
	\end{itemize}
	where $m=\card{\mathcal{R}}$ is the number of resources and $k = \max_{i\in\mathcal{N}} \card{\Sigma_i}$ is the maximum strategy set size.
\end{theorem}

\begin{proof}
	In \cref{sec:NE_as_CLO} we already showed how congestion games can be
	interpreted as CLO problems. In particular, we established a one-to-one
	correspondence between better-responses of the players to local improvements
	of the CLO cost objective (which corresponds to the value of Rosenthal's
	potential). Using this interpretation, we can now make use of our main
	black-box tool from~\cref{sec:sin-step}: bounding the
	expected number of local search steps will induce the same bound in the
	expected iterations of better-response dynamics in the original congestion
	game. Therefore, the gist of our proof is to construct coverings
	$(\mathcal{E}, \mathcal{I})$ so that the induced CLO problem can be shown to
	be $(\lambda, \beta, \mu)$-separable (see~\cref{def:separable-instances})
	for parameters with appropriately small magnitude. 

	We start by establishing some properties that will be shared across all
	three different cost models. Fix a congestion game and its corresponding CLO
	problem, as described in~\cref{sec:NE_as_CLO}.\footnote{We will use standard
	notation for the various components and parameters of the game and the CLO
	problem, as introduced above
	in~\cref{sec:congestion-games-smoothed-models}.} For convenience, we denote
	by $k_i\coloneqq \card{\Sigma_i}$ the size of the strategy set of a player
	$i$, and let $k\coloneqq \max_{i\in\mathcal{N}} k_i$.
	Also, in all models we have a similar index structure in the cost part;
    the cost-part is given by $(s_{r, j})_{r\in \mathcal{R}, j\in J_r}$, where $J_r$ depends on the cost function model ($J_r = \intupto{n}$ for general, $J_r = \intuptozero{d_r}$ for polynomial, and $J_r = \intupto{d_r}$ for step-function costs).

	The transition cover $\mathcal{E}$ is constructed from clusters that collect all edges in the neighbourhood graph (of the CLO problem) that correspond to a fixed deviation of a player, regardless of the configuration of the remaining players.
	Formally, we let
		\begin{gather*}
			\mathcal{E}\coloneqq \sset{E(i,
			\sigma_i, \sigma_i')\fwh{i\in\mathcal{N},\; \sigma_i,\sigma_i'\in\Sigma_i}}, \quad
			\text{where}\\
			E(i,\sigma_i, \sigma_i')\coloneqq \sset{\left(\vecc{s}(\sigma_i,\vecc{\sigma}_{-i}),\vecc{s}(\sigma_i',\vecc{\sigma}_{-i})\right)\in E\fwh{\vecc{\sigma}_{-i}\in\vecc{\Sigma}_{-i}}},
		\end{gather*}
	and $\vecc{s}(\vecc{\sigma})$ is used to denote the CLO configuration corresponding to strategy profile $\vecc{\sigma}$ in the congestion game.
	Now we immediately get the bound
		$$
		\card{\mathcal{E}} \leq n k (k-1),
		$$
	which will be used as the value for our separability parameter $\lambda$ (see~\cref{def:separable-instances}).
	
	Next, for the coordinate cover $\mathcal{I}$, we cluster the indices with respect to each resource, i.e.\ we choose
		$$
		\mathcal{I}\coloneqq\sset{I_r\fwh{r\in\mathcal{R}}},
		\quad\text{where}\quad
		I_r = \{(r, j) \fwh{j\in J_r\}}.
		$$
	In congestion games, a deviation $\sigma_i \to \sigma_i'$ only affects the resources $ r\in \symd{\sigma_i}{\sigma_i'}$.
	Their loads are changed to increase by $1$ for $ r\in \sigma_i' \setminus \sigma_i$ and decrease by $1$ for $ r\in \sigma_i \setminus \sigma_i'$; the load of all other resources does not change.
	Our choice of the cover $\mathcal{I}$, therefore, will allow us to settle $\beta$ for all models due to the $B$-restraint assumption on the size of $\symd{\sigma_i}{\sigma_i'}$. In more detail, recall that
	the cost parts of a configuration depend only on the loads of the resources, thus
	all components associated with resources $r\notin \symd{\sigma_i}{\sigma_i'}$ (specifically, the components with coordinates $I_r$) remain unchanged during the transition, since the load of $r$ does not change either. The fact that the size of those sets $\symd{\sigma_i}{\sigma_i'}$ is universally bounded, by assumption, will let us use $\beta = B$ as a separability parameter (see~\cref{def:separable-instances}).

	The remaining parameter $\mu$ depends on the structure of the configurations in the cost part.
	Again, we emphasize that for all cost models in \cref{sec:NE_as_CLO}, the sub-configuration $\vecc s_{I_r} (\vecc\sigma)$, which comprises all components of $\vecc s(\vecc\sigma)$ that correspond to a resource $r$, depends only on the load $\ell_r(\vecc \sigma)$ of resource $r$ (under strategy profile $\vecc \sigma$).
	We can therefore represent it as a function $\vecc h_r:\intuptozero{n}\to \intuptozero{M}^{J_r}$, i.e., $\vecc s_{I_r}(\vecc \sigma) = \vecc h_r(\ell_r(\vecc \sigma))$.
	For ease of notation we write $\vecc s_r = \vecc s_{I_r}$ in the following.

	To discuss $\range{I_r}{E(i, \sigma_i, \sigma_i')}$, we need to consider the configuration changes given by
		\begin{equation*}
			\vecc s_r(\vecc \sigma)-\vecc s_r(\vecc \sigma') = \begin{cases}
				\vecc h_r(\ell_r(\vecc \sigma))-\vecc h_r(\ell_r(\vecc \sigma')) = \vecc h_r(\ell_r(\vecc \sigma))-\vecc h_r(\ell_r(\vecc \sigma)+1),& r\in \sigma_i'\setminus \sigma_i,\\
				\vecc h_r(\ell_r(\vecc \sigma))-\vecc h_r(\ell_r(\vecc \sigma')) = \vecc h_r(\ell_r(\vecc \sigma))-\vecc h_r(\ell_r(\vecc \sigma)-1),& r\in \sigma_i\setminus \sigma_i'. 
			\end{cases}
		\end{equation*}
	The only variable in this expression is therefore the initial load $\ell_r(\vecc \sigma)$.
	In either case, there are $n$ possible initial loads for each resource\footnote{Note that although there are $n+1$ different loads $0,\dots, n$, a load-increasing resource cannot already have load $n$ and a decreasing one cannot have $0$.} and therefore also at most $n$ difference-vectors within $\range{I_r}{E(i, \sigma_i, \sigma_i')}$; thus $\mu \leq n$.
	In the following, we will discuss the actual configuration difference structure for each model and whether we can improve $\mu$ over this basic bound.
	
	\begin{itemize}
		\item \emph{General. }
			The CLO representation follows a binary model with $M=1$ and $ \nu = m n$: each resource $r$ corresponds to components $s_{r, j}$, $j = 1,\dots, n$ (the indices from $I_r$), with a value of $s_{r, j} = 1$ for $j \leq \ell_r(\vecc \sigma)$ and $s_{r, j} = 0$ otherwise.
			The function $\vecc h_r: \intuptozero{n} \map \{0, 1\}^{n} $ is thus given by 
			$$(\vecc h_r(\ell))_j = \indicator{j \leq \ell}, \qquad\text{for all}\;\; j\in[n].$$

			In particular, the vectors $\vecc \delta \in \range{I_r}{E(i, \sigma_i, \sigma_i')}$ are a result of moving the rightmost entry with value $1$ within the vector $\vecc s_r$ to the next larger or smaller load in the configuration, i.e.\ they are given by: $ \delta_{\ell_r(\vecc \sigma)+1} = -1 $ for $r\in \sigma_i'\setminus \sigma_i$; $\delta_{\ell_r(\vecc \sigma)} = +1$ for $r\in \sigma_i\setminus \sigma_i'$; and zeroes elsewhere. Since there are $n$ many such vectors for every resource $r$, we cannot improve over $\mu = n$.

			In this case, thus, \cref{thm:single-step-bound} yields an expected running time of at most
				\[ 3\cdot n^B nk(k-1)  \cdot (m n)^2 \cdot \phi = \landau{n^{B+3}k^2m^2 \phi}. \]
		\item \emph{Polynomial games. }
			For this model, each (cost-part) configuration component $s_{r, j}$ is given by accumulated monomials $\mathfrak{S}_j(\ell_r(\vecc \sigma))$, for degrees $j = 0,1,\dots, d_r$. Thus, $\nu = m(d+1)$, and also $M=n^{d+1}$, in order to capture all possible values of these functions.
			Therefore, we now get $\vecc h_r: \intuptozero{n} \map \intuptozero{n^{d+1}}^{d_r+1} $ with
				\[ \vecc h_r(\ell) = \left(\mathfrak{S}_0(\ell), \dots, \mathfrak{S}_{d_r}(\ell)\right) \]
			for the configuration component of $I_r$.
			Again, we cannot do better than the basic bound, so we use $\mu = n$.
			
        	Similarly to the previous case for general latency functions, using~\cref{thm:single-step-bound} we can bound the expected number of better-response iterations by
				\[ 3\cdot n^B nk(k-1)  \cdot (m(d+1))^2 n^{d+1}\log(n^{d+1}) \cdot \phi = \landau{n^{B+d+2}\log(n)(d+1)^3k^2 m^2\phi}. \]
		\item \emph{Step-function games. }
			The configuration mapping is now given by $s_{r, j} = \max\left(0, \ell_r(\vecc{\sigma}) - \mathcal{B}_r(j) + 1\right)$ and $M=n, \nu = m\cdot d$.
            Therefore, the (cost-part) configuration components of a resource $r$ are represented by the function $\vecc h_r:\intuptozero{n}\map \intuptozero{n}^{d_r}$ given by
				\[\vecc h_r(\ell) = \left(\max(0, \ell - \mathcal{B}_{r}(1) + 1),\dots, \max(0, \ell - \mathcal{B}_{r}(d_r) + 1)\right). \]

			In this case, we can do even better than the basic $\mu = n$ bound.
			We investigate the structure of the differences with respect to each coordinate $j\in J_r$, for an increase in the load of a resource $r$, i.e.\ for $r\in \sigma'_i\setminus \sigma_i$ (the decreasing case follows analogously):
				\begin{align*}
					\left(\vecc h_r(\ell)-\vecc h_r(\ell+1)\right)_j &= \max\left(0, \ell - \mathcal{B}_{r}(j) + 1\right) -\max\left(0, (\ell+1) - \mathcal{B}_{r}(j) + 1\right)\\
					&= \begin{cases}
						0 - 0 = 0, &\text{if}\;\; \ell < \mathcal{B}_{r}(j)-1, \\
						-\left((\ell+1) - \mathcal{B}_{r}(j) + 1 - 0\right) = -1, &\text{if}\;\; \ell = \mathcal{B}_{r}(j) - 1,\\
						(\ell - \mathcal{B}_{r}(j) + 1) - ((\ell+1) - \mathcal{B}_{r}(j) + 1) = -1, &\text{if}\;\; \ell > \mathcal{B}_{r}(j) - 1.
					\end{cases}
				\end{align*}

			Because the jump points $\mathcal{B}_{r}(j)$ are ordered increasingly with respect to $j$, the resulting vectors $\vecc h_r(\ell)-\vecc h_r(\ell+1)$ are of the form $(-1,\dots, -1, 0, \dots, 0)$, including the zero vector.
			Therefore, with respect to the coordinate cluster $I_r$, there are at most $d_r+1$ possible distinct vectors for the (cost-part) configuration differences, and we choose $\mu = \max_r d_r + 1 = d + 1$.
			By \cref{thm:single-step-bound} we can thus bound the expected number of iterations by
				\[ 3\cdot (d+1)^B nk(k-1) \cdot (m\cdot d)^2 n\log(n) \cdot \phi = \landau{(d+1)^{B+2} n^2\log(n) k^2 m^2 \phi}. \]
	\end{itemize}
\end{proof}

An immediate corollary
of~\cref{th:smoothed-p-bounded-interaction-congestion-games} is that, when
congestion games are sufficiently restrained, PNE can be found efficiently
via better-response dynamics. In what follows, by a \emph{constantly}- and \emph{polylogarithmically}-restrained congestion game we mean a $B$-restrained game with $B \in \landau{1}$ and $B \in \landau{\log^c N}$ for some constant $c > 0$, respectively, where $N$ is the size of the input.
\begin{corollary}
    \label{th:restrained-cg-smoothed-poly}
    Better-response dynamics terminate in polynomial smoothed
    time for the class of constantly-restrained congestion games,
    and in quasipolynomial smoothed time for
    polylogarithmically-restrained games, under any cost model.
    Under the step-function cost model, in $\landau{\log N}$-restrained congestion games with a constant number of steps $d$, better-response dynamics terminate in polynomial smoothed time.
\end{corollary}

We now show that the class of constantly-restrained congestion games, for which~\cref{th:restrained-cg-smoothed-poly} provides efficient smooth running time, constitutes a computationally meaningful restriction of arbitrary congestion games, since they can still encode the \pls/-completeness of the original problem.
The hardness is a straightforward adaptation of that in~\cite{Roughgarden16}, and so the proof of~\cref{th:pls-hardness-restrained} is deferred to~\cref{sec:congestion-hardness} for completeness.
It makes use of the fact that \lmaxcut is \pls/-complete even for constant-degree graphs.

\begin{theorem}
\label{th:pls-hardness-restrained}
The problem of computing a PNE of a
constantly-restrained congestion game is \pls/-complete, for all the input models
described in~\cref{sec:congestion-games-smoothed-models} (namely general,
polynomial, or step-function cost representations). 
\end{theorem}

In addition to the conditional intractability that \pls/-hardness implies for these families of congestion games, we show an unconditional lower bound on the worst-case running time of the standard local search algorithm of the problem. 
We do this by using the notion of a \emph{tight \pls/-reduction}, as discussed in \cref{sec:model}. 
Since we reduce from the problem \lmaxcutd for $d \geq 5$ (defined in \cref{sec:congestion-hardness}), which admits a configuration starting from which the standard local search algorithm needs exponentially many iterations (see discussion after proof of \cref{thm:local_max-k-cut}), our tight \pls/-reduction implies that standard local search of our families of congestion games -- under \emph{any} pivoting rule --  takes exponential time in the worst case.

\subsection{Network Congestion Games}
\label{sec:network-games-compact}

An interesting, and very well studied, variation on the vanilla
representation model for congestion games (which we presented at the start
of~\cref{sec:congestion-games-smoothed-models} above) is that of \emph{network}
congestion games. In such games, the strategy sets $\Sigma_i$ of the players are
not given explicitly in the input, but implicitly via an underlying directed
graph $G$ whose edges constitute the resources of the  game. More precisely, for
each player $i$ we are given an origin $o_i$ and a destination $d_i$ node of
$G$. Then, $\Sigma_i$ is defined (implicitly) as the set of all (simple) $o_i\to d_i$ paths in $G$. Importantly, this means that now players may
have exponentially many strategies available to them. 

A critical implication is that $\Sigma_i$ cannot be searched \emph{exhaustively}
for better-responses. However, a better-response can still be found efficiently:
keeping all other players fixed, a minimum-cost strategy of player $i$ is a
shortest path on graph $G$ with edge costs equal to the cost
$c_r\left(\ell_r(\vecc{\sigma}_{-i})+1\right)$ of an edge/resource $r$ when used
by player $i$. This means that actually a polynomial-time \emph{best}-response
oracle does exist. This immediately places network congestion games in the
complexity class \pls/ (since neighbourhoods can be searched efficiently for a
local cost improvement; see our discussion in~\cref{sec:model}) and thus it
constitutes a valid CLO problem (via a similar interpretation as we did for
general congestion games in~\cref{sec:NE_as_CLO}). To emphasize this, we
will refer to these ``canonical'' best-response dynamics of network congestion
games as \emph{shortest-path} dynamics.

Finding PNE of network congestion games remains a \pls/-complete
problem~\parencite{Ackermann2008}. Given the prominence of these games, both in
the theoretical and applied literature, in this section we want to identify
conditions under which network congestion games inherit the desirable properties of their
general counterparts that allow them to be tractable under
smoothed analysis. More precisely, can our positive result from~\cref
{th:smoothed-p-bounded-interaction-congestion-games} be applied to network congestion games
in a straightforward way? And which structural parameters are now relevant for
the running time bound? At the same time, the network congestion game instances that
allow for efficient smoothed solutions should still be interesting enough to
remain \pls/-hard under traditional worst-case analysis. We introduce the following family of network congestion games that are defined by two parameters.

\begin{definition}[Compact Network Congestion Games]
	\label{def:compact-network-game}
	For $A, B$ positive integers, a network congestion game is called \emph{$(A,B)$-compact} if (\ref{prop:bounded-space-equilibria-network}) each player has at most $A$ different best-response strategies, and (\ref{prop:short-paths-equilibria-network}) all such strategies are paths of length at most $B$.
	Formally, there exist strategy sets
	$\Sigma_i^*\subseteq{\Sigma_i}$, such that:
	\begin{enumerate}[(a)]
		\item\label{prop:bounded-space-equilibria-network} $\card{\Sigma_i^*}\leq A$ for all $i\in\mathcal{N}$ and $\argmin_{\sigma_i\in\Sigma_i} C_i(\sigma_i,\vecc{\sigma}_{-i})\subseteq \Sigma_i^*$ for all $i\in\mathcal{N}, \vecc{\sigma}_{-i}\in\vecc{\Sigma}_{-i}$, and
		\item\label{prop:short-paths-equilibria-network} $\card{\sigma_i}\leq B$ for all $i\in\mathcal{N}, \sigma_i\in\Sigma_i^*$.
	\end{enumerate}
\end{definition}

Property~(\ref{prop:short-paths-equilibria-network})
above will serve the purpose of imposing the restraint condition
(see~\cref{def:restrained-games}) needed to
deploy~\cref{th:smoothed-p-bounded-interaction-congestion-games}.
Property (\ref{prop:bounded-space-equilibria-network}) will help us constrain the
exponentiality of the strategy space of network games, in order to be able to
handle them using tools designed for general games.
Intuitively, this property can be related to classical vehicle routing
settings in which players can a priori exclude unreasonable detours or paths
that involve a road with a construction site with large delay.
Both
properties are illuminated within the proof of our following positive result for
network congestion games.

\begin{theorem}
\label{th:smoothed-bound-dynamic-network-games}
Consider an $(A,B)$-compact $n$-player network congestion game, under any
	of the smoothed-analysis models described
	in~\cref{sec:congestion-games-smoothed-models} (namely general, polynomial,
	or step-function latencies), with maximum density parameter $\phi$. 
    Then, performing shortest-path dynamics, starting from an arbitrary strategy profile, converges to a PNE of the
	game in an expected number of iterations that is polynomial in $\phi$,
	$(d+1)^B$, $A$, and the description of the game, where the parameter $d$ depends on the cost function representation. In particular: for general latencies, $(d+1)$ can be replaced by $n$; for polynomial latencies, $d$ is the maximum degree; and for step-functions, $d$ is the maximum number of break-points.
\end{theorem}

\begin{proof}
    Let us count iterations, i.e., steps of the dynamics, starting from step $0$ when the dynamics are in the initial arbitrary strategy profile. We will say that the dynamics are in step $t$ if $t$ improving moves have taken place since step $0$.
    
    We introduce a random variable $S_i$ for the step of the dynamics in which player $i$ best-responded (that is, chose a shortest path) for the first time. Observe that once player $i$ best-responds, by definition, she plays a strategy in $\Sigma_i^*$ and from then on, whatever strategy she plays will be in that set. Without loss of generality, let us rearrange the players according to the order with which each best-responds for the first time, meaning that player $i$ best-responds only after all players $j$ with $j < i$ did so.
    Focusing on a player $i+1$, $i < n$, that has not best-responded yet in a given step of the dynamics, notice that neither have all players $j > i+1$.
    
    Let us now define the subgame $G_{i}$ in which only players in $[i]$ participate and where each edge has its load artificially increased according to the fixed strategies of the other players.
    The initial strategy profile of each subgame $G_i$ is the profile of the original game's dynamics after step $S_i$, which is in $\prod_{j\in [i]} \Sigma_j^*$, by definition.
    $G_{i}$ allows us to exactly reflect the original game's dynamics between time steps $S_i$ and $S_{i+1}-1$, while all strategies of the participating players are in $\prod_{j\in [i]} \Sigma_j^*$.
    
    Denote by $T_i$ the random variable of the number of iterations until a PNE in $G_{i}$ is reached.
    As soon as this is the case, the next time step in the original game must involve a new player, so $S_{i+1} \leq S_{i} + T_{i} + 1$.
    Note that not all $T_i$ steps of $G_i$ are necessarily played in the original game's dynamics, as player $i+1$ can deviate from their original action at time step $S_{i+1}$ before an equilibrium is reached in $G_i$.

    Now we can express the number of better-response steps in the original network congestion game as $T = S_{n} + T_n$, and by the observation above, we can recurse to bound this variable by
        $$
            T = S_{n} + T_n \leq S_{n-1} + T_{n-1} + 1 + T_n \leq \cdots \leq  (n-1) + \sum_{i=1}^n T_n.
        $$
    
    Now notice that in subgame $G_i$ of an $(A,B)$-compact network congestion game, any better-response dynamics (such as best-response dynamics) are equivalent to the same dynamics on a $2B$-restrained\footnote{When each strategy consists of at most $B$ edges, their pairwise symmetric differences have cardinality at most $2B$.} $i$-player congestion game with maximum strategy set size $A$ (see \cref{sec:congestion-positive}). Therefore, we get that
	\begin{align*}
		\expect{T} \leq (n-1) + \sum_{i=1}^{n} \expect{T_i},
	\end{align*}
    where the expectation $\expect{T_i}$ is bounded polynomially due to \cref{th:smoothed-p-bounded-interaction-congestion-games} for the various latency representations.
\end{proof}

An immediate consequence of~\cref{th:smoothed-bound-dynamic-network-games} is that (analogously to~\cref{th:restrained-cg-smoothed-poly} for general congestion games) in $(A,B)$-compact network games with sufficiently small parameters $A, B$, a PNE can be found efficiently under smoothness. 

\begin{corollary}
	\label{th:compact-netc-smoothed-poly}
    Let $N$ be the size of the input, and $A$ be a polynomial in $N$. Shortest-path dynamics on $(A,B)$-compact network congestion games within the general, polynomial, and step-function cost model terminate in polynomial smoothed time when $B \in \landau{1}$ and in quasipolynomial smoothed time when $B \in \landau{\log^c N}$ for some constant $c > 0$.
    Under the step-function cost model, when $A$ is a polynomial in $N$, $B \in \landau{\log N}$, and the number of steps $d$ is constant, shortest-path dynamics terminate in polynomial smoothed time.
\end{corollary}

The following hardness result establishes that such games, even for $A, B \in \landau{1}$, are \pls/-hard. The proof is deferred to~\cref{sec:pl-hardness-network-compact}. It is based on the reduction constructed by~\textcite{Ackermann2008}, with special care taken in order to incorporate constant-length paths that can be established by making use of the fact that \lmaxcut is \pls/-complete even for constant degree graphs.

\begin{theorem}
	\label{th:pls-hardness-compact-network}
	The problem of computing a PNE of an $(A,B)$-compact network congestion game is \pls/-complete, even for $A, B \in \landau{1}$, for all the
	input models described in~\cref{sec:congestion-games-smoothed-models}
	(namely general, polynomial, or step-function latencies). 
\end{theorem}

Similarly to our \pls/-hardness reduction of \cref{th:pls-hardness-restrained}, the above theorem's \pls/-reduction is tight, in the sense of \cite{Schaffer91} (see \cref{sec:model}). The chain of tight \pls/-reductions that leads to Network Congestion Games starts from \lmaxcut with maximum degree $5$, and includes Congestion Games. As we discuss after the proof of \cref{thm:local_max-k-cut}, in such \lmaxcut instances there is a starting configuration from which all improvement sequences of standard local search have exponential length in the worst case. Therefore, shortest-path dynamics on our family of network congestion games -- under \emph{any} pivoting rule -- 
need exponentially many iterations in the worst case. In contrast, our \cref{th:smoothed-bound-dynamic-network-games} and \cref{th:compact-netc-smoothed-poly} show that under smoothness, even for significantly wider instance families that include these problematic cases, shortest-path dynamics terminate after polynomially many steps in expectation.

\section{Further Applications}
\label{sec:sin-step-apps}
For the remainder of the paper,
we will apply our main tool from \cref{sec:sin-step} to directly get bounds on the expected number of iteration steps for local search algorithms, for various well-known local optimization problems. It is important to emphasize that we will be looking into \pls/-complete problems, that is, problems that are the hardest ones in \pls/: unless $\pls/ = \fp/$, there is no (worst-case) polynomial time algorithm to find a local optimum of these problems. Furthermore, due to their reductions being tight \pls/-reductions (see \cite{Schaffer91} for the definition), their individual standard local search algorithms are unconditionally inefficient: in each of these problems there exists an initial configuration, such that if the standard local search algorithm starts from there, every sequence of improving moves until termination is exponentially long.

All the proofs follow a similar structure:
for each optimization problem, we first show that it
can be formulated as a CLO problem, establishing one-to-one correspondence between feasible solutions and configurations, respecting the intended neighbourhood structure and the parameters to be perturbed under smoothness; then, we apply \cref{thm:single-step-bound} to the smoothed CLO problem, via an appropriate covering, to bound the expected running time. 
Coming up with the right choice for a covering is the technical part of our proofs.

Many of the following problems have a parameterized neighbourhood (\kopt, \kflip, etc.).
Their neighbourhood size usually increases exponentially in this parameter and therefore it is often assumed as being constant: this implies efficiently searchable neighbourhoods and, importantly, membership in \pls/.
In principle, however, it might be of interest to also consider instance-dependent neighbourhood sizes, either in terms of superpolynomial local improvements or for special types of instances where, e.g., neighbourhoods for logarithmically large values of $k$ might actually be efficiently searchable for the specific problem at hand.
For that reason, we make a point for all our running-time bounds to \emph{not} hide any dependencies on these parameters in the $\landau{\cdot}$-notation.
In that way, even when the bound depends on $k$ in the exponent, one can readily deduce a bound on the expected number of local-search iterations, beyond constant $k$; in particular, for polylogarithmic $k$ we derive quasipolynomial bounds in the input size.

A technical remark is in order here: in the following, we will not explicitly differentiate between maximization and minimization problems, as \cref{thm:single-step-bound} can be applied analogously in either case.
In order to appeal to intuition though, for maximization problems we will look at the configuration difference $\vecc s_I'-\vecc s_I$, instead of $\vecc s_I-\vecc s_I'$, when discussing the cover properties (see~\cref{sec:sin-step}), as this corresponds to the expression giving a positive improvement for a transition $(\vecc s, \vecc s')$.

\subsection{\texorpdfstring{\tsp}{TSP} with the \texorpdfstring{\kopt}{k-Opt} Heuristic}
\label{sec:TSP-k-opt}

The Travelling Salesman problem (\tsp) is one of the best studied problems in combinatorial optimization~\parencite{10.2307/j.ctt7t8kc,10.2307/j.ctt7s8xg}.
The \tsp is the problem of finding a Hamiltonian cycle -- also called \emph{tour} in this setting -- with minimum sum of edge weights in a weighted undirected graph.
This problem is \np/-hard as a global optimization problem, which is a standard result in classical complexity theory \cite{GareyJohnson:NPcompleteness}.
One of the most commonly used methods to approximate the solution is the \kopt heuristic, where the neighbourhood is induced by replacing at most $k$ edges of the current tour on $n$ vertices with another set of edges, resulting in a new tour.
Despite the worst-case results on its approximation ratio, such as the tight $\lantheta{n^{1/k}}$ ratio in Metric TSP for any constant $k$ due to \cite{Zhong:TSPapprox}, more promising probabilistic observations exist and this heuristic seems to perform well in practice, see e.g.\ \cite{GutinPunnen:TSP, Johnson97, EngelsManthey:TSP2OptAvg}.

\textcite{Krentel:structurelocalopt} settled that the local optimization problem of \tspkopt is \pls/-complete, for a sufficiently large constant $k$, as well as the existence of instances that require exponentially many iterations of local improvements.
\textcite{Klauck96} improves this result to a tight \pls/-reduction to \tspkopt.
\textcite{Englert14} constructed two-dimensional Euclidean TSP instances in which \IIopt needs exponentially many iterations, under the $L_p$ metric for any $p >0$.
Furthermore, \textcite{CKT:kOpt} provide instances where \kopt takes exponentially long to converge, for every $k > 2$.
Note that these explicit constructions provide \emph{a} path of exponential length, but not \emph{all} paths starting at the same configuration have this property.

As for the smoothed complexity of \tsp, \cite{Englert14} established polynomial smoothed running time for the \IIopt heuristic for \tsp in the input model where the graph nodes are given as points in $\R^d$, drawn from independent distributions over $[0, 1]^d$ with density bounded by $\phi$, and their pair-wise distances follow the $L_p$ metric.
This result has later been improved in \cites{MantheyVeenstra:2opt}{MantheyRhijn:2opt} for the special case of smoothness by Gaussian noise.

Furthermore, in \cite{Englert_2016} a different model of smoothness was discussed for \IIopt in the general \tsp (on undirected graphs), where the cost of each edge is drawn from independent distributions, and a polynomial bound on the smoothed running time was established.
Our analysis extends the polynomial smoothed complexity of the \kopt heuristic for \tsp from $k=2$ to arbitrary $k \geq 3$.
Using distributions that are independent of the (undirected) edges reflects the application better, where we can think of the location of each city as fixed, but the travelling times between them can vary due to traffic, which is uncorrelated between different cities.
Our analysis covers the case of \tsp on a directed graph, known as \emph{Asymmetric} TSP (\atsp)~\parencite{KanellakisPapadimitriou:ATSP,GutinPunnen:TSP,Michiels2007,STV:ATSP_constantapprox}, which has not been studied before in the context of smoothed analysis.
Under this asymmetric model, the travelling times in different directions, between the same two cities, can differ arbitrarily.

When discussed in the context of a local search heuristic for the global optimization problem, the \kopt heuristic also includes every \IIopt improvement, and is therefore often considered a stronger method (which comes at the price of higher per-step complexity).
In terms of the running time of local search, however, the methods are not as easily comparable:
every solution of the \kopt heuristic also provides a solution for \IIopt on the same instance, but it might take more steps to actually terminate the local search procedure for \kopt due to the larger neighbourhood;
a solution for \kopt not only needs to be at least as good as every possible exchange of $2$ edges, but is compared to every \kopt move and local search may terminate only a long time after the first local optimum sufficient for the \IIopt heuristic has been encountered.

\begin{definition}[\tspkopt]
	By \tspkopt we denote the following local optimization problem. Let $G=(V, E)$ be a weighted undirected graph. Find a locally minimal tour $\mathcal{C}$ on $G$ with respect to the \kopt neighbourhood.
    The cost of a tour is given by the sum of its edge weights.
	The \kopt neighbourhood of $\mathcal{C}$ is the set of all tours that can be generated from $\mathcal{C}$ by replacing some set of $r \in \{ 2, 3, \dots, k\}$ edges from it with a different set of $r$ edges in the graph. 
    A solution is a tour that does not have a neighbour with smaller cost.
\end{definition}

Analogously, we define the problem \atspkopt on directed graphs.
We will use the notion of \emph{edges} to refer to both, edges in an undirected graph and arcs in a directed graph.

\begin{definition}[\atspkopt]
	Given a weighted directed graph $G=(V, E)$, \atspkopt is the problem of finding a locally minimal directed tour (a directed Hamiltonian cycle) under the \kopt neighbourhood.
    The cost of a directed tour is given by the sum of its edge weights.
    The neighbours of a given tour are all directed tours that can be generated by replacing at most $k$ edges from it by other edges.
\end{definition}

Note that to find a neighbour of $\mathcal{C}$ using the \kopt heuristic we are removing any set of $2 \leq r \leq k$ edges of a tour (for which there are $\binom{\card{V}}{r}$ possible choices of edges for each $r$, as $\mathcal{C}$ consists of $\card{V}$ many edges). Let us fix a set of $r$ edges from the current tour to be removed, thus splitting the tour into $r$ disjoint paths (or isolated vertices, which we consider as degenerate paths).
Then, for each such path $p$, its endpoints can only be reconnected via an edge to an endpoint of the other $r-1$ paths to form a new tour; in other words, at most $2(r-1)\leq 2r$ vertices need to be considered for each of the $2r$ endpoints.
Any full reconnection of the set of endpoint vertices can be thought of as a permutation of this set, so there are at most $(2r)!$ possible new tours.
In total, for each fixed value of $r$ we thus get at most $\binom{|V|}{r} \cdot (2r)!$ neighbouring tours, resulting in at most $(k-1) \cdot \binom{|V|}{k} \cdot (2k)! \in \landau{|V|^{k} (2k+1)!}$ many neighbours for any given tour.
Therefore, when $k$ is constant, the \kopt-neighbourhood of any configuration is efficiently searchable.
The same bound holds for \kopt of \atsp, where fewer of the reconnections outlined above are feasible.
In particular, for \atsp a valid heuristic is only given for $k\geq 3$, as the unique new (undirected) tour given by any application of the \IIopt procedure does not form a valid directed cycle.

Similarly to \cite{Englert_2016}, we assume that an initial tour of the graph is given or can be queried by an oracle, because even the problem of deciding whether a Hamiltonian cycle exists in a graph is \np/-complete.

For the smoothed versions of \tspkopt and \atspkopt, the edge weights $w_e$ are drawn from independent distributions $f_e: [0, 1] \map [0, \phi]$ each.
Our analysis only makes use of the concept of a subset of edges that forms a cycle and thus applies to both, undirected and directed edges.
Note that \tsp can be reduced to \atsp in standard complexity by choosing symmetric edge weights, which introduces some overhead by doubling the representation size, as every edge is represented by a pair of opposing edges.
However, under smoothed analysis these two representations of the \tsp differ substantially, because both the opposing edges' weights appear in the input of the problem, they experience independent noise in the \atsp, contrary to the single edge's weight that appears in the input of the original \tsp.

\begin{theorem}
\label{thm: smoothed-tsp}
	The problem of computing a local optimum of the \tsp or \atsp with the \kopt heuristic (for constant $k$), on a graph with $m$ edges, whose weights $w_e, e\in E$ are drawn independently from distributions with densities $f_e:[0, 1]\to [0, \phi]$, has polynomial smoothed complexity, i.e.\ the expected number of iterations of the local search algorithm can be bounded by $\landau{4^{k^2} m^{k+2} \phi}$.
\end{theorem}

\begin{proof}
	We define a binary ($M=1$) CLO problem. Each component of the cost-part corresponds to an edge in the graph, i.e.\ $\nu = m$ and the cost part is $\vecc s^\costv = (s_e)_{e\in E}$.
    A tour $\mathcal{C}$ is mapped to a CLO configuration $\vecc s(\mathcal{C})$ by indicating which edges are being used in it:
	\begin{equation*}
		s_e(\mathcal{C}) = \indicator{e\in \mathcal{C}}.
	\end{equation*}
	The cost coefficients are given by $c_e = w_e$.
	By this model, the one-to-one correspondence between the tours and the configurations of the CLO problem is immediate, as either is a representation of a subset of edges and their costs coincide.
	The neighbourhood relation for the CLO problem is derived directly from the neighbouring tours in the \tspkopt instance.
	Because the CLO cost coefficients are equal to the edge weights, the smoothed costs of the \tsp instance also translate directly to the smoothed CLO model in \cref{thm:single-step-bound}.
	We establish the remaining properties of the theorem (a covering $(\mathcal{E}, \mathcal{I})$ and its parameters $\lambda = m^{k}, \mu \in \{1, 2\}, \beta = 2k^2$) in the following.

	For the coordinate cover, we can simply use the singleton coordinate clusters: $ \mathcal{I} = \{\{e\}\fwh e\in E\} $.
	To cover the neighbourhood graph, we define transition clusters $E(e_1,\dots, e_k)$ for every choice of $e_i\in E$.
    Each cluster captures the set of transitions where the given edges $e_i$ are removed from the tour and replaced by (equally many) other edges to form a new tour.
	Note that we allow for some edges to coincide, $e_i = e_j$, to account for transitions where fewer than $k$ edges are replaced.
    Therefore, the number of clusters is $\lambda = m^k$.

	Next, we investigate for every transition cluster $E(e_1,\dots, e_k)$ which edges can possibly replace the removed ones to form a new tour.
    This is equivalent to determining the core of each transition cluster and allows us to settle $\beta$.
	Note that this is \emph{not} the same as the number of neighbours of a given tour.
	When a subset of $k$ edges is removed from a tour, there remains a set of $k$ paths.
	Note that some of these paths might be degenerate in the sense that they only contain a single vertex, which occurs when two adjacent $e_i, e_j$ are removed from the tour.
	We call these vertices \emph{isolated}.
	Denote by $k_1$ the number of isolated vertices and by $k_2$ the number of proper paths that remain, thus $k = k_1 + k_2$.
	In order to form a new tour, each of the endpoints and isolated vertices needs to be reconnected.
	Clearly, we need to reconnect the vertices within this set, as every other vertex in the graph still has $2$ incident edges.
	Additionally, we also do not reconnect them to their previous neighbour (or neighbours in the case of an isolated vertex).
	For every end point of a non-degenerate path, the edge that leads to the other end point of the same path also cannot be used, as this would create a short cycle.
	In total, we thus have $k_1 + 2k_2$ vertices that need to be reconnected and each of them has at most $k_1+2k_2-3$ possible edges to reconnect them properly.
	Note that due to symmetry, every edge we count appears twice in the analysis, so we can bound the number of possible edges that can be added to form a new tour by
		\begin{equation*}
			\frac{1}{2}(k_1+2k_2)(k_1+2k_2-3) \leq 2k^2-k,
		\end{equation*}
	and together with the $k$ edges that are being removed, at most $\beta = 2k^2$ different edges can be affected by any transition cluster.
    This argument, and therefore also the bound, still holds when a set of $r < k$ edges are replaced from a tour.
	
	Regarding $\mu$, we need to bound the number of possible values of $s_{e}-s_{e}'$ for any transition $(\vecc s, \vecc s')$ within a given transition cluster $E(e_1,\dots, e_k)$.
    For edges $e \in \{e_1,\dots, e_k\}$ that are removed, we observe $s_e - s_e' = 1$, whereas for any of the previously discussed edges $e$ that can possibly be added to form a new tour, we have $s_{e}-s_{e}'\in \{0, -1\}$, depending on whether it is actually used in the new tour represented by the configuration $\vecc s'$ or not.
    Therefore we conclude with $\mu = 2$ and apply \cref{thm:single-step-bound} with the derived parameters to get a bound on the expected number of iterations for \tspkopt of
		\begin{equation*}
			3 \cdot 2^{2k^2} m^k \cdot m^2 \cdot \phi = 3\cdot 4^{k^2} m^{k+2} \phi.
		\end{equation*}
\end{proof}

\begin{remark}
	A slightly more refined analysis allows to improve the constant for \tspIIopt:
	When removing two edges from a given tour, there exists only one other possibility to construct a new tour, as from the three neighbouring vertices available for reconnection, one is the previous neighbour and another is the other end of the same path.
	Therefore, we can choose $\mu = 1$ instead, as we know precisely which edges are being added to the tour, and get a bound on the expected number of iterations of $3m^{k+2}\phi$ instead.
\end{remark}

\subsection{Maximum Satisfiability}
\label{sec:max-sat}
\maxsat is the optimization problem of finding an assignment of a CNF formula that maximizes the weight of satisfied clauses.
This problem is \np/-hard and therefore heuristics are often used to approximate solutions in practice.
The \kflip local search algorithm is one of the most commonly used heuristics, both by itself or as part of a more general heuristic, see e.g.\ \cites{YagiuraIbaraki:kFlip_MaxSat,Szeider:kFlipSAT} for theoretical discussions on it.
Using this heuristic, we can define the local optimization problem \lmaxsatk:

\begin{definition}
	\lmaxsatk is the problem of finding a locally maximal truth assignment of a CNF formula with clauses $\mathcal{C} = \{C_1,\dots, C_m\}$ over binary variables $\mathcal{X} = \{x_1,\dots, x_n\}$.
    There are weights $w_1,\dots, w_m$ associated with the clauses and the weight of an assignment is given by the sum of weights of the satisfied clauses.
	The neighbours of a truth assignment are given by the \kflip neighbourhood, which are all assignments generated by changing the truth value of up to $k$ variables.
\end{definition}

When $k$ is a constant, this neighbourhood has polynomial size and is therefore efficiently searchable.
\lmaxsatI has been shown to be \pls/-complete in a series of works \cite{Krentel:MaxSAT, Krentel:structurelocalopt, Klauck96}, even for the case where the size of the clauses (i.e.\ the number of literals) and the number of occurrences of any single variable across all clauses are constant. These hardness results were proven via tight \pls/-reductions (see \cref{sec:model}), starting from the maximization or its equivalent minimization problem \circuitflip \parencite{Johnson:1988aa}, for which there exists a configuration with only exponentially long sequences to solutions. This implies directly that the \lmaxsatI instances we study have also such a configuration whose improvement sequences are all exponentially long under the standard local search algorithm. 

In \cite{Brandts-Longtin:smoothed_maxksat}, a smoothed version of \lmaxsatI  has been analysed, by adapting the method for \lmaxcut from \cite{Angel:2017aa}, claiming a bound on the expected running time of $\landau{n^{5\kappa+5}}$, where $\kappa$ is a bound on the size of the clauses. Note, however, that similarly to how the original method relies on the complete graph for the result, this result applies only to the complete formula, in which every possible clause of given fixed size is present.

Our contribution shows polynomial smoothed time for \lmaxsatI with bounded variable occurrence, but without any other restrictions on the formula.

\begin{theorem}
	\label{thm:smoothed_max-sat}
	The expected number of iterations of smoothed \lmaxsatI, where each clause weight $w_i$ is drawn from independent distributions with densities $f_i:[0, 1]\map [0, \phi]$, is bounded by $\landau{2^{B}n m^2\phi}$, where $B$ is the maximum number of different clauses that any variable occurs in.
\end{theorem}

\begin{proof}
	\lmaxsatI can be stated as a binary ($M=1$) CLO problem.
	We identify the cost part with the clauses ($\nu = m$) by defining a variable $s_i$ for every clause $C_i$ with cost coefficient $c_i = w_i$.
    Denote by $a:\mathcal{X}\map \{0, 1\}$ a truth assignment and by $\vecc s(a)$ its induced CLO configuration, defined as follows:
	In the cost part, we assign $s_i(a) = 1$ if the associated clause $C_i$ is satisfied, and $s_i(a) = 0$ otherwise.
    Thus, the cost of the CLO configuration $\vecc s(a)$ is equal to the weight of the assignment $a$ in the \lmaxsatI problem.
	Because this does not necessarily represent all the information of the underlying truth assignment, we use the non-cost part of $\vecc s(a)$ to represent $a$ explicitly as a binary string of length $\bar \nu = n$.
	The CLO neighbourhood can now be derived from the assignments in the \Iflip neighbourhood and be searched efficiently.
	For the smoothed problems, the random distributions of the weights in the smoothed \lmaxsatI instance are mapped identically to the distributions of the costs in the smoothed CLO problem.
	In the following, we apply \cref{thm:single-step-bound} to achieve a bound on the smoothed running time.

	To do so, we first specify the singleton coordinate cover $\mathcal{I} = \{\{i\}, i = 1,\dots, m \}$ of the indices of the cost part.
	The transition cover $\mathcal{E}$ comprises clusters $E(x_j, b)$ for each variable $ x_j $ and $b\in \{0, 1\}$.
    The cluster $E(x_j, b)$ collects all transitions where the truth assignment of variable $x_j$ flips from $b-1$ to value $b$, regardless of the assignment of the remaining variables.
	Formally, we can cover the edges of the neighbourhood graph using $\lambda = 2n$ many transition clusters by
	\begin{align*}
		\mathcal{E} &= \{E(x_j, b)\fwh j\in \intupto{n}, b\in \{0, 1\}]\}, \\
        \text{where} \quad E(x_j, b) &= \{(\vecc s(1-b, a_{-j}), \vecc s(b, a_{-j})) \fwh a:\mathcal{X}\map \{0, 1\} \text{ arbitrary}\},
	\end{align*}
    and the notation $(b, a_{-j})$ refers to the truth assignment that maps variable $x_j$ to value $b$ and all other variables according to assignment $a$.
    A change in the assignment of variable $x_j$ can only affect the satisfaction of clauses that contain the literal $x_j$ or $\bar x_j$.
    Therefore, we know that $\core{E(x_j, b)}$ comprises precisely the indices of those clauses.
	By assumption, there are at most $B$ such clauses for every $x_j$, and we have established $\beta = B$.

	Finally, we specify $\mu = 2$ by inspecting, for fixed $(x_j, b)$, the entries $s_i, s_i'$ of a transition $(\vecc s, \vecc s')\in E(x_j, b)$:
    For $b = 1$, when clause $C_i$ contains the literal $x_j$, it may either become satisfied by the changed assignment (i.e.\ $s_i = 0$ and $s_i' = 1$ after the transition) or remain satisfied ($s_i=s_i'=1$).
    For a clause that contains the negative literal $\bar x_j$, either it becomes unsatisfied ($s_i=1$ and $s_i'=0$) or remains satisfied ($s_i=s_i=1$).
    Analogously, the same observation holds with exchanged roles in the literals for $b = 0$.
	In any case, the difference $s_i' - s_i$ can only take two possible values for every transition in $E(x_j, b)$ and $\mu = 2$ is justified.

	Thus, by application of \cref{thm:single-step-bound}, the expected number of steps is bounded by $ 3\cdot 2^B2n \cdot m^2 \cdot \phi$.
\end{proof}

Our proof method can easily be adapted to account for more general \kflip neighbourhoods as well, where at most $k$ variables can change their values simultaneously.
As the \kflip neighbourhood for $k > 1$ also includes every \Iflip move, every local optimum to the former is also optimal for the latter; thus \lmaxsatk is still \pls/-hard.

\begin{theorem}
	\label{thm:smoothed_max-sat_kflip}
	The expected number of iterations of smoothed \lmaxsatk, where each clause weight $w_i$ is drawn from independent distributions with densities $f_i:[0, 1]\map [0, \phi]$ is bounded by $\landau{3^{k B} n^{k} m^2\phi }$, where $B$ is the maximum number of clauses that a variable occurs in.
\end{theorem}
\begin{proof}
	The proof is analogous to that of \cref{thm:smoothed_max-sat} and we only state the necessary changes to deal with the \kflip neighbourhood instead of \Iflip:
	
	For the transition cover, we now use clusters $E(x_{j_1},\dots, x_{j_k})$, $x_{j_\ell}\in \mathcal{X}$.
    Each cluster collects all transitions where all given variables change their truth assignment, regardless of their initial assignment or the assignment of the remaining variables.
	We do not require all $x_{j_\ell}$ to be different, so as to account for neighbouring assignments where fewer than $k$ variables flip.
	Therefore, we can now bound the number of transition clusters by $\lambda = n^{k}$.
	Every variable occurs in at most $B$ clauses, therefore the transitions in $E(x_{j_1},\dots, x_{j_k})$ can affect at most $\beta = k B$ cost components.

	Finally, because the CLO problem is binary, we know that within every component $s_i$ the transition difference $s_i' - s_i$ can take at most the three values $ \{-1, 0, 1\}$, which gives $\mu = 3$.
	Applying \cref{thm:single-step-bound} yields a bound for the expected number of steps of $3\cdot 3^{k B} n^{k}\cdot m^2 \cdot \phi $.
\end{proof}

\subsection{Maximum Cut}
\label{sec:localmaxcut}
\lmaxcut is among the problems in \pls/ that has been most extensively studied in the context of smoothed analysis.
The result that our method generalizes is the polynomial smoothed complexity for graphs of up to logarithmic degree due to \cite{ElsasserT11}, who studied a smoothness model with Gaussian noise.
By application of our black-box tool, we are able to extend this result to the smoothness model of arbitrary distributions with density bounded by $\phi$ and also prove it directly for the expected running time instead of the weaker \emph{probably smoothed complexity} they use.
Further polynomial results under smoothness exist for \lmaxcut on complete graphs, both for a version of ``metric'' \lmaxcut in \cite{Etscheid:2015} and for the noise model with general edge weights in \cites{Angel:2017aa,Bibak21}.
For general graphs on the other hand, quasipolynomial smoothed complexity was shown in \cite{Etscheid17} and further improved by \cite{Chen20}, as well as parametrized by the arboricity of the graph in \cite{Schwartzman:SparseSmoothedLocalMaxCut}.

It is important to note that none of the aforementioned graph models is stronger than any other, and that is due to a delicate difference between worst-case and smoothed analysis in the equivalence between input representations. For example, in worst-case analysis, having a general graph is equivalent to having the complete graph as input to \lmaxcut, in the sense that the complexity of the problem remains the same (up to a constant factor): any general graph can be turned into a complete graph with zero weights for the non-existing edges. On the contrary, in smoothed analysis this is not the case, since an instance in which an edge does not exist cannot be straightforwardly turned into an -- in terms of smoothed complexity -- equivalent instance by setting that edge's weight to zero: as soon as the zero weight appears in the input it is supposed to undergo smoothness, something that we wish to disallow for an essentially non-existing edge.

Our method allows us to easily extend the analysis to \lmaxkcut, a generalization of \lmaxcut that considers partitions of the vertex set into $k$ blocks instead of only $2$.
The \flip neighbourhood for \lmaxkcut is similarly given by a single vertex changing from one of the cut-sets to another.
\lmaxkcut under this neighbourhood has been analysed by \cite{Bibak21}, who established polynomial smoothed complexity for the complete graph and a bound of $\landau{\phi n^{2(2k-1)k\log(kn)+3+\eta}}$, for every $\eta > 0$, in the general case.
Our method approaches this problem from a different perspective through the parameterization on the degree of the graph and thereby improves on this result for graphs of up to logarithmic degree with a bound that does not depend on $k$ at all.

We first analyse the case of $k=2$, namely, the classical \lmaxcut problem.

\begin{definition}
\label{def:lmaxcut}
	\lmaxcut is the problem of finding a locally maximal cut of a weighted, undirected graph $G = (V, E), n = \card{V}, m = \card{E}$ with weights $w_e$ for $e\in E$.
	A cut is a partition $(V_1, V_2)$ of $V$ and the neighbours of a cut are given by the \flip neighbourhood, in which any single vertex can move from one partition block to the other.
	The weight of a cut is given by the sum of all edge weights across the cut, i.e.\
	\[ W(V_1, V_2) = \sum_{e\in E(V_1, V_2)}w_{e}, \quad \text{where}\quad E(V_1, V_2) = \{uv \in E \fwh u\in V_1, v\in V_2\} \]
    is the cut-set of the cut $(V_1, V_2)$ and we call the edges $e\in E(V_1, V_2)$ cut-edges. 
\end{definition}
\lmaxcut is \pls/-complete, as shown in \cite{Schaffer91};
this is still the case for graphs whose degree is bounded as low as five \cite{ElsasserT11}.
Independently of this reduction, \cite{MT10} have shown that even when the degree is bounded by four, there exist instances and initial cuts such that every local search sequence is exponentially long.

\begin{theorem}
	\label{thm:local_max-cut}
    The expected number of iterations of smoothed \lmaxcut with weights $w_e$ drawn from independent distributions $f_e:[-1, 1]\map [0, \phi]$ is bounded by $\landau{2^{\Delta(G)} n m^2\phi }$, where $\Delta(G)$ denotes the maximum degree of the graph.
\end{theorem}

\begin{proof}
	To formulate this problem as a CLO problem, we define a binary ($M=1$) variable $s_e$ in the cost-part for every edge $e\in E$ in the graph of the \lmaxcut instance, thus $\nu = m$.
    For the CLO configuration $\vecc s(V_1, V_2)$ induced by a cut $(V_1, V_2)$, we set $s_e = 1$ if and only if the edge $e$ is a cut-edge and $0$ otherwise.
	Even though the vertex partition can be efficiently recovered from this representation (up to the arbitrary assignment of connected components to either part of the partition), we disregard the discussion and represent the partition explicitly in the non-cost part instead, using a binary string of length $\bar \nu = n$ that indicates which vertices belong to $V_1$.
	The edge weights are translated directly to the cost coefficients by $c_e = w_e$ (which also makes smoothed \lmaxcut correspond directly to the smoothed CLO problem) and the neighbourhood relation is derived from the original problem as well to be in one-to-one correspondence.
	This neighbourhood is efficiently searchable and the cost of the CLO configuration induced by a cut is equal to the weight of the cut in \lmaxcut.

    Now, we apply \cref{thm:single-step-bound} on the smoothed CLO problem defined above.
    We derive the remaining properties in the following.
    The covering uses a singleton coordinate cover $ \mathcal{I} = \{\{e\}, e\in E \} $.
    For the transition cover, $\mathcal{E}$ comprises clusters $E(v)$ for every vertex $v\in V$, capturing all transitions where the vertex $v$ moves from one side to the other, and thus $\lambda = n$.
    When vertex $v$ moves, all edges that are not incident to $v$ either stay cut-edges or non-cut-edges, so the only components possibly affected are those of edges to neighbours of $v$.
    Therefore, we have $\beta = \Delta(G)$.
	Moreover, all edges that are incident to $v$ do in fact change their configuration: Every edge to a neighbour of $v$ in the same partition block as $v$ before the transition is a cut-edge afterwards and also the other way around.
	Alas, we can choose $\mu = 2$, as the difference $s_e'-s_e$ for every affected edge $e$ is always $\pm 1$.
    Thus, \cref{thm:single-step-bound} provides a bound on the expected number of iterations of $3\cdot 2^{\Delta(G)} n \cdot m^2 \cdot \phi$.
\end{proof}

The result also applies with little modification to \lmaxkcut. 

\begin{definition}
	\lmaxkcut is the problem of finding a locally maximal cut of a weighted, undirected graph $G = (V, E), n = \card{V}, m = \card{E}$ with weights $w_e$ for $e\in E$.
	A cut is a partition $(V_1, \dots, V_k)$ of $V$ and the neighbours of a cut are given by the \flip neighbourhood, in which a single vertex can move from one partition block to another.
	The weight of a cut is given by the sum of all edge weights across the $k$-cut, i.e.\
    \[ W(V_1, \dots, V_k) = \sum_{e\in E'} w_{e}, \quad \text{where}\quad E' = \bigcup_{i < j \leq k}E(V_i, V_j), \]
	and we call the edges $e\in E'$ cut-edges.
\end{definition}

This generalization only has a minor impact on our bound, where the exponential part of the bound increases from $2^{\Delta(G)}$ to $3^{\Delta(G)}$.
Note that compared to the result by \cite{Bibak21}, the parameter $k$ does not appear in our bound.

\begin{theorem}
    \label{thm:local_max-k-cut}
    The expected number of iterations of smoothed \lmaxkcut with densities bounded by $\phi$ is bounded by $\landau{3^{\Delta(G)} n m^2 \phi }$, where $\Delta(G)$ denotes the maximum degree of the graph.
\end{theorem}

\begin{proof}
	The proof is analogous to the proof of \cref{thm:local_max-cut}, except for the changes outlined in the following.
	The non-cost part now needs to encode which of the $k$ partition blocks every vertex belongs to.
	It is straightforward to do this using a binary vector of polynomial size.
	The transition clusters remain structurally the same as $E(v), v\in V$, and each of these sets covers all transitions of a vertex between any of the partition blocks.
	
    The argument, that only the incident edges are affected still holds, but for $k>2$ not all neighbouring edges of a moving vertex must change between being a cut-edge or not.
	Therefore, we can bound $\mu = 3$ by observing that in this binary CLO problem, the difference $s_e'-s_e$ for any cost component can only take a value from $\{-1, 0, 1\}$.
	Thus we conclude with an application of \cref{thm:single-step-bound} to achieve a bound on the expected number of iterations of $3\cdot 3^{\Delta(G)} n \cdot m^2 \cdot \phi$.
\end{proof}

The family of instances for which we prove smoothed efficiency are \pls/-hard. In particular, in \cref{sec:maxkcut_pls} we prove that \lmaxkcut is \pls/-hard for any $k \geq 2$ and even when the input graph has maximum degree $k+3$. \lmaxcut is its special case $k=2$, which was already known by \cite{ElsasserT11} to be \pls/-hard. This means that, unless $\pls/=\fp/$, there is no worst-case polynomial time algorithm to solve \lmaxkcut for any $k \geq 2$.

Apart from the aforementioned conditional intractability, there are results on the unconditional worst-case performance of the standard local search algorithms for \lmaxcut. \textcite{MT10} constructed instances for which there exists a configuration (i.e., a cut) such that, if the standard local search algorithm starts from there, it will need exponential time, regardless of what path of improvements it chooses.
\begin{theorem}[\cite{MT10}, Theorem~5]
    For every $n \in \N^*$, there is a graph on $O(n)$ vertices with maximum degree four, and an initial cut $\vecc s_0$, such that every sequence of improving \flip moves starting from $\vecc s_0$ and ending in a local optimum contains a vertex that is moved $2^{n+1}$ times.
\end{theorem}
In contrast, our \cref{thm:local_max-cut} shows that, under smoothness, the expected number of iterations of standard local search is polynomial even for graphs with logarithmically large degree. Furthermore, in \cref{thm:local_max-k-cut} we prove that this still holds for the general problem \lmaxkcut with any $k \geq 2$.

\subsubsection{Equivalent Problems}
The following problems are straightforward equivalent reformulations of \lmaxkcut, which allows us to derive analogous smoothness results from \cref{thm:local_max-cut} and \cref{thm:local_max-k-cut} by just restating the problem as its \lmaxkcut counterpart instead of applying the CLO results to them directly.

\paragraph{Party Affiliation Games}
In these games, as outlined by e.g.\ \cite{BCK:PartyAffiliation_approximatePNE}, players correspond to vertices in the graph $G=(V, E_f\cup E_e)$ and choose one of two sides of a cut (one of two \emph{parties}).
Each edge in the weighted game graph is considered either a ``friendly'' or an ``enemy'' edge, denoted by $e \in E_f$ or $e\in E_e$.
The weights $w_e$ are non-negative.
The utility of a player is given as the sum of edge weights of friendly edges within the same party and of enemy edges across the two parties.
If we denote the parties by $V_1, V_2$ and use the cut notation from above, the game admits a potential function for a strategy profile $\vecc{s}$ as
  \[ \Phi(\vecc{s}) = \sum_{e\in E_f \setminus E_f(V_1, V_2)} w_{e} + \sum_{e \in E_e(V_1, V_2)} w_{e}. \]
This puts the problem of finding a PNE in this game into \pls/ by using best-response dynamics.
By subtracting $\sum_{e\in E_f}w_e$ from this function, we achieve another potential
  \[ \Phi'(\vecc s) = -\sum_{e\in E_f(V_1, V_2)} w_e + \sum_{e \in E_e(V_1, V_2)} w_{e}. \]
It is immediate that by changing the signs of the weights of all edges in $E_f$, the resulting problem can be considered an instance of \lmaxcut.
Therefore, we state the following result.
\begin{corollary}
    The expected number of steps in a best-response dynamic of a smoothed party affiliation game with weights $w_e$ drawn from independent distributions $f_e:[0, 1]\map [0, \phi]$ is bounded by $\landau{2^{\Delta(G)} n m^2\phi }$, where $\Delta(G)$ denotes the maximum degree of the graph.
\end{corollary}

Note that the assumption of non-negative weights and the separation of the edge set into friendly and enemy edges simply corresponds to the sign of the weights in the corresponding \lmaxcut instance.
The \pls/-hardness for this problem is usually attributed to \cite{Schaffer91} due to the equivalence with \lmaxcut.
Following the \pls/-hardness discussion for \lmaxcut above, we deduce that even for graphs with constantly bounded degree, party affiliation games are \pls/-hard.

\paragraph{Symmetric Additively Separable Hedonic Games}
Again, we consider the players as vertices in a weighted graph $G = (V, E)$.
In these games, players partition themselves into sets $V_1,\dots, V_n$, $n = \card{V}$, and achieve utility according to the sum of weights of edges within their respective party.
If we denote by $E(V_i)$ the set of edges of the subgraph induced by $V_i$, this game has a potential function for a strategy profile $\vecc{s}$ given by
  \[ \Phi(\vecc{s}) = \sum_{i\in \intupto{n}} \sum_{e\in E(V_i)} w_e. \]
Therefore, the problem \problemfont{NashStable} (see \cite{GairingSavani:hedonicgames}) of finding a PNE (also called a Nash-stable outcome in this setting due to the existence of alternative stability properties) for these games is in \pls/ using best-response dynamics.
By subtracting $\sum_{e\in E} w_e$ from this function and denoting $E' = \bigcup_{i<j\leq n} E(V_i, V_j)$, we achieve another potential given by
  \[ \Phi'(\vecc{s}) = -\sum_{e\in E'} w_e. \]
By changing the sign of all weights, this is precisely the potential function of the \lmaxkcut problem on the same graph. In other words, working on the same graph, asking to find a PNE in these games is equivalent to finding a solution to \lmaxkcut for $k=n$.
Thus, we state the following bound.
\begin{corollary}
	The expected number of steps in a best-response dynamic of a smoothed symmetric additively separable hedonic game with weights $w_e$ drawn from independent distributions $f_e:[-1, 1]\map [0, \phi]$ is bounded by $\landau{3^{\Delta(G)} n m^2\phi }$, where $\Delta(G)$ denotes the maximum degree of the graph.
\end{corollary}

While the discussion in \cref{sec:maxkcut_pls} can be used to confirm the \pls/-hardness result for this problem first given by \textcite{GairingSavani:hedonicgames}, neither suffices to immediately deduce hardness results for bounded degree graphs.
To our knowledge, the complexity in the bounded degree case is open.

\subsection{Network Coordination Games}
\label{sec:netcoord}
Network Coordination Games are a special class of potential games (see \cite{Monderer1996a}), where players want to coordinate with each other locally with respect to some graph structure.
The game is given by an undirected graph $G = (V, E), n = \card{V}, m = \card{E}$, where $V$ represents the set of players. Each player $v\in V$ chooses one of $k$ actions and plays a separate matrix coordination game with each of their neighbours simultaneously, i.e.\ the same action is chosen in each of these games.
For every edge $uv\in E$ of the graph, there is a $k\times k$ payoff matrix given that we denote by $A_{uv}$, according to which the players receive their partial payoffs of the game on this edge.
A player's total payoff is then given by the sum of payoffs over all the games on their incident edges; for a strategy profile $\vecc \sigma$, player $u$'s payoff is $\operatorname{payoff}_u(\sigma) = \sum_{v\in \Gamma(u)} A_{uv}(\sigma_u, \sigma_v)$, where $\Gamma(u)$ is the set of neighbours of $u$ in the graph $G$.
Note that we specify each matrix only once for each edge, so $A_{vu} = A_{uv}^{\top}$.

The potential structure of the game implies that, just like in congestion games, better-response dynamics converge to a PNE.
The dynamics are derived from the PNE definition; a single player can deviate unilaterally to increase their payoff.
The local search neighbourhood is thus given by all possible single player deviations and a solution is a local optimum of the potential with respect to this neighbourhood, or equivalently a PNE of the game. 
The potential function is given by
	\begin{equation*}
		\Phi(\vecc \sigma) = \sum_{uv\in E}A_{uv}(\sigma_u, \sigma_v) = \frac{1}{2}\sum_{v\in V}\operatorname{payoff}_v(\vecc \sigma) 
	\end{equation*}
where $\vecc \sigma$ denotes a strategy profile and $A_{uv}(i, j)$ denotes the entry at coordinates $(i, j)$ of the $k\times k$ payoff matrix in the game between $u$ and $v$.
The problem \netcoord of finding a PNE in a Network Coordination Games is \pls/-complete by a reduction from \lmaxcut due to \cite{CaiDaskalakis:multiplayergames}.

In \textcite{Boodaghians20}, the authors conducted smoothed analysis on \netcoord.
Their smoothness model draws each matrix entry for each of the games from independent distributions $f_{uv, i, j}:[-1, 1]\map [0, \phi]$, where $uv \in E$ denotes an edge in the graph and $i, j$ refer to the entry of the matrix $A_{uv}$ at coordinates $(i, j)$.
The authors have established that, in general graphs, the expected running time of a better-response algorithm under this smoothness model is bounded by $\phi \cdot (nk)^{\landau{k \log(nk)}}$.

We investigate this problem by introducing the degree of the graph as a new parameter, which allows us to establish polynomial smoothed complexity for graphs with constantly bounded degree.
This approach allows us to achieve a bound that only depends polynomially on $k$ and the graph size, and the only property that appears in its exponent is the degree of the graph.
For instances with either (a) constant degree and arbitrary $k$, or (b) logarithmic degree and constant $k$, we provide polynomial smoothed time bounds, which improves upon the bounds induced by~\cite{Boodaghians20} for those instances.

Note that the \pls/-hardness reduction for Network Coordination Games from \lmaxcut in \cite{CaiDaskalakis:multiplayergames} preserves the graph structure of the \lmaxcut instance in the game graph.
As \lmaxcut is \pls/-complete even for graphs of maximum degree $5$ (see \cref{sec:localmaxcut}), our result establishes another \pls/-complete problem class with polynomial smoothed complexity: \netcoord on graphs with bounded degree.
We remark here that the \pls/-hardness reduction in the aforementioned paper is tight, in the sense of \cite{Schaffer91} (see \cref{sec:model}). Since the standard local search of \lmaxcut with maximum degree $5$ has a starting configuration whose all improvement paths to solutions are exponentially long (see our discussion after the proof of \cref{thm:local_max-k-cut}), the same holds for the standard local search of the \netcoord instances we study here. Our smoothed analysis results show that, contrary to the worst-case performance, these algorithms terminate after polynomially many steps in expectation, even for logarithmically large degree of the game graph.

\begin{theorem}
	The expected running time of better-response dynamics for smoothed Network Coordination Games on a graph $G = (V, E)$ with costs drawn from independent distributions $f_{uv, i, j}:[-1, 1]\map [0, \phi]$ for every $uv\in E$, $i, j\in \intupto{k}$ is bounded by $\landau{k^{4\Delta(G)+6}nm^2 \phi}$, where $\Delta(G)$ is the degree of the game graph $G$.
\end{theorem}

\begin{proof}
	We define a binary CLO problem $(M=1)$ to represent \netcoord.
    The structure of the cost part corresponds to the different games played (the edge set $E$ of the game graph) and their matrix entries: for every edge $uv\in E$ and every entry of its payoff matrix, $i, j\in \intupto{k}$ we define a cost component $s_{uv, i, j}$ with its corresponding cost coefficient $c_{uv, i, j} = A_{uv}(i, j)$.
	The size of the cost part therefore is $\nu = mk^2$.

	From each strategy profile $\vecc \sigma$ of the game we can induce a CLO configuration $\vecc s(\vecc \sigma)$.
    This configuration follows an indicator structure within the components for a fixed edge $uv$: we set $s_{uv, i, j}(\vecc \sigma) = 1$ if and only if $\sigma_u = i \land \sigma_v = j$ and $s_{uv, i, j}(\vecc \sigma) = 0$ otherwise.
	Thus, within every set of the $k^2$ components of a fixed edge $uv$, there is exactly one entry of $1$ for every valid CLO configuration.
	The CLO cost of these induced CLO configurations $\vecc s(\vecc \sigma)$ is equal to the value of the potential $\Phi(\vecc \sigma)$ of \netcoord.
	From this representation, we can recover the strategy profile of the game in a straightforward way and therefore also compute all neighbouring configurations according to the better-response dynamics efficiently.

	Next, we establish all properties for an application of \cref{thm:single-step-bound}.
	The assumed smoothness model on the matrix entries translates directly to our smoothed CLO model, as every cost coefficient $c_{uv, i, j} $ is equal to a single random variable $A_{uv}(i, j)$ in smoothed \netcoord.
	In the following, we specify the covering $(\mathcal{E}, \mathcal{I})$ and its parameters $\lambda = nk(k-1), \mu = k^4, \beta = \Delta(G)$.
	Let
		\begin{equation*}
			\mathcal{I} = \{I_{uv}, uv\in E\},\quad \text{where} \quad I_{uv} = \{(uv, i, j)\fwh i, j\in \intupto{k}\}
		\end{equation*}
	is the coordinate cluster of all components associated with edge $uv$, i.e.\ the game between $u$ and $v$.
	As outlined above, every CLO configuration has exactly one entry of $1$ within the coordinates of any $I_{uv}$.
	Therefore, when considering a pair $(\vecc s, \vecc s')$ of configurations, the expression $ \vecc s_{I_{uv}}' - \vecc s_{I_{uv}}$ can take at most $\mu = k^4$ different values, as both $\vecc s_{I_{uv}}, \vecc s_{I_{uv}}'$ are one of $k^2$ configurations each.

	Next, we cover the neighbourhood graph by $\mathcal{E} = \{E(u, a, a')\fwh u\in V, a\neq a' \in \intupto{k}\}$.
    The transition cluster $E(u, a, a')$ contains all transitions where player $u$ deviates from action $a$ to action $a'$.
	There are $\lambda = nk(k-1)$ different transition clusters in this cover.
	Finally, we establish $\beta$:
	Every deviation of a player $u$ only affects the outcome of the games they participate in.
    These games are specified by the graph $G$, therefore each player participates in at most $\Delta(G)$ games.
    In the CLO problem, this corresponds to a change in the components related to the coordinate clusters $I_{uv}, v\in \Gamma(u)$, so every game a player participates in corresponds to a single coordinate cluster.
	This gives $\beta = \Delta(G)$, and we conclude the proof by applying \cref{thm:single-step-bound} to bound the expected number of iterations by
		\begin{equation*}
			3 \cdot (k^4)^{\Delta(G)} nk(k-1) \cdot (mk^2)^2 \cdot \phi \leq 3 k^{4\Delta(G)+6}nm^2 \phi.
		\end{equation*}
\end{proof}

\subsection{Weighted Set Problems}
\label{sec:weightedset}
In this section we discuss a selection of standard weighted set problems.
The global decision versions of these problems are \np/-complete even for the unweighted case, dating back to \cite{GareyJohnson:NPcompleteness}.
The discussion of the local optimization versions of the problems in our context is inspired by \cite{Dumrauf10}, who classified all of them as \pls/-hard by tight \pls/-reductions from Maximum Constraint Assignment \mcapqr, a problem shown to be \pls/-hard in \cite{DumraufMonien:PLS_MCA} for fixed parameters $(p,q,r) \in \{(3,2,3), (2,3,6), (6,2,2)\}$.
This section's focus is on several weighted set problems and not on \mcapqr. For our smoothed analysis results on the latter problem itself, see \cref{sec:MCA}.

It is important to highlight the fact that all the problems in this section have been shown in \cite{Dumrauf10} to be \pls/-hard via tight \pls/-reductions (see \cref{sec:model}). The starting problem in the chain of reductions is \circuitflip \parencite{Johnson:1988aa}. Since the standard local search algorithm for \circuitflip has exponentially long sequences in the worst case (see \cite[Lemma 12]{Yannakakis1997}), by our argumentation right after \cref{th:pls-hardness-restrained}, the standard local search algorithms for all the weighted set problems we study here have a configuration from which all improving sequences to solutions have exponential length. In contrast, we prove that under smoothed analysis, these algorithms terminate after polynomially many iterations in expectation.

\subsubsection{Weighted 3D-Matching}
An instance of \wIIIdmpq is given by an integer $n$ and a weight function $w:\intupto{n}^3\map \R$ on triples.
A 3-dimensional matching is a subset $S\subseteq \intupto{n}^3, \card{S} = n $ of triples $T = (T_1, T_2, T_3)$ that covers all values in every component; that is, 
for all $i \in \intupto{n}$ and $k\in \{1, 2, 3\}$ there exists a $T\in S$ such that $T_k = i$.
For the sake of intuition, we call the first component the \emph{boys}, the second the \emph{girls} and the third one \emph{homes} and think of the problem as finding an allocation that puts a boy and a girl together into their respective home for each of the triples.
The weight of a matching $S$ is given by the sum of weights of the triples in use, i.e.\ $W(S) = \sum_{T\in S}w(T)$.
A solution of \wIIIdmpq is a matching that is locally maximal in weight regarding the $(p, q)$-neighbourhood, defined as follows.

The neighbourhood of a matching $S$ is given by all matchings created by replacing up to $p$ triples by other (exactly $p$ many) triples.
Additionally, we require that the total number of boys or girls relocated into new homes is at most $q$.
Note that any triple that is being replaced moves the boy or the girl to a new home, but not necessarily both; the parameter $q$ can be thought of as restricting how often both of them may be relocated.
When $p$ is considered a constant, this neighbourhood can be efficiently searched. Furthermore, the second restriction $q$ does not appear in our bound, but we have stated the problem as it appears in \cite{Dumrauf10}, who established \pls/-completeness for all $p\geq 6, q\geq 12$.

\begin{theorem}
    The expected number of iterations of local search for smoothed \wIIIdmpq, where the weights $w(T)$ for each triple $T\in \intupto{n}^3$ are drawn from independent distributions with density $f_T:[0, 1]\map [0, \phi]$, is bounded by $\landau{n^{6p+6} \phi}$.
\end{theorem}

\begin{proof}
	We formulate this problem as a CLO problem while maintaining the weight structure of \wIIIdmpq in the CLO cost coefficients and then apply \cref{thm:single-step-bound}.
	Consider a binary ($M=1$) CLO problem where the cost part corresponds to the set of all triples, i.e.\ we choose $\nu = n^3$ and the cost variables are given by $s_T, T\in \intupto{n}^3$.
	The cost coefficients are set to the weights of the corresponding triples, $c_T = w(T)$.
	The CLO configuration $\vecc s(S)$ induced by a matching $S$ has $s_T(S) = 1$ whenever $T \in S$ and $s_T(S) = 0$ otherwise.
        Thus, the weight of a matching corresponds to the cost of the induced CLO configuration.
	The neighbourhood is derived directly and efficiently from the original instance.

	To apply \cref{thm:single-step-bound}, we establish a covering $(\mathcal{E}, \mathcal{I})$ and the parameters $\lambda = (n^3)^p, \beta = p, \mu = 1$ in the following.
	We start with the singleton coordinate cover, $\mathcal{I} = \{\{T\}, T\in \intupto{n}^3\}$.
	The neighbourhood graph is covered by transition clusters $E(T_1,\dots, T_p, U_1,\dots, U_p), T_i, U_i\in \intupto{n}^3$.
        Each of these clusters contains all valid transitions where the triples $T_i$ are replaced by $U_i$ in the matching, regardless of the remaining triples in the matching.
	There are at most $\lambda = (n^3)^{2p}$ such transition clusters when accounting every possible choice of the $2p$ triples.
	Note that some triples $T_i = U_j$ may coincide to account for fewer than $k$ tuple replacements in the solution, and that each permutation within the triples $T_i$ or $U_j$ refers to the same set of transitions.

        Now, we investigate the structure of this covering to determine the remaining parameters.
	Within the transition cluster $E(T_1,\dots, T_p, U_1,\dots, U_p)$, the only cost components that can change their assignment are the ones corresponding to a triple $T_i$ or $U_j$.
	Because we add and remove at most $p$ triples each, we therefore have $\beta = 2p$.
	More importantly, we know that for any choice of $(\vecc s, \vecc s')\in E(T_1,\dots, T_p, U_1,\dots, U_p)$ for a fixed transition cluster, the difference vector $\vecc s' - \vecc s$ is uniquely determined:
	Outside of the selected tuples, there is no change in the configuration.
	When $T_i \neq U_j$ or $U_i \neq T_j$ for all $j = 1,\dots, k$ (i.e.\ $T_i$ is actually removed from and $U_i$ is actually being added to the matching), we know that $s_{T_i} = 1$ before the transition and $s_{T_i}' = 0$ afterwards, as well as $s_{U_i} = 0$ and $s_{U_i}' = 1$, so the difference in these components has only one possible value.
	Otherwise, when some $T_i = U_j$ coincide and thus $T_i$ stays in the matching, we know that $s_{T_i} = s'_{T_i}$ and this component difference is again fixed as $0$.
	Thus, we have settled $\mu = 1$.

	Applying \cref{thm:single-step-bound} gives a bound on the expected number of steps of $3\cdot (n^3)^{2p}  \cdot (n^3)^2 \cdot \phi$.
\end{proof}

\subsubsection{Exact Cover by 3-Sets and Set-Cover}
The similar structure of the following cover problems allows for shared analysis.
Both use the $k$-differ neighbourhood, where the neighbours may differ in up to $k$ selected sets for the cover.
For constant $k$, this is an efficiently searchable neighbourhood.
First, let us define them separately:

\begin{itemize}
	\item Exact Cover by 3-Sets (\xIIIck) is the problem of finding a maximum weight exact cover of a finite set $\mathcal{B}$ with $\card{\mathcal{B}} = 3q$.
		An exact cover in this problem is a subset $S\subseteq \mathcal{C} = \{C_1,\dots, C_n\}, \card{S} = q$ of sets $C_i\subseteq \mathcal{B}$ with $\card{C_i} = 3$, such that $\bigcup_{C_i\in S}C_i = \mathcal{B}$.
        Thus, the sets in any $S$ must be disjoint and $S$ is a partition of $\mathcal{B}$ into sets of size $3$.
		With every set $C_i$ there is given a weight $w_i$ and we want to maximize the total weight $ w(S) = \sum_{i: C_i\in S}w_i $.

		The neighbours of an exact cover $S$ are given by adding or removing at most $k$ sets in total to/from $S$ so as to generate another exact cover.
		Because we know that $\card{S} = q$ is fixed, this means we replace up to $k/2$ sets from the solution by the same number of other sets from $\mathcal{C}$.

		\xIIIck is \pls/-complete due to \cite{Dumrauf10} for all $k \geq 12$.

	\item Set-Cover (\sck) is the problem of finding a minimum weight cover from a collection of sets $\mathcal{C} = \{C_1,\dots, C_n\}, C_i\subseteq \mathcal{B}$, i.e.\ a subset $S\subseteq \mathcal{C}$ such that $\bigcup_{C_i\in S}C_i = \mathcal{B}$.
		Each set $C_i$ has some weight $w_i$.
		The cost of a cover $S$ is given by $c(S) = \sum_{i:C_i\in S}w_i$.

		The neighbourhood of a cover is given by changing for up to $k$ sets whether they are being used.
		This means that for any $k_1+k_2\leq k$ there may be $k_1$ sets added to and $k_2$ sets removed from $S$ in order to construct a new cover of $\mathcal{B}$.
	
		\sck is \pls/-complete due to \cite{Dumrauf10} for all $k \geq 2$.
\end{itemize}

\begin{theorem}
	The expected number of iterations of local search for \xIIIck and \sck, when the weights $w_i$ for each set $C_i\in \mathcal{C}$ are drawn from independent distributions with density $f_i:[0, 1]\map [0, \phi]$, is bounded by $\landau{3^k n^{k+2} \phi}$.
\end{theorem}

\begin{proof}
	We formulate these problems as CLO problems, while directly translating the random distributions of the weights to the distributions of smoothed CLO costs and apply \cref{thm:single-step-bound}.

	Consider a binary ($M=1$) CLO problem where the cost part corresponds to the elements of the collection $\mathcal{C}$, i.e.\ $\nu = n$ and every set $C_i\in \mathcal{C}$ has a corresponding cost component $s_i$.
	The cost coefficients are set to the weights of the respective sets, $c_i = w_i$.
	The CLO configuration $\vecc s(S)$ induced by a cover $S$ for \xIIIck and \sck has $s_i(S) = 1$ whenever $C_i \in S$ and $s_i(S) = 0$ otherwise.
        Thus the weight of any cover is equal to the cost of its induced CLO configuration. 
	The CLO neighbourhood is derived directly and efficiently from the original instance of the problem.

	To apply \cref{thm:single-step-bound}, we establish a covering $(\mathcal{E}, \mathcal{I})$ and the parameters $\lambda = n^k, \beta = k, \mu = 3$ in the following.
	We start with the singleton coordinate cover, $\mathcal{I} = \{\{i\}, i = 1,\dots, n\}$.
	The transition cover $\mathcal{E}$ comprises clusters $E(D_1,\dots, D_k), D_i\in \mathcal{C}$, each indicating which $D_i$ are changed, regardless of their initial status; this means that if $D_i\in S$, it is being removed, and otherwise it is being added to the cover.
	There are at most $\lambda = n^k$ such clusters by accounting for every possible choice of $D_i$ in $k$ components.
	Note that the same $D_i$ may be chosen multiple times to account for changes of fewer than $k$ sets in the cover and that each permutation within the $D_i$ refers to the same transition cluster.
	Clearly, with each of these clusters we affect only the cost configuration of the sets $D_i$ that we change.
	Because we change at most $k$ sets in the transition, we have $\beta = k$.

	Finally, as these are binary problems, we know that $s_i'-s_i\in \{-1, 0, 1\}$ and therefore $\mu = 3$.
	By \cref{thm:single-step-bound} this gives a bound on the expected number of steps of $3\cdot 3^k n^k  \cdot n^2 \cdot \phi$.
\end{proof}

\subsubsection{Hitting Set}
An instance of Hitting Set (\hsk) is given by a collection of sets $\mathcal{C} = \{C_1,\dots, C_n\}, C_i \subseteq \mathcal{B}$ where $\mathcal{B}$ is a finite ground set, weights $w_i$ for each $C_i\in \mathcal{C}$ and some integer $m$.
A solution is a locally maximal subset $S \subseteq \mathcal{B}$ that is bounded in size $\card{S}\leq m$.
The weight of a subset $S$ is given by $ w(S) = \sum_{i: S\cap C_i \neq \emptyset} w_i$, i.e.\ the sum of all sets in $\mathcal{C}$ that are \emph{hit} by the subset $S$, which means that at least a single element of $C_i$ appears in $S$.

The local search neighbourhood of a subset $S$ is given by changing for at most $k$ elements from $\mathcal{B}$ whether they occur in $S$ or not.
This means for any $k_1+k_2\leq k$ we can add $k_1$ new elements to and remove $k_2$ elements from $S$ to create one of its neighbours.

\hsk is \pls/-complete due to \cite{Dumrauf10} for all $k \geq 1$ (note that the proof can only be found in the full version of the paper \cite{Dumrauf10_arXiv}).
In fact, the reduction allows for a slight adaptation that gives \pls/-completeness of \hsk for every $k \geq 1$ even when the maximum number of occurrences of an element within $\mathcal{C}$ is bounded.
In our proof deferred to \cref{sec:hittingset_pls} we reduce from \lmaxsatI with bounded variable occurrence (every variable appears in at most $\tilde{B}$ many clauses) instead of the standard \lmaxsatI; the additional property translates to the fact that the number of sets that contain a given element is bounded by $\tilde{B}+1$, and therefore also by a constant. \lmaxsatI with bounded variable occurrence was shown to be \pls/-complete in \cite{Krentel:structurelocalopt} and \cite{Klauck96} via tight \pls/-reductions. 

The starting problem of these reduction paths needs exponentially many steps to find a solution in the worst case, under its standard local search algorithm. Therefore, the \hsI instances we study here inherit this property, meaning that the standard local search algorithm has exponentially long improvement sequences in the worst case (see the relevant discussion in \cref{sec:model}). However, we prove that under smoothness, even in significantly wider families of instances, standard local search terminates after polynomially many iterations in expectation.

\begin{theorem}
	The expected number of iterations of local search for \hsk, when the weights $w_i$ for each set $C_i\in \mathcal{C}$ are drawn from independent distributions with density $f_i:[0, 1]\map [0, \phi]$, is bounded by $\landau{3^{kB} \card{\mathcal{B}}^k n^2\phi }$.
     The parameter $B$ denotes the maximum number of occurrences of an element within $\mathcal{C}$, i.e., $ B = \max_{b\in \mathcal{B}} \card{\{C_i\in \mathcal{C}\fwh b\in C_i\}} $.
\end{theorem}

\begin{proof}
	We formulate this problem as a CLO problem apply \cref{thm:single-step-bound} to its smoothed counterpart:

	Consider a binary ($M=1$) CLO problem, where the cost part corresponds to the elements of the collection $\mathcal{C}$, i.e.\ $\nu = n$ and every set $C_i$ corresponds to a cost component $s_i$.
	A subset $S$ induces a CLO configuration $\vecc s(S)$ with $s_i(S) = 1$ if and only if $S\cap C_i \neq \emptyset$, $s_i = 0$ otherwise.
	The cost coefficients are set to $c_i = w_i$ and thus, the cost of $\vecc s(S)$ corresponds to the weight of $S$ in the \hsk instance.
	The subset $S$ itself can be maintained in the non-cost configuration as an indicator vector over the ground set $\mathcal{B}$, i.e.\ $\bar \nu = \card{\mathcal{B}}$.
	Also, the neighbourhood of the CLO problem is efficiently derived from the neighbourhood of the \hsk instance and the smoothed costs of the \hsk instance correspond to the derived smoothed CLO model because the definition of the CLO cost coefficients.  

	To apply \cref{thm:single-step-bound}, we need to establish a covering $(\mathcal{E}, \mathcal{I})$ and the parameters $\lambda = \card{\mathcal{B}}^k, \beta = kB, \mu = 3$ in the following.
	We start with the singleton coordinate cover, $\mathcal{I} = \{\{i\}, i = 1,\dots, n\}$.
	The neighbourhood graph is covered by transition clusters $E(b_1,\dots, b_k), b_j\in \mathcal{B}$, each indicating for which elements in $\mathcal{B}$ the subset changes, regardless of the initial subset.
	Every such cluster therefore contains all valid transitions from an initial subset $S$, where each $b_j$ is removed from the set if $S$ already contains $b_j$, or added to the set if it doesn't.
	There are at most $\lambda = \card{\mathcal{B}}^k $ different transition clusters, accounting for each choice of $b_j$ in $k$ components.
	Note that we do allow the same $b\in \mathcal{B}$ to appear multiple times in order to represent changes of fewer than $k$ elements.
	Each permutation of the $b_j$ gives the same transition cluster.
	
	Lastly, we settle the properties $\beta$ and $\mu$ for a fixed transition cluster $E(b_1,\dots, b_k), b_j\in \mathcal{B}$.
        When changing any element $b$ in the subset, this can only affect the cost components $s_i$ of the $C_i$, where $b\in C_i$; the weight contribution of all other sets is independent of $b$.
	Because, by assumption, each $b$ occurs in at most $B$ different sets $C_i$, we know that in total there are at most $\beta = kB$ components affected by any change in $k$ different elements.
	Finally, within the configuration change in a transition for any of the singleton coordinate clusters, we simply use that the CLO problem is binary, and therefore $s_i' - s_i \in \{-1, 0, 1\}$ in any case, directly giving $\mu = 3$.
	Applying \cref{thm:single-step-bound} with these parameters yields a bound on the expected numbers of iterations of $3\cdot 3^{kB} \card{\mathcal{B}}^k  \cdot n^2 \cdot \phi$.
\end{proof}

\section{Conclusion and Open Problems}

In this work, we have introduced the class of CLO problems, which captures many relevant problems in \pls/ and allows us to define a corresponding smoothed version of each problem in a unifying way.
This framework allowed us to establish a black-box tool to give a bound on the smoothed running time for such problems by merely supplying some structural parameters of the problem, which can often be found naturally.
We have made use of this newly established framework by applying it to a multitude of existing problems.
Thereby, we rederive some existing results on the smoothed complexity of these problems. Furthermore, the framework allows us to apply the same structural ideas to problems that have not been discussed in the field of smoothed complexity before. For all these problems, we complement the positive results of smoothed running time with negative worst-case results. In particular, we show that they are \pls/-hard, and additionally, that the local-search algorithms used for the positive results actually admit exponential time instances.

A variety of more recent results in the field of smoothed analysis apply a technique that is not captured by our framework: instead of considering a single step in the improvement dynamics, the path through the neighbourhood graph is structurally clustered into sequences, for which the \emph{total} improvement is bounded by a problem-specific rank-based argument.
A natural question that arises is whether our framework can be generalized to capture this technique in a streamlined fashion as well, in order to unify even more of the landscape of smoothed analysis or establish better bounds than we can currently provide. This open question is related to Point~\ref{enum:open-problem-theorem3-1-1} of our discussion about the limitations of our key technical tool (\cref{thm:single-step-bound}), in \cref{sec:limitations-theorem}. 

There have been some observations on smoothed models, where considering a fully parameterized instance enabled improving the bounds on the smoothed running times of the standard local search algorithm compared to the general instances.
An example of this phenomenon is \lmaxcut for complete graphs for which smoothed polynomial time has been established in \cite{Angel:2017aa}, whereas for general graphs there is no better than quasi-polynomial smoothed time known. The former kind of graph input is fully parameterized (all possible edges are considered to undertake smoothed perturbation), as opposed to the latter.
We conjecture that there is a unifying framework of a similar flavour to ours which would work for these fully parameterized instances of \pls/-complete problems. This is related to Point~\ref{enum:open-problem-theorem3-1-2} of our discussion in \cref{sec:limitations-theorem}.

In this paper, we show that the problem of computing (pure, exact) equilibria in congestion games admits a smoothed polynomial running time algorithm for special classes of general and network congestion games (see \cref{th:smoothed-p-bounded-interaction-congestion-games} and \cref{th:pls-hardness-compact-network}, respectively). However, outside of these restricted instances of the problem, the complexity remains unresolved. An efficient algorithm for the computation of \emph{approximate} equilibria (FPTAS) under smoothed analysis, was recently presented in~\parencite{g2024}.

To the best of our knowledge, the published literature on smoothed analysis of problems in \pls/ has been limited to establishing upper bounds on the smoothed running time of the standard local search algorithm. The question that arises here is whether we can show that a problem remains \pls/-hard even if its input is perturbed (within some range) according to the smoothed model.
From the positive side, an interesting question is to characterize the class of \pls/-complete problems whose standard local search algorithm terminates in smoothed polynomial time.

\section*{Acknowledgements} Y.\ Giannakopoulos is grateful to Diogo Poças for many useful discussions and inspiration during the early stages of this project.
The authors thank Rahul Savani for pointing towards the inclusion of hedonic games to this paper.

\appendix

\section{Omitted Proofs}
\label{sec:proof-lemmas-appendix}

\subsection{Proof of \texorpdfstring{\cref{lemma:diversity-cover}}{Lemma~3.3}}

First, we show that the statement holds when $\mathcal{J}$ is a partition instead of a cover, i.e.\ we assume the sets in $\mathcal{J}$ are pairwise disjoint and $ J\subseteq I$ for every $J\in \mathcal{J}$.
Let us fix a subset of transitions $T\subseteq E$.
Denote by $\intuptozero{M}_\pm = \{-M, \dots, -1, 0, 1, \dots, M\}$ the set of possible values for the transition difference $s_i - s_i'$ in every component.
The set $\range{I}{T}$ thus contains vectors from $\intuptozero{M}_\pm^{I}$ and $\range{J}{T}\subseteq \intuptozero{M}_\pm^{J}, J\in \mathcal{J}$ analogously.
We identify each $\vecc x\in \intuptozero{M}_\pm^I $ with the family $(\vecc x_{J})_{J\in \mathcal{J}}$, as every index $i\in I$ is contained in exactly one of the $J\in \mathcal{J}$, and observe that $(\vecc x_{J})_{J\in \mathcal{J}} \in \prod_{J\in \mathcal{J}}\intuptozero{M}_\pm^J$.

We prove that $\range{I}{T} \subseteq \prod_{J\in \mathcal{J}}\range{J}{T}$ under this identification.
For the sake of contradiction, suppose that there exists $\vecc x\in \range{I}{T}$ with $\vecc x\notin \prod_{J\in \mathcal{J}}\range{J}{T}$. 
Therefore there exists $J \in \mathcal{J}$ such that $\vecc x_{J} \notin \range{J}{T}$.
But by definition, $\range{J}{T}$ contains $\vecc x_{J}$ due to the same transition $(\vecc s, \vecc s')$ that generated $\vecc x$ itself; a contradiction.

Now, we show that for $A\subseteq B$ we have $\delta_{A}(T)\leq \delta_{B}(T)$. Every $\vecc x, \tilde{\vecc x} \in \range{A}{T}, \vecc x\neq \tilde{\vecc x}$ are coming from transitions $(\vecc s, \vecc s'), (\tilde{\vecc{s}},  \tilde{\vecc{s}}')$.
These transitions also provide $\vecc x', \tilde{\vecc x}' \in \range{B}{T}$.
Since $A \subseteq B$ and for the projections of $\vecc x'$ and $\tilde{\vecc x}'$ on $A$ we have $\vecc x'_{A} = \vecc x, \tilde{\vecc x}'_{A} = \tilde{\vecc x}$, it holds that also $\vecc x'\neq \tilde{\vecc x}'$.
Thus, all elements in $\range{A}{T}$ have distinct representatives in $\range{B}{T}$ and therefore $\delta_{A}(T) \leq \delta_{B}(T)$.

Combining the observations above proves the main statement: Every cover
$\mathcal{J}$ of $I$ can be modified to be a partition $\mathcal{J}_p$ of $I$ by
removing every index $i\in I$ that appears more than once from all but a single
set and remove every $i \notin I$ entirely. Therefore, every set $A \in
\mathcal{J}_p$ has a corresponding superset $B \in \mathcal{J}$, and we conclude 
	\begin{equation*}
		\delta_I(T) \leq \prod_{A \in \mathcal{J}_p} \delta_{A}(T) \leq \prod_{B\in \mathcal{J}}\delta_{B}(T).
	\end{equation*}

\subsection{Proof of \texorpdfstring{\cref{lemma:ER-rank1-maxdensity}}{Lemma~3.4}}

For an arbitrary (absolutely) continuous real random variable $X$ with density
function $f:\R\to\R_{\geq 0}$ we will denote $M(X)\coloneqq
\esssup_{x\in\R}f(x)$.\footnote{Here $\esssup$ denotes the \emph{essential supremum} operator (see, e.g., \textcite[Sec.~21]{halmos1974measure}.); that is, roughly speaking, $M(X)$ is upper-bounding function $f$ \emph{almost everywhere} in its domain, excluding possibly a zero-measure set of points. Formally, $\esssup_{x\in\R}f(x)\coloneqq \inf\sset{c\in\R\fwh{\mu(\ssets{x\in \R\fwh{f(x)>c}})=0}}$ where $\mu$ is the standard Lebesgue measure on $\R$.} 
For the random variables in the statement of our Lemma we have $M(X_i)\leq \phi$ for all $i\in[m]$. Let also $Y_i=\xi_i X_i$, for all
$i\in[m]$, and 
$$Y=\sum_{i\in I^*} Y_i\qquad\text{where}\;\; I^*=\sset{i\in[m]\fwh{\xi_i\neq 0}}.$$
Notice that $I^*\neq \emptyset$, since $\vecc{\xi}=(\xi_1,\dots,\xi_m)$ is nonzero.

Then, it holds that $M(Y_i)=\frac{1}{\xi_i}M(X_i)$ for all $i\in I^*$. Furthermore,
since $Y_i$'s are independent, by~\textcite[Remark, p.~105]{Bobkov_2014} we know
that the following inequality applies:
$$ \frac{1}{M^2(Y)}\geq \frac{1}{2}\sum_{i\in I^*}\frac{1}{M^2(Y_i)}.$$
From this we get:
\begin{align*} 
M(Y) \leq \sqrt{2}\frac{1}{\sqrt{\sum_{i\in I^*}\frac{1}{M^2(Y_i)}}}
= \sqrt{2}\frac{1}{\sqrt{\sum_{i\in I^*}\frac{\xi_i^2}{M^2(X_i)}}}
\leq \sqrt{2}\phi\frac{1}{\sqrt{\sum_{i\in I^*}\xi_i^2}}
= \sqrt{2}\phi\frac{1}{\norm{\vecc{\xi}}_2}.
\end{align*}
Finally, using the above inequality we get that, for any $\varepsilon\geq 0$,
\begin{equation*}
\label{eq:ER-rank1-maxdensity-1}
\prob{0\leq \vecc{\xi}\cdot \vecc X \leq \varepsilon} 
= \prob{0 \leq Y \leq \varepsilon} 
\leq \int_{0}^{\varepsilon} M(Y)
\leq \varepsilon\cdot \sqrt{2}\phi\frac{1}{\norm{\vecc{\xi}}_2}.
\end{equation*}

We will now also derive an alternative bound. Let $i^*\in\argmax_{i\in I^*} \card{\xi_i}$, i.e.,  $\card{\xi_{i^*}}=\norm{\vecc{\xi}}_\infty$. Since $I^*\neq\emptyset$, we know that $\xi_{i^*}\neq 0$. Then, we get:
\begin{align}
\label{eq:ER-rank1-maxdensity-helper-1}
\prob{0 \leq \vecc{\xi}\cdot \vecc X \leq \varepsilon} 
	&=\prob{-\sum_{i\neq i^*}\xi_i X_i \leq \xi_{i^*} X_{i^*} \leq \varepsilon -\sum_{i\neq i^*}\xi_i X_i}\notag\\
	&=
\begin{dcases}
\prob{X_{i^*}\in\left(Z,Z+\frac{\varepsilon}{\card{\xi_{i^*}}}\right)}, &\text{if}\;\; \xi_{i^*}>0,\\
\prob{X_{i^*}\in\left(Z-\frac{\varepsilon}{\card{\xi_{i^*}}},Z\right)}, &\text{if}\;\; \xi_{i^*}<0,
\end{dcases}
\end{align}
where
$$
Z=-\frac{1}{\xi_{i^*}}\sum_{i\neq i^*}\xi_i X_i.
$$
Since random variables $X_{i^*}$ and $Z$ are independent, both probabilities in~\eqref{eq:ER-rank1-maxdensity-helper-1} can be uniformly
bounded, for each realization of $Z$, by $\frac{\varepsilon}{\card{\xi_{i^*}}}\phi$.

Merging the two bounds derived above, we finally get the desired
$$
\prob{0\leq \vecc{\xi}\cdot \vecc X \leq \varepsilon} \leq 
\min\left(\frac{1}{\norm{\vecc{\xi}}_\infty},\frac{\sqrt{2}}{\norm{\vecc{\xi}}_2}\right) \varepsilon\phi
$$

For nonzero \emph{integral} vectors $\vecc{\xi}$, we have $\min\left(\frac{1}{\norm{\vecc{\xi}}_\infty},\frac{\sqrt{2}}{\norm{\vecc{\xi}}_2}\right) \leq 1$, since $\norm{\vecc{\xi}}_\infty \geq 1$.

\section{\texorpdfstring{\pls/}{PLS}-Hardness of Restrained Congestion Games}
\label{sec:congestion-hardness}

Here we prove \cref{th:pls-hardness-restrained}, by showing that the problem of finding a PNE in a congestion game is
\pls/-complete even for constantly-restrained games and for the
simplest versions of general, polynomial, and step-function costs. Particularly, in \cref{th:pls-10-bounded-interaction} and \cref{thm: pls-bounded-linear-cong-games} we show that even $10$-restrained congestion games are \pls/-hard under all considered cost function models. 

We start by presenting the main
reduction from the problem \lmaxcutd to congestion games with at most two
strategies per player and access to at most $d$ resources per strategy. 
\lmaxcutd is the restricted version of \lmaxcut (see \cref{def:lmaxcut}) where the input graph has maximum degree $d$. 
Since \lmaxcutV is \pls/-complete (\cite{ElsasserT11}), our reduction shows that congestion games with two strategies per player and access to at most five resources is already \pls/-hard. The reduction we present is identical to the one that appears in the textbook of~\textcite[Theorem~19.4]{Roughgarden16} which reduces \lmaxcut to finding a PNE in congestion games. For completeness, we present the reduction here while keeping track of the important parameters of the resulting instance: the number of strategies per player and the number of resources per strategy.
Note that within our reduction, the strategies and the local search neighbourhood of the reduced instance correspond one-to-one to the cuts and neighbourhood of the \lmaxcutd instance.
Therefore, the reduction is tight in the sense of \cite{Schaffer91}.

Given a \lmaxcutd instance on the graph $G = (V, E)$ we will create a congestion game with player set $V$, where each player corresponds to a vertex of the graph. For each edge $e \in E$ we create two resources $r_e, \bar{r}_e$. Let the (graph) neighbourhood of a vertex $i \in V$ be denoted by $\nei{i} := \{j \in V ~|~ ij \in E\}$. Let us also denote the edges incident to $i$ by $\inc{i} := \{ij \in E \fwh j \in \nei{i}\}$. Each player $i \in V$ has two strategies, namely, $S_1(i) := \{r_e\}_{e \in \inc{i}}$ and $S_2(i) := \{\bar{r}_e\}_{e \in \inc{i}}$. For $e \in E$, the cost for using a resource is
\begin{align}\label{eq: pls-congestion-cost}
	\kappa_{r_e}(1) = \kappa_{\bar{r}_e}(1) = 0, \quad \text{and} \quad \kappa_{r_e}(\ell) = \kappa_{\bar{r}_e}(\ell) = w_e, \quad \text{for} \quad \ell \geq 2.
\end{align}
Notice that, by construction, at most two players can use the same resource.

We claim that any PNE of the congestion game corresponds to a solution to the initial \lmaxcutd instance, that is, a cut $(V_1, V_2)$ in the following way: For any player $i$ that plays strategy $S_1(i)$, the corresponding node $i \in V$ is placed in $V_1$, otherwise she is placed in $V_2$. Consider a PNE of the congestion game, and for the sake of contradiction, suppose there exists a vertex $i \in V_1$ which, if moved to $V_2$, yields a greater cut weight. Let $E_1(i) := \{ij \in E \fwh j \in V_1\}$ and $E_2(i) := \{ij \in E \fwh j \in V_2\}$ be the edges between $i$ and her neighbours in $V_1$ and $V_2$, respectively. Then $\sum_{e \in E_1(i)} w_e - \sum_{e \in E_2(i)} w_e > 0$. Now consider the difference in experienced cost of player $i$ when deviating from $S_1(i)$ to $S_2(i)$. This would be 
\begin{align*}
	 \sum_{e \in E_{1}(i)} \kappa_{\bar{r}_e}(1) + \sum_{e \in E_{2}(i)} \kappa_{\bar{r}_e}(2) - \sum_{e \in E_{1}(i)} \kappa_{r_e}(2) - \sum_{e \in E_{2}(i)} \kappa_{r_e}(1) &= \sum_{e \in E_{2}(i)} \kappa_{\bar{r}_e}(2) - \sum_{e \in E_{1}(i)} \kappa_{r_e}(2) \\
	 &= \sum_{e \in E_{2}(i)} w_e - \sum_{e \in E_{1}(i)} w_e \\
	 &< 0.
\end{align*}
This means that player $i$ would decrease her cost if she unilaterally deviated to $S_2(i)$, therefore, the original state was not a PNE, a contradiction. Similarly, we get a contradiction by assuming that there is a vertex $i \in V_2$ which, if moved to $V_1$, yields a greater cut weight. Hence, any PNE of the congestion game corresponds to a solution of \lmaxcutd.
In this reduction, the constructed congestion game has step-function costs (with only a single step at $\ell = 2$), implying immediately that it holds for general cost functions (see \eqref{eq: pls-congestion-cost}). 
Since for any $d \geq 5$ \lmaxcutd is \pls/-complete (\cite{ElsasserT11}), the above reduction yields the following.

\begin{theorem}\label{thm: pls-bounded-cong-games}
    Finding a PNE, in congestion games with step function or general cost representation, is \pls/-complete even when every player has 2 strategies, each consisting of at most 5 resources. 
\end{theorem}

Recall now that, by construction, $|\inc{i}| \leq d$ for all $i \in V$. Therefore, $|S_1(i)| = |S_2(i)| \leq d$, which implies that the symmetric difference between any pair of strategies of a player in such congestion games is at most $2 d$ (see \cref{def:restrained-games}). Therefore, we get the following.

\begin{corollary}
	\label{th:pls-10-bounded-interaction}
    Finding a PNE in $10$-restrained congestion games, with step-function or general cost representation, is \pls/-complete.
\end{corollary}

We will now show how they also hold for affine cost functions with nonnegative coefficients, that is, the simplest class of polynomial functions, excluding constant functions.
For $e \in E$, let the cost for using the corresponding resource be
\begin{align}\label{eq: pls-linear-congestion-cost}
    \kappa_{r_e}(\ell) = \kappa_{\bar{r}_e}(\ell) = \ell \cdot w_e, \quad \text{for} \quad \ell \geq 1,
\end{align}
and recall that, by construction, at most 2 players can be using the same resource.

Similarly to the previous reduction, we show that any PNE of this congestion game corresponds to a solution to the initial \lmaxcutd instance. In particular, for every player $i$ using $S_1(i)$ (resp.\ $S_2(i)$) we place its corresponding vertex $i$ in $V_1$ (resp.\ $V_2$). Consider a PNE of the congestion game, and suppose there exists a vertex $i \in V_1$ which, if moved to $V_2$, yields a greater cut weight. Therefore, $\sum_{e \in E_1(i)} w_e - \sum_{e \in E_2(i)} w_e > 0$. The difference in experienced cost of player $i$ when deviating from $S_1(i)$ to $S_2(i)$ is
\begin{align*}
    \sum_{e \in E_{1}(i)} \kappa_{\bar{r}_e}(1) + \sum_{e \in E_{2}(i)} \kappa_{\bar{r}_e}(2) - &\sum_{e \in E_{1}(i)} \kappa_{r_e}(2) - \sum_{e \in E_{2}(i)} \kappa_{r_e}(1) =  \\
    &= \sum_{e \in E_{1}(i)} w_e + \sum_{e \in E_{2}(i)} 2w_e - \sum_{e \in E_{1}(i)} 2w_e - \sum_{e \in E_{2}(i)} w_e \\
    &= \sum_{e \in E_{2}(i)} w_e - \sum_{e \in E_{1}(i)} w_e \\
    &< 0.
\end{align*}
So, player $i$ would decrease her cost if she unilaterally deviated to $S_2(i)$, contradicting the assumption that the original state was a PNE. Similarly, we are led to a contradiction by assuming that there is a node $i \in V_2$ which, if moved to $V_1$, yields a greater cut weight. Therefore, any PNE of the congestion game with linear cost function corresponds to a solution of \lmaxcutd, and again, since for any $d \geq 5$ \lmaxcutd is \pls/-complete (\cite{ElsasserT11}) we get the following.

\begin{theorem}\label{thm: pls-bounded-linear-cong-games}
    Congestion games with affine cost functions are \pls/-complete even when every player has 2 strategies, each consisting of at most 5 resources, and each resource being available to at most 2 players. In particular, such games are $10$-restrained.
\end{theorem}

\section{\texorpdfstring{\pls/}{PLS}-Hardness of Compact Network Congestion Games}
\label{sec:pl-hardness-network-compact}

Here we prove that computing a PNE in network congestion games is \pls/-hard even for games with the simplest versions of general, polynomial, and step cost functions. Our reduction is similar to that of \cite{Ackermann2008}, however, our reduction route is simpler: we start from \lmaxcutd (see introduction of \cref{sec:congestion-hardness}) and we reduce to network congestion games via general congestion games. We present it for completeness while remarking how the parameters of interest, namely, the maximum number of best-responses per player, $A$, and the maximum path-length per player, $B$ (see \cref{def:compact-network-game}), are affected by the corresponding parameters of the initial problem. 

We will take the reduction of \cref{sec:congestion-hardness} a step further and end up with a network congestion game with at most 2 \emph{cost levels} per edge-cost; the number of cost levels for a given edge cost function is defined as the number $d+1$, where $d$ is the number of break-points when the cost function is in its step function representation (see \cref{sec:congestion-games-smoothed-models}). The reduction in the proof of \cref{thm: pls-bounded-cong-games} had as a starting point the \pls/-complete problem of computing a local max cut in a \lmaxcutd instance, so we will be essentially reducing this problem to computing PNEs in our family of network congestion games. Since the latter is also in \pls/ (\cite{FPT04}), our reduction implies its \pls/-completeness.
Furthermore, the reduction is tight, since $\mathcal{R}$ (the set of reasonable feasible solutions in tight reductions according to Definition 3.2 in \cite{Schaffer91}) is our set $\prod_{i\in \mathcal{N}} \Sigma_i^*$ of strategy profiles that consist of best-responses.

Let us continue the reduction from the proof of \cref{thm: pls-bounded-cong-games}. In this section we will be calling $I$ the congestion game instance constructed in the aforementioned proof. Let the set of players be $[n]$, and the two strategies of player $i \in [n]$ be $S_1 (i)$ and $S_{2}(i)$, where $|S_1 (i)| = |S_2 (i)| \leq d$. We now construct the network congestion game instance $I'$. To avoid confusion with the terms ``vertices'' and ``edges'' of the local max cut problems, we use instead ``nodes'' and ``arcs'', respectively, for network congestion games. We therefore also attain the vertex set and edge set notation $V$ and $E$ of the former problem, and denote by $N$ and $A$ the corresponding sets of the latter problem. The set of players in $I'$ is identical to that of $I$. Let us make a pair of nodes for each player $i$, an origin node $o_{i}$ and a destination node $d_{i}$. For any pair $i,j \in [n]$ with $i < j$, if edge $ij \in E$, we create a \emph{supernode} $U_{ij}$ as shown in \cref{fig:supnode}.

\begin{figure}
\centering
\begin{minipage}[t]{\dimexpr.35\textwidth}
  \centering
  \includegraphics[width=1\textwidth]{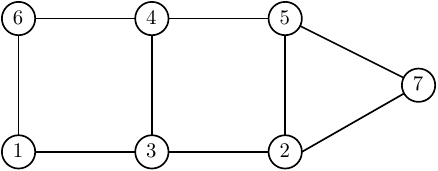}
  \captionof{figure}{The graph of a \lmaxcutIII instance.}
  \label{fig:local-max-cut}
\end{minipage}%
\qquad
\begin{minipage}[t]{\dimexpr.6\textwidth}
  \centering
  \includegraphics[width=0.4\textwidth]{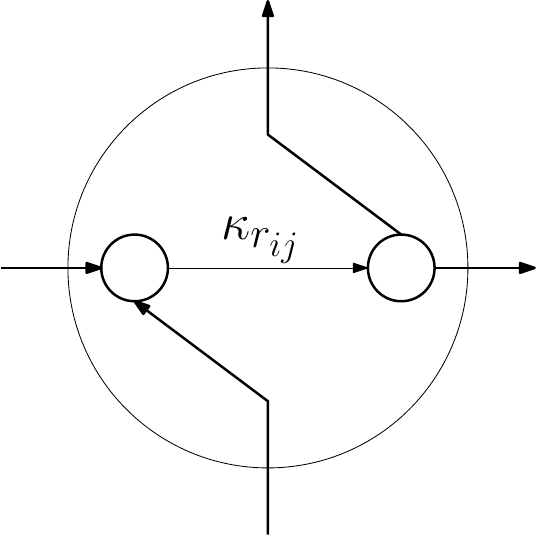}
  \captionof{figure}{Supernode $U_{ij}$. It admits at most two incoming arcs and at most two outgoing arcs. The arc between the two nodes inside it has cost function $\kappa_{r_{ij}}$ (see \eqref{eq: pls-congestion-cost}). Supernode $\bar{U}_{ij}$ is defined analogously.
  }
  \label{fig:supnode}
\end{minipage}
\end{figure}

\begin{figure}
	\centering
	\includegraphics[scale=0.8]{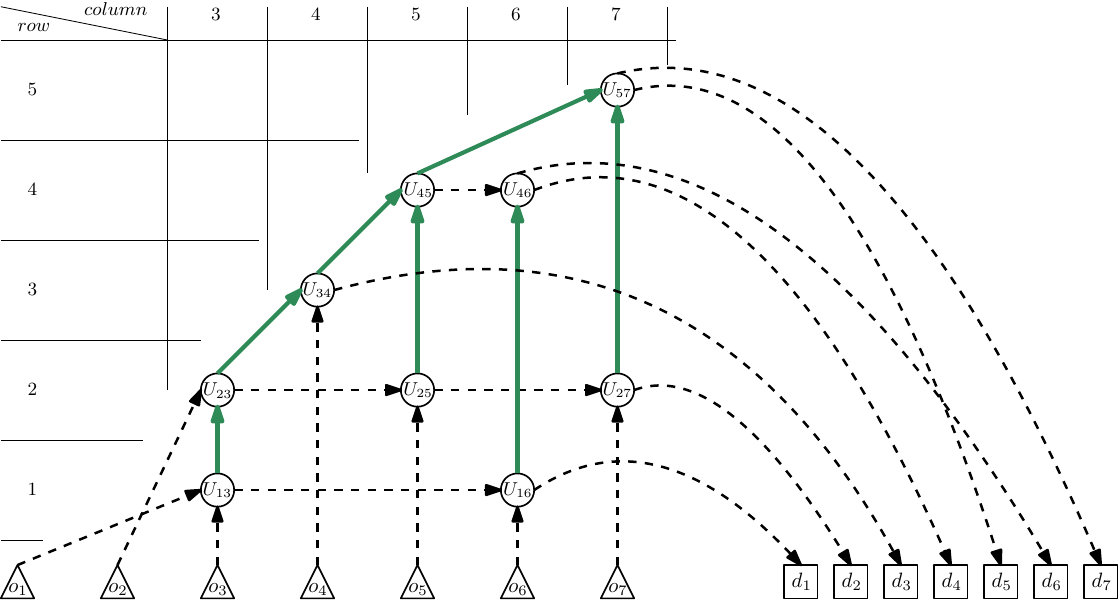}
	\caption{Half of the network congestion game instance $I'$ corresponding to the \lmaxcutIII instance of Figure \ref{fig:local-max-cut}; the other half is symmetric and consists of supernodes $\bar{U}_{ij}$.
    Supernodes $U_{ij}$ have a circular shape. The origin nodes have a triangular shape, while the destination nodes have a square shape. The black/dashed arcs are light, while the green/solid ones are heavy.
    }\label{fig:net-cong}
\end{figure}

We will now define the arcs of $I'$, and let us remark that all the arcs are directed. We have three kinds of arcs: the \emph{regular} ones inside each supernode $U_{ij}$ which have cost function $\kappa_{r_{ij}}$ as defined in \eqref{eq: pls-congestion-cost}\footnote{In \cref{cl: two players per str} we will show that at most two players can simultaneously use a regular arc. Similarly to $U_{ij}$, in each supernode $\bar{U}_{ij}$ there exists a regular arc with cost function $\kappa_{\bar{r}_{i,j}}$.}; the \emph{light} ones that have cost $0$ for any number of players using them; and the \emph{heavy} ones between supernodes $U_{ij}, U_{i'j'}$ which have cost $(i' - i + j - 1) \cdot W$, for any number of players using them, where $W > \sum_{ij \in E} w_{ij}$. Consider now all existing supernodes $U_{ij}$ for a fixed $i$, and let $R_{i}$ be the increasingly ordered tuple of the elements in $\{ j ~|~ ij \in E \land i < j \}$. Note that $i \in [n-1]$. We create a line graph, called \emph{row $i$}, by connecting each supernode $U_{ij}$ with a light arc to $U_{ij'}$, where $j,j'$ are two consecutive elements of tuple $R_{i}$. Now consider all existing supernodes $U_{ij}$ for a fixed $j$, and let $C_{j}$ be the increasingly ordered tuple of the elements in $\{ i ~|~ ij \in E \land i < j \}$. Note also that $j \in [n]\setminus\{1\}$. We connect each supernode $U_{ij}$ with a heavy arc to $U_{i'j}$, where $i,i'$ are two consecutive elements of tuple $C_{j}$. 

For each $i \in [n]$, let $U_{k \ell}$ be a supernode such that either $k = i$ or $\ell = i$, and among such existent supernodes consider $U_{k^{-} \ell^{-}}$, where $k^{-} + \ell^{-} \leq k + \ell$. Similarly, consider $U_{k^{+} \ell^{+}}$, where $k^{+} + \ell^{+} \geq k + \ell$. We connect via light arcs $o_{i}$ to $U_{k^{-} \ell^{-}}$, and $U_{k^{+} \ell^{+}}$ to $d_{i}$. Also, if both $C_{i}$ and $R_{i}$ are non-empty, let $k^{*}, \ell^{*}$ be the largest element in $C_{i}$ and the smallest element in $R_{i}$, respectively. We connect $U_{k^{*} i}$ to $U_{i \ell^{*}}$ via a heavy arc.

We have described half of the construction, where we translated the resources $r_{e}, e \in E$ of the intermediate congestion game instance $I$ into paths of $I'$. We have to implement an entirely symmetric construction where we also translate resources $\bar{r}_{e}$ into paths by adding an identically connected set of supernodes $\bar{U}_{ij}$ for $ij \in E$, which connect to the existing set of $o_i$'s and $d_i$'s. As an example of the reduction, we present in \cref{fig:local-max-cut} a \lmaxcutIII instance, and in \cref{fig:net-cong} the respective network congestion game instance $I'$. For presentation simplicity, we only show half of the network since the one consisting of supernodes $\bar{U}_{ij}$ is symmetric.

The strategy set of each player $i \in [n]$ consists of all possible $o_i - d_i$ paths. For fixed $i$, let $P(i)$ be the path consisting of all existing $U_{k \ell}$'s for which either $k = i$ or $\ell = i$, and similarly, let $\bar{P}(i)$ be defined analogously using $\bar{U}_{k\ell}$'s. We will call $Q_{1}(i)$ and $Q_{2}(i)$ the strategies $o_i - P(i) - d_i$ and $o_i - \bar{P}(i) - d_i$, respectively.

We now show how to translate a solution of $I'$ directly to a solution of the initial \lmaxcutd instance in polynomial time. For any player $i$ that plays $Q_{1}(i)$ (resp.\ $Q_{1}(i)$) in $I'$, we place vertex $i$ in $V_1$ (resp.\ $V_2$), and observe that due to \cref{cl: dom-strat} only $Q_{1}(i)$ or $Q_{2}(i)$ can be played in a PNE. The following auxiliary definition and claims will be needed in the correctness proof of our reduction. 

\begin{definition}[Strictly dominant set]
    For a given player $i$ with strategy set $\Sigma_i$, the set $D_i \subseteq \Sigma_i$ is strictly dominant if every strategy in $D_i$ strictly dominates all strategies outside $D_i$; formally, for every $\sigma^{*}_i \in D_i$, $\sigma'_i \in \Sigma_i \setminus D_i$, we have $C_i (\sigma^*_{i}, \vecc\sigma_{-i}) < (\sigma'_{i}, \vecc\sigma_{-i}) $, for all strategies $\vecc\sigma_{-i}$ of the other players.
\end{definition}
Observe that, for any fixed player $i$, the strictly dominant set $D_i$ is unique, but it can be empty.

\begin{claim}\label{cl: dom-strat}
	For any player $i \in [n]$, $Q(i) := \{Q_{1}(i), Q_{2}(i)\}$ is strictly dominant.
\end{claim}

\begin{proof}
	For some player $i$, if $C_{i}$ is non-empty let $k$ be its smallest item, otherwise $k = 0$. The cost experienced by $i$ when playing $Q_{1}(i)$ is at most $\sum_{ij \in E} w_{ij} + (i - k + i - 1) \cdot W < (2 i - k) \cdot W$, where in the left-hand side of the inequality the leftmost term is due to the regular arcs, and the rightmost term is due to the heavy arcs. The exact same upper bound holds if she played $Q_{2}(i)$, by definition of the costs on the corresponding arcs.
	
	If $i$ played any strategy other than $Q_{1}(i)$ or $Q_{2}(i)$, then she would be using a path that reaches row $i$, otherwise she would not have access to $d_i$. Therefore, her path passes through heavy arcs starting from row $k \leq i - 1$ and reaching row $i$. Note that if $k = i$ then the path is $Q_{1}(i)$ or $Q_{2}(i)$, and if $k > i$ then $i$ does not have access to $d_i$. Also, at least one of the heavy arcs must be between two supernodes that are on columns $\ell, \ell'$, respectively, where $\ell' \geq \ell \geq i + 1$, otherwise the path is $Q_{1}(i)$ or $Q_{2}(i)$. Therefore, her cost would be at least $0 + (i - k + \ell - 1) \cdot W \geq (2 i - k) \cdot W$, where again, in the left-hand side of the inequality the leftmost term is due to the regular arcs, and the rightmost term is due to the heavy arcs.
	
	We conclude that for any player $i$, each of the strategies in $Q(i)$ strictly dominates all strategies outside $Q(i)$.
\end{proof}

\begin{claim}\label{cl: two players per str}
	In any PNE of the network congestion game, at most two players can be using the same regular arc.
\end{claim}

\begin{proof}
    Suppose we are in a PNE. Since no strictly dominated strategy can be played by any player in a PNE, we know that any player $i$ can only be playing either $Q_{1}(i)$ or $Q_{2}(i)$ by \cref{cl: dom-strat}. For any regular arc consider its supernode $U_{i j}$. By definition of the aforementioned strategies, this supernode can only be used by players $i$ and $j$.
\end{proof}

Now we will show the correctness of the reduction. First, note that the construction described can be performed in polynomial time. It remains to prove that any PNE of the network congestion game corresponds to a solution of the \lmaxcutd instance. Consider a PNE in the constructed game and observe that due to \cref{cl: dom-strat}, no player $i$ can be playing some strategy other than $Q_{1}(i)$ or $Q_{2}(i)$. For the sake of contradiction, suppose that there is a node $i \in V_1$ which yields greater cut value if moved to $V_2$. Then $\sum_{e \in E_1(i)} w_e - \sum_{e \in E_2(i)} w_e > 0$. 

Recall that there exists a supernode $U_{m \ell}$ only if $m\ell \in E$. If $C_{i}$ is non-empty let $k$ be its smallest item, otherwise $k = 2i$. Let us compute the difference in experienced cost of player $i$ when she deviates from $Q_{1}(i)$ to $Q_{2}(i)$ (the only other strategy that can yield lower cost). This is   
\begin{align*}
	\sum_{e \in E_{1}(i)} \kappa_{\bar{r}_e}(1) + \sum_{e \in E_{2}(i)} \kappa_{\bar{r}_e}(2) + (2i - k) \cdot W - \sum_{e \in E_{1}(i)} \kappa_{r_e}(2) - &\sum_{e \in E_{2}(i)} \kappa_{r_e}(1) - (2i - k) \cdot W = \\ 
	&= \sum_{e \in E_{2}(i)} \kappa_{\bar{r}_e}(2) - \sum_{e \in E_{1}(i)} \kappa_{r_e}(2) \\
	&= \sum_{e \in E_{2}(i)} w_e - \sum_{e \in E_{1}(i)} w_e \\
	&< 0.
\end{align*}
This means that the initial strategy profile was not a PNE, a contradiction. Assuming there is a vertex in $V_2$ in the \lmaxcutd instance that could improve the cut value by moving to $V_1$, via a similar reasoning we can arrive to a contradiction. Notice that by construction of $Q_{1}(i)$ and $Q_{2}(i)$, each player $i$ passes through at most $d$ supernodes in a PNE. Note also that since every best-response of a player is necessarily a strategy in the strictly dominant set $D$, we have that for parameter $B$ of \cref{def:compact-network-game} it holds $B \leq \card{D} = 2$. In this reduction, the constructed network congestion game has step cost functions, implying immediately that it holds for general cost functions (see \eqref{eq: pls-congestion-cost}). Finally, since \lmaxcutd is \pls/-complete for any $d \geq 5$ and network congestion games are in \pls/ (\cite{FPT04}), we get the following theorem.

\begin{theorem}\label{thm: pls-bounded-net-cong-games}
	Network congestion games, with step function or general cost representation, are \pls/-complete even when every player has 2 strategies that are in the strictly dominant set, each having maximum length $11$, and each of the network's edges has at most $2$ cost levels. In particular, such games are $(2,11)$-compact. 
\end{theorem}

We will now show how the same result also holds for affine cost functions with nonnegative coefficients, that is, the simplest class of polynomial functions, excluding constant functions.
Similarly to the proof of \cref{thm: pls-bounded-linear-cong-games}, let us set the cost functions of regular arcs to be \eqref{eq: pls-linear-congestion-cost}. 
Consider a PNE of the Network Congestion Game, and recall that no player $i$ can be playing some strategy other than $Q_{1}(i)$ or $Q_{2}(i)$ due to \cref{cl: dom-strat}. Suppose now that there is a node $i \in V_1$ which yields greater cut value if moved to $V_2$. Therefore, $\sum_{e \in E_1(i)} w_e - \sum_{e \in E_2(i)} w_e > 0$. 

Again, if $C_{i}$ is non-empty let $k$ be its smallest item, otherwise $k = 2i$. The difference in experienced cost of player $i$ when she deviates from $Q_{1}(i)$ to $Q_{2}(i)$ is   
\begin{align*}
	\sum_{e \in E_{1}(i)} \kappa_{\bar{r}_e}(1) + \sum_{e \in E_{2}(i)} \kappa_{\bar{r}_e}(2) + (2i - k) \cdot W &- \sum_{e \in E_{1}(i)} \kappa_{r_e}(2) - \sum_{e \in E_{2}(i)} \kappa_{r_e}(1) - (2i - k) \cdot W = \\ 
    &= \sum_{e \in E_{1}(i)} w_e + \sum_{e \in E_{2}(i)} 2w_e - \sum_{e \in E_{1}(i)} 2w_e - \sum_{e \in E_{2}(i)} w_e \\
    &= \sum_{e \in E_{2}(i)} w_e - \sum_{e \in E_{1}(i)} w_e \\
    &< 0.
\end{align*}
Therefore, this contradicts the assumption that the initial strategy profile was a PNE. Similarly, by assuming there is a vertex in $V_2$ in the \lmaxcut-$d$ instance that could improve the cut value by moving to $V_1$, we get a contradiction. Hence, we get the following theorem.
\begin{theorem}\label{thm: pls-bounded-linear-net-cong-games}
    Network Congestion Games with affine cost functions are \pls/-complete, even when every player has 2 strategies that are in the strictly dominant set, each having maximum length $11$. In particular, such games are $(2,11)$-compact.
\end{theorem}

\section{\texorpdfstring{\pls/}{PLS}-Hardness of \texorpdfstring{\lmaxkcut}{Local Max-k-Cut}}
\label{sec:maxkcut_pls}
Here we present a \pls/-hardness reduction from \lmaxcutd to \lmaxkcut for any $k \geq 2$ under the \flip neighbourhood. \lmaxcutd is the restricted version of \lmaxcut where the input graph has maximum degree $d \in \N^*$.

\begin{theorem}
    \lmaxkcut is \pls/-hard for any $k \geq 2$, even when the maximum degree of the input graph is $k + 3$.
\end{theorem}

\begin{proof}
    Consider an instance of \lmaxcutd on graph $G' = (V,E)$ and weights $w_e, e \in E$, where $|V| = n \in \N^*$, and without loss of generality, let $V = \{v_1, \dots, v_n\}$. For each vertex $v_i \in V$ we introduce $k$ \emph{special vertices} $v_i^1, \dots, v_i^{k}$. For each fixed $i \in [n]$, we introduce a \emph{special edge} $v_i^j v_i^\ell$ for every $j,\ell \in [k]$ with $j \neq \ell$, where $w_{v_i^j v_i^\ell} := W^* > \sum_{e \in E} w_e$. Additionally, for each fixed $i \in [n]$ we introduce a special edge $v_i v_i^j$ for every $j \in \{3,4,\dots,k\}$, where $w_{v_i v_i^j} := W^* > \sum_{e \in E} w_e$. This instance creation clearly takes polynomial time. Suppose the cut $(V_1, \dots, V_k)$ is a solution to the constructed \lmaxkcut instance, where each $V_j$ is a \emph{partition block}. We translate this solution to a solution of \lmaxcutd by considering as one partition block $V_1 \cap V$ and the other block $V_2 \cap V$. 
    
    It remains to show that the latter is indeed a solution to the initial \lmaxcutd instance.  First, observe that no two special vertices $v_i^j, v_i^\ell$ are in the same partition block. If that were the case, then there would exist a vertex-part $V_j$ that does not contain any special vertex from the set $\{v_{i}^1, \dots, v_i^{k}\}$, and moving $v_i^j$ to $V_j$ increases the weight of the cut by at least $W^* - \sum_{e \in E} w_e > 0$, contradicting the local optimality of the cut. Therefore, without loss of generality (by renaming the vertex-parts), for any fixed $j \in [k]$, let $V_j$ contain $v_i^j$ for all $i \in [n]$. Now we claim that no vertex $v_i \in V$ is contained in any $V_j$ for $j \in \{3,4,\dots,k\}$. If that were the case, then moving $v_i$ to $V_1$ (or $V_2$) would increase the weight of the cut by at least $W^* - \sum_{e \in E} w_e > 0$, contradicting the local optimality of the cut.
    
    To prove that this instance's solution induces a solution to the initial \lmaxcutd instance, we examine the vertices $v_i \in V$. According to the above, all $v_i$'s are in $V_1 \cup V_2$. Let us now consider some $v_i \in (V_1 \cup V_2) \cap V$ of the \lmaxkcut solution, and suppose that in the induced \lmaxcutd cut, by moving $v_i$ from its current cut-set to the other cut-set, resulted in a cut with greater weight. Then this would also be the case in the \lmaxkcut solution, since there are no special edges between $v_i$ and any vertex of $V_1 \cup V_2$. This contradicts its local optimality, and completes the correctness of the reduction. 
    
    Notice now that the \lmaxkcut instance has maximum degree $k+d-2$: the maximum degree of any vertex $v_i$ is $d + (k-2)$, while that of any vertex $v_i^j$ is $(k-1) + 1$. Since \lmaxcutd is \pls/-hard for any $d \geq 5$ (\cite{ElsasserT11}), we get the statement of our theorem.
\end{proof}

\section{\texorpdfstring{\pls/}{PLS}-Hardness of Hitting Set with Bounded Element Occurrence}
\label{sec:hittingset_pls}
Here we present an adaptation of the proof of \pls/-hardness due to \cite{Dumrauf10_arXiv}.
Instead of reducing from \lmaxsatI (or CNFSat as it is called in the aforementioned work), we show a tight \pls/-reduction from \lmaxsatI with bounded variable occurrence, which is also \pls/-complete due to \cite{Krentel:structurelocalopt} and \cite{Klauck96} via tight \pls/-reductions. 

\begin{theorem}
	\hsI is \pls/-hard, even when the number of sets in which any single element occurs is bounded by a constant.
\end{theorem}
\begin{proof}
We reduce from \lmaxsatI with bounded variable occurrence, i.e.\ we are given a CNF formula with clauses $\tilde{\mathcal{C}} = \{\tilde{C}_1,\dots, \tilde{C}_{\tilde{m}}\}$ over binary variables $\tilde{\mathcal{X}} = \{x_1, \dots, x_{\tilde{n}}\}$ and weights $\tilde{w}_i$ for every $\tilde{C}_i$.
Additionally, every variable only occurs in at most $\tilde{B}$ clauses.

Given such an instance of \lmaxsatI, we define an instance of \hsI and show that every locally optimal subset corresponds to a locally optimal assignment of the \lmaxsatI instance.
The ground set is given by a positive and negative element for every variable; $\mathcal{B} = \bigcup_{x\in \tilde{\mathcal{X}}}\{x, \bar{x}\}$.
For each of the variables we also define a set $C_x = \{x, \bar{x}\}, x\in \tilde{\mathcal{X}}$ with large weight $W$ exceeding the sum of all clause weights $\tilde{w}_i$.
We restrict the subset size to the number of variables, $m = \tilde{m}$, and the subset corresponding to a truth assignment is given by choosing the respective positive or negative element for every variable according to the assignment.
Finally, for every clause $\tilde{C}_i$ of the \lmaxsatI instance, we define a set $C_i$ for \hsI that contains exactly the variables of the clause ($x$ for positive and $\bar x$ for negative literals) and its weight corresponds to the weight of the clause, $ w_i = \tilde{w}_i $.
The auxiliary sets $C_x$ imply that every locally optimal subset $S$ contains exactly one of $x, \bar{x}$ for every variable:
Otherwise, because of our choice of $m$, it must hold that for some variable both $x, \bar{x}$ do not occur in $S$, while for another one both $y, \bar{y}$ do occur.
Because the sets $C_x$ have large weight, it is a local improvement to replace $y$ by $x$, thus making both $C_x$ and $C_y$ contribute their weight to the total weight of the subset.

Now, every locally optimal subset $S$ corresponds to a truth assignment and their weights also coincide up to the constant weight $m\cdot W$ of the auxiliary sets:
A clause $\tilde{C}_i$ is satisfied as soon as any of its literals is true, and analogously the set $C_i$ contributes its weight to the solution as soon as any of its elements are contained in $S$.
The neighbourhood graph contains all changes of an element $x$ to its negative copy $\bar{x}$ or vice-versa.
However, when starting from a subset that corresponds to a truth assignment, a change to a neighbouring subset that does not correspond to one is never an improving move, which concludes the reduction as being tight.

Up to this point, this is exactly the reduction from \cite{Dumrauf10_arXiv}.
When the original instance of \lmaxsatI has bounded variable occurrence, it follows immediately that the \hsI instance has bounded element occurrence $B = \max_{b\in \mathcal{B}}\card{\{C\in \mathcal{C}\fwh b\in C\}}$:
The maximum number of element occurrence in the \hsI instance is $B \leq \tilde{B} + 1$, as apart from the sets $C_i$ that correspond to the clauses, we have only added the auxiliary set $\{x, \bar{x}\}$ for every variable to the problem.
\end{proof}

\section{Maximum Constraint Assignment}
\label{sec:MCA}
This section deals with the problem of Maximum Constraint Assignment (\mcapqr) as defined in \cite{DumraufMonien:PLS_MCA}.
\mcapqr is a general optimization problem over integer variables.
A heavily restricted variant of this problem has undergone smoothed analysis by the name of Binary Max-CSP in \cite{Chen20}.
Compared to (even non-binary) Max-CSP, $\mcapqr$ covers many possible weights for the different assignments of a single constraint instead of just a single one.
Our focus will be on translating the parameters of \mcapqr directly to parameters in the smoothed running time.

\begin{definition}[\mcapqr]
	By \mcapqr we denote the following local optimization problem. Let $\mathcal{C} = \{C_1,\dots, C_m\} $ be a set of constraints over variables $ \mathcal{X} = \{x_1,\dots, x_n\}$, i.e.\ each constraint $C_i$ is associated with a set of variables and written symbolically as $C_i(x_{i_1},\dots, x_{i_{p_i}})$.
	When assigning integers from $\intupto{r}$ to the variables, these constraints are evaluated as functions $C_i:\intupto{r}^{p_i}\map [0, 1]$ with $p_i\leq p$.
	The feasible solutions are the integer assignments of the variables, i.e.\ functions $a: X \map \intupto{r}$.
	The sum of all constraints evaluated for the assignment is the weight of the solution:
		\begin{equation*}
			w(a) = \sum_{i=1}^m C_i(a(x_{i_1}),\dots, a(x_{i_{p_i}})).
		\end{equation*}
	Every variable $x_i$ appears in at most $q$ constraints from $\mathcal{C}$.
	The local search neighbourhood is given by changing the mapping of a single variable of a given assignment to a different value.
	This is an efficiently searchable neighbourhood.
\end{definition}

As shown in \cite{DumraufMonien:PLS_MCA}, \mcavar{3}{2}{3} and \mcavar{2}{3}{6} are \pls/-complete.
Therefore, we will treat the parameters $(p, q, r)$ as constant in the following, leaving of interest only the instance size parameters $m$ and $n$.
This also justifies the assumption that the every function $C_i$ can be given as a list of values for all $r^{p_i}$ possible input assignments.

\begin{theorem}
	Given a smoothed \mcapqr problem, where the cost values $C_i(\vecc v_i)$ for every integer assignment $\vecc v_i\in \intupto{r}^{p_i}$ are drawn from independent distributions $f_{i, \vecc v_i}: [0, 1]\map [0, \phi]$, the expected number of iterations of local search is bounded by $\landau{r^{2p(q+1)+2} nm^2 \phi}$.
\end{theorem}

\begin{proof}
	We first define a binary ($M=1$) CLO problem, where the cost part corresponds to all constraints and their possible assignments, i.e. for each $i \in \intupto{m}$ and every $\vecc v_i \in \intupto{r}^{p_i}$ we have a binary cost variable $s_{i, \vecc v_i}$ with associated cost coefficient $c_{i, \vecc v_i} = C_i(\vecc v_i)$.
	The model works on an indicator basis, i.e.\ for every assignment $a$ of \mcapqr, the cost variables of the induced CLO configuration $\vecc s(a)$ take value $s_{i, \vecc v_i} = 1$ only for $\vecc v_i = (a(x_{i_j}))_{j=1}^{p_i}$, and $s_{i, \vecc v} = 0$ for every other integer vector $\vecc v \neq \vecc v_i$, for all $i\in \intupto{m}$.
    Intuitively, the single entry of $1$ within every set of coordinates for fixed $i$ indicates which of the possible integer assignments is currently chosen, and thus this value contributes its weight to the total weight of the configuration.
    Therefore, we have in total $\nu = \sum_{i=1}^m r^{p_i}$ many cost components.
    Now it holds that for every integer assignment $a$ we have
    \begin{equation*}
        c(a) = \sum_{i=1}^m C_i(a(x_{i_1}),\dots, a(x_{i_{p_i}})) = \sum_{i=1}^m \sum_{\vecc v_i \in \intupto{r}^{p_i}} c_{i, \vecc v_i} \indicator{\vecc v_i = (a(x_{i_{1}}),\dots, a(x_{i_{p_i}}))} = C(\vecc s(a))
    \end{equation*}
    and the cost functions coincide.
    The neighbourhood relation for the CLO problem is derived directly from \mcapqr as well and is also efficiently searchable.

    We now translate the cost distributions of smoothed \mcapqr to the distributions of the induced smoothed CLO problem and apply \cref{thm:single-step-bound} to it.
    We specify the covering $(\mathcal{E}, \mathcal{I})$ and the properties $\lambda$, $\beta$ and $\mu$ in the following.
    To capture the indicator structure defined above, we define a coordinate cluster $I_i = \{(i, \vecc{v}_i)\fwh{\vecc{v}_i \in \intupto{r}^{p_i}}\}$ for all cost components associated with constraint $C_i$.
    There are in total $r^{p_i}$ different indicator assignments for the configuration of $\vecc s_{I_i}$.
    In the difference $\vecc s_{I_i}'-\vecc s_{I_i}$, both $\vecc s_{I_i}$ and $\vecc s_{I_i}'$ are one of these assignments, and we therefore have at most $ (r^{p_i})^2$ distinct differences by taking all possible combinations.
    Thus, we choose $\mu = r^{2p}$.

    Now, we specify the transition cover $\mathcal{E}$.
    We denote by $E(j, v, v')$ the transition cluster that captures all transitions where the assignment of variable $x_j$ changes from integer $v$ to $v'$.
    Formally, this is
		\begin{equation*}
			E(j, v, v') = \{(\vecc s(v, a_{-j}), \vecc s(v', a_{-j}))\fwh a_{-j}: \mathcal{X}\setminus \{x_j\} \map \intupto{r} \}
		\end{equation*}
    where $a_{-j}$ denotes an arbitrary assignment of the remaining variables to integers and $(v, a_{-j})$ denotes the assignment that maps $\mathcal{X}\setminus \{x_j\}$ as defined by $a_{-j}$ and $x_j$ to $v$.
    Now, $\mathcal{E} = \{E(j, v, v')\fwh j\in \intupto{n}, v, v'\in \intupto{r}, v\neq v'\}$ and therefore $\lambda = nr(r-1)$.

    Finally, in order to specify $\beta$, we investigate which coordinate clusters are affected by the transitions in a fixed cluster $E(j, v, v')$:
    In such transitions, the configuration of $\vecc s_{I_i}$ changes if and only if $x_j$ occurs in constraint $C_i$.
    By assumption, every variable appears in at most $q$ different constraints.
    This gives $\beta = q$ and concludes the proof with a bound on the expected number of steps of
        \begin{equation*}
            3\cdot (r^{2p})^q nr(r-1) \cdot \left(\sum_{i=1}^m r^{p_i}\right)^2 \cdot \phi
            \leq 3 (r^{2p})^q nr^2 \phi  m^2 r^{2p} 
            = 3 r^{2p(q+1)+2} nm^2 \phi. 
        \end{equation*}
\end{proof}

Note that a variant of \mcapqr can be defined which bridges the gap between its vast description of the constraints and the constraints in Max-CSP: instead of separate weights for every possible assignment of the variables in a constraint in \mcapqr or just a set of \emph{satisfying} assignments with a single weight, regardless of the specific assignment, we can use several (disjoint) sets of assignments with a given weight each for every constraint.
For the smoothed problem, each of these weights -- and therefore the contribution of several assignments in a given constraint -- is perturbed as a single value.
By doing so, many problems can be directly expressed by an instance of \mcapqr, which gives another unified proof method for smoothed complexity.
This model requires only slight adaptations from our analysis and also gives polynomial smoothed complexity under the variable occurrence constraints derived from \mcapqr.

\newpage
\printbibliography

\end{document}